\RequirePackage[l2tabu,orthodox]{nag}
\documentclass[11pt,letterpaper]{article} 

\usepackage{setspace}
    
\usepackage[dvipsnames]{xcolor}
\usepackage[notes=true,later=false,camera=true]{dtrt}
\usepackage[utf8]{inputenc}
\usepackage{etex}
\usepackage{ stmaryrd }
\usepackage{xspace,enumerate}
\usepackage[T1]{fontenc}
\usepackage[full]{textcomp}
\usepackage[american]{babel}
\usepackage{mathtools}
\usepackage{algorithm}
\usepackage[noend]{algpseudocode}
\usepackage{caption}
\usepackage{amssymb}

\usepackage{amsthm}
\usepackage{thmtools}
\usepackage{thm-restate}
\usepackage{empheq}

\usepackage{hyperref}
\hypersetup{hyperindex=true,pdfpagemode=UseOutlines,bookmarksnumbered=true,bookmarksopen=true,bookmarksopenlevel=2,pdfstartview=FitH,pdfborder={0 0 1},linkbordercolor=dt@linkcolor,citebordercolor=dt@linkcolor,urlbordercolor=dt@linkcolor}
\usepackage[capitalise,nameinlink]{cleveref}
\crefname{lemma}{Lemma}{Lemmas}
\crefname{fact}{Fact}{Facts}
\crefname{theorem}{Theorem}{Theorems}
\crefname{mtheorem}{Theorem}{Theorems}
\crefname{mcorollary}{Corollary}{Corollaries}
\crefname{mproposition}{Proposition}{Propositions}
\crefname{mquestion}{Question}{Questions}
\crefname{corollary}{Corollary}{Corollaries}
\crefname{claim}{Claim}{Claims}
\crefname{example}{Example}{Examples}
\crefname{algorithm}{Algorithm}{Algorithms}
\crefname{problem}{Problem}{Problems}
\crefname{definition}{Definition}{Definitions}
\crefformat{footnote}{#2\footnotemark[#1]#3}
\usepackage{paralist}
\usepackage{turnstile}
\usepackage{mdframed}
\usepackage{tikz}
\usepackage{caption}
\DeclareCaptionType{Algorithm}
\usepackage{newfloat}
\newtheorem{theorem}{Theorem}[section]
\newtheorem*{theorem*}{Theorem}

\newtheorem{proposition}[theorem]{Proposition}

\newtheorem*{proposition*}{Proposition}
\newtheorem{lemma}[theorem]{Lemma}
\newtheorem*{lemma*}{Lemma}

\newtheorem*{conjecture*}{Conjecture}

\newtheorem*{fact*}{Fact}

\newtheorem*{hypothesis*}{Hypothesis}

\theoremstyle{definition}
\newtheorem{definition}[theorem]{Definition}
\newtheorem*{definition*}{Definition}

\newtheorem{mquestion}{Question}

\theoremstyle{remark}
\newtheorem{claim}[theorem]{Claim}
\newtheorem*{claim*}{Claim}
\newtheorem{remark}[theorem]{Remark}
\newtheorem*{remark*}{Remark}

\newtheorem*{observation*}{Observation}

\usepackage[margin=1in]{geometry}
\usepackage{newpxtext} % T1, lining figures in math, osf in text
\usepackage{textcomp} % required for special glyphs

\usepackage[varg,bigdelims]{newpxmath}
\usepackage[scr=rsfso]{mathalfa}% \mathscr is fancier than \mathcal
\usepackage{bm} % load after all math to give access to bold math
\let\mathbb\varmathbb
\usepackage{microtype}

\let\svthefootnote\thefootnote
\newcommand\blfootnote[1]{%
  \let\thefootnote\relax%
  \footnotetext{#1}%
  \let\thefootnote\svthefootnote%
}

\allowdisplaybreaks
\newcommand{\FormatAuthor}[3]{
	\begin{tabular}{c}
		#1 \\ {\small\texttt{#2}} \\ {\small #3}
	\end{tabular}
}

\newcommand{\keywords}[1]{\bigskip\par\noindent{\footnotesize\textbf{Keywords\/}: #1}}

\newcommand{\cM}{\ensuremath{\mathcal{M}}}

\newcommand{\cG}{\ensuremath{\mathcal{G}}}

\newcommand{\cE}{\ensuremath{\mathcal{E}}}

\newcommand{\R}{{\mathbb R}}

\newcommand{\N}{\ensuremath{\mathbb{N}}}
\newcommand{\cY}{\ensuremath{\mathcal{Y}}}
\newcommand{\cZ}{\ensuremath{\mathcal{Z}}}

\newcommand{\PR}[1]{\mathrm{Pr}\left[ #1\right]}
\newcommand{\PROver}[2]{\mathrm{Pr}_{#1}\left[ #2\right]}

\newcommand{\EE}{\mathbb{E}}
\newcommand{\Eover}[2]{\mathbb{E}_{#1}\left[ #2 \right]}

\newcommand{\norm}[1]{\left\|#1\right\|}
\newcommand{\smnorm}[1]{\|#1\|}
\newcommand{\abs}[1]{\lvert #1 \rvert}
\newcommand{\twonorm}[1]{\left\|#1\right\|_2}

\newcommand{\infnorm}[1]{\left\|#1\right\|_\infty}

\newcommand{\inabs}[1]{\left|#1\right|}
\newcommand{\inabset}[1]{\inabs{\inset{#1}}}
\newcommand{\inset}[1]{\left\{#1\right\}}

\newcommand{\inparen}[1]{\left(#1\right)}
\newcommand{\inbrak}[1]{\left[#1\right]}

\newcommand{\supp}{\mathrm{supp}}

\newcommand{\polylog}{\mathrm{polylog}}
\newcommand{\poly}{\mathrm{poly}}

\newcommand{\tr}{\mathrm{tr}}

\newcommand{\nomadicv}[1]{V_{\mathrm{Nomadic}\inparen{#1}}}
\newcommand{\ehikes}[1]{\cE_{#1}}

\newcommand\numberthis{\addtocounter{equation}{1}\tag{\theequation}}

\newcommand{\eps}{\varepsilon}
\renewcommand{\epsilon}{\varepsilon}

\renewcommand{\exp}[1]{\mathrm{exp}\left(#1\right)}

\newcommand{\Id}{I}

\newcommand{\cH}{\mathcal H}
\newcommand{\ceil}[1]{{\left\lceil #1 \right\rceil}}
\newcommand{\floor}[1]{{\left\lfloor #1 \right\rfloor}}

\newcommand{\fresh}{\textsc{f}}
\newcommand{\bound}{\textsc{b}}
\newcommand{\stale}{\textsc{s}}
\newcommand{\highmult}{\textsc{h}}
\newcommand{\unforced}{\textsc{u}}
\newcommand{\forced}{\textsc{r}}
\newcommand{\Fits}{\{1,-1\}}
\newcommand{\sign}{\mathsf{sign}}

\newcommand{\Spec}{\mathrm{Spec}}
\newcommand{\defeq}{:=}

\newcommand{\ball}{{\mathsf B}}

\begin{document}
	%%%%%%%%%%%%%%%%%%%%%%%%%%%%%%%%%%%%%%%%%%%%%%%%%%%%%%%%%%%%%%%%%%%%%%%%%%%%%%%%
	%%%%%%%%%%%%%%%%%%%%%%%%%%%%%%%%%%%%%%%%%%%%%%%%%%%%%%%%%%%%%%%%%%%%%%%%%%%%%%%%
	%%%%%%%%%%%%%%%%%%%%%%%%%%%%%%%%%%%%%%%%%%%%%%%%%%%%%%%%%%%%%%%%%%%%%%%%%%%%%%%%

	%%%%%%%%%%%%%%%%%%%%%%%%%%%%%%%%%%%%%%%%%%%%%%%%%%%%%%%%%%%%%%%%%%%%%%%%%%%%%%%%
	\title{$\ell_p$-Spread and Restricted Isometry Properties of \\ Sparse Random Matrices}

\author{
\begin{tabular}[h!]{ccc}
      \FormatAuthor{Venkatesan Guruswami\thanks{Supported in part by NSF grants CCF-1908125 and CCF-2210823, and a Simons Investigator Award.}}{venkatg@berkeley.edu}{UC Berkeley}
      \FormatAuthor{Peter Manohar\thanks{Supported in part by an ARCS Scholarship, NSF Graduate Research Fellowship (under grant numbers DGE1745016 and DGE2140739), and NSF CCF-1814603.}}{pmanohar@cs.cmu.edu}{Carnegie Mellon University}
    \FormatAuthor{Jonathan Mosheiff\thanks{Supported in part by NSF CCF-1814603.}}{jmosheif@cs.cmu.edu}{Carnegie Mellon University}
\end{tabular}
} %
\date{}
%	\date{\texttt{\{venkatg,pmanohar,jmosheif\}@cs.cmu.edu} \\ Computer Science Department \\ Carnegie Mellon University \\ Pittsburgh, USA}
	%%%%%%%%%%%%%%%%%%%%%%%%%%%%%%%%%%%%%%%%%%%%%%%%%%%%%%%%%%%%%%%%%%%%%%%%%%%%%%%%

	%%%%%%%%%%%%%%%%%%%%%%%%%%%%%%%%%%%%%%%%%%%%%%%%%%%%%%%%%%%%%%%%%%%%%%%%%%%%%%%%
	\maketitle\blfootnote{Any opinions, findings, and conclusions or recommendations expressed in this material are those of the author(s) and do not necessarily reflect the views of the National Science Foundation.}
	\thispagestyle{empty}
	%%%%%%%%%%%%%%%%%%%%%%%%%%%%%%%%%%%%%%%%%%%%%%%%%%%%%%%%%%%%%%%%%%%%%%%%%%%%%%%%

	%%%%%%%%%%%%%%%%%%%%%%%%%%%%%%%%%%%%%%%%%%%%%%%%%%%%%%%%%%%%%%%%%%%%%%%%%%%%%%%%
	%%%%%%%%%%%%%%%%%%%%%%%%%%%%%%%%%%%%%%%%%%%%%%%%%%%%%%%%%%%%%%%%%%%%%%%%%%%%%%%%
	%%%%%%%%%%%%%%%%%%%%%%%%%%%%%%%%%%%%%%%%%%%%%%%%%%%%%%%%%%%%%%%%%%%%%%%%%%%%%%%%
	\begin{abstract}
	Random subspaces $X$ of $\R^n$ of dimension proportional to $n$ are, with high probability, well-spread with respect to the $\ell_2$-norm. Namely,
	every nonzero $x \in X$ is ``robustly non-sparse'' in the following sense: $x$ is $\eps \norm{x}_2$-far in $\ell_2$-distance from all $\delta n$-sparse vectors, for positive constants $\eps, \delta$ bounded away from $0$.
	This ``$\ell_2$-spread'' property is the natural counterpart, for subspaces over the reals, of the minimum distance of linear codes over finite fields, and corresponds to $X$ being a Euclidean section of the $\ell_1$ unit ball.  Explicit $\ell_2$-spread subspaces of dimension $\Omega(n)$, however, are unknown, and the best known explicit constructions (which achieve weaker spread properties), are analogs of \emph{low density parity check} (LDPC) codes over the reals, i.e., they are kernels of certain \emph{sparse} matrices.
	
	\smallskip
	 Motivated by this, we study the spread properties of the kernels of \emph{sparse random matrices.} We prove that with high probability such subspaces contain vectors $x$ that are $o(1)\cdot \norm{x}_2$-close to $o(n)$-sparse with respect to the $\ell_2$-norm, and in particular are \emph{not} $\ell_2$-spread. This is strikingly different from the case of random LDPC codes, whose distance is asymptotically almost as good as that of (dense) random linear codes.

	\smallskip
On the other hand, for $p < 2$ we prove that such subspaces \emph{are} $\ell_p$-spread with high probability. The spread property of sparse random matrices thus exhibits a threshold behavior at $p=2$. Our proof for $p < 2$ moreover shows that a random sparse matrix has the stronger restricted isometry property (RIP) with respect to the $\ell_p$ norm, and in fact this follows solely from the unique expansion of a random biregular graph, yielding a somewhat unexpected generalization of a similar result for the $\ell_1$ norm \cite{BerindeGI+08}. Instantiating this with suitable explicit expanders, we obtain the first explicit constructions of $\ell_p$-RIP matrices for $1 \leq p < p_0$, where $1 < p_0 < 2$ is an absolute constant.

	\end{abstract}		

	%%%%%%%%%%%%%%%%%%%%%%%%%%%%%%%%%%%%%%%%%%%%%%%%%%%%%%%%%%%%%%%%%%%%%%%%%%%%%%%%
	%%%%%%%%%%%%%%%%%%%%%%%%%%%%%%%%%%%%%%%%%%%%%%%%%%%%%%%%%%%%%%%%%%%%%%%%%%%%%%%%
	%%%%%%%%%%%%%%%%%%%%%%%%%%%%%%%%%%%%%%%%%%%%%%%%%%%%%%%%%%%%%%%%%%%%%%%%%%%%%%%%
	
	\clearpage
\setcounter{tocdepth}{2}
\begin{spacing}{0.9}
{\small \tableofcontents}
\end{spacing}
	\thispagestyle{empty}
	
	\clearpage
	
	\pagestyle{plain}
	\setcounter{page}{1}
		%%%%%%%%%%%%%%%%%%%%%%%%%%%%%%%%%%%%%%%%%%%%%%%%%%%%%%%%%%%%%%%%%%%%%%%%%%%%%%%%
	%%%%%%%%%%%%%%%%%%%%%%%%%%%%%%%%%%%%%%%%%%%%%%%%%%%%%%%%%%%%%%%%%%%%%%%%%%%%%%%%
	%%%%%%%%%%%%%%%%%%%%%%%%%%%%%%%%%%%%%%%%%%%%%%%%%%%%%%%%%%%%%%%%%%%%%%%%%%%%%%%%
	\section{Introduction}
	\label{sec:intro}

Classical results in asymptotic geometric analysis on the Gelfand/Kolmogorov widths of $\ell_2$ balls~\cite{FLM77,kashin,garnaev-gluskin} show that random subspaces $X$ of $\R^n$  of dimension proportional to $n$ (say, defined as the kernel of random $n/2 \times n$ matrices with i.i.d.\ Gaussian or $\pm 1$ entries) are \emph{good Euclidean sections} of $\ell_1^n$: namely, $\norm{x}_1 \ge \Omega(\sqrt{n}) \norm{x}_2$ for every $x \in X$. An elementary proof of this fact also follows from the Johnson-Lindenstrauss (JL) property of random matrices, its connection to the restricted isometry property (RIP) and compressed sensing, and their relationship to the Euclidean sections property~\cite{BDDw07}. 

The condition $\norm{x}_1 \ge \Omega(\sqrt{n}) \norm{x}_2$ can equivalently\footnote{\label{footnote:equivalence}See \cref{prop:companddist}.}  be expressed as a ``well-spreadness'' criterion satisfied by every nonzero vector $x \in X$: the largest $\delta n$ entries of $x$ have at most $1-\eps$ of its $\ell_2$ mass, for some positive constants $\delta,\eps$ bounded away from $0$ as $n \to \infty$. Equivalently, this means that all nonzero vectors $x \in X$ are \emph{incompressible}---there is no sparse vector that approximates $x$ well in $\ell_2$ norm (in other words, $\norm{x-y}_2 \ge \eps \norm{x}_2$ for all $\delta n$-sparse vectors $y$). This can be naturally viewed as a robust analog, for subspaces of $\R^n$, of the distance property of linear error-correcting codes.
		    
The above well-spreadness criterion can naturally be imposed with respect to any $\ell_p$ metric: a subspace $X$ is said to be \emph{$\ell_p$-spread} if every nonzero vector $x \in X$ is $\eps \norm{x}_p$-far in $\ell_p$-distance from all $\delta n$-sparse vectors. The $\ell_p$-spread property is a more stringent requirement for larger $p$ (see~\cref{prop:pspreadtoqspread}).
For $p > 2$, the optimal asymptotic dimension of $\ell_p$-spread subspaces is at most $O_p(n^{2/p})$ and thus $o(n)$~\cite{gluskin}. In this work, we therefore focus on $p \in [1,2]$ where it is possible to have $\ell_p$-spread subspaces of dimension proportional to $n$.

For a subspace $X$ of $\R^n$, define its \emph{$\ell_p$-distortion} $\Delta_p(X)$ to be the following quantity:
\begin{equation*}
\Delta_p(X) := \sup_{x \in X \setminus \{0^n\}} \frac{n^{1 - \frac{1}{p}}\norm{x}_p }{\norm{x}_1} \enspace.
\end{equation*}
Note that $1 \le \Delta_p(x) \le n^{1-1/p}$. Good $\ell_p$-spread of $X$ can be captured by the condition that $\Delta_p(X)$ is bounded by a fixed constant independent of $n$; this generalizes the aforementioned equivalence\cref{footnote:equivalence} of $\ell_2$-spread and the Euclidean section property. 
The term \emph{distortion} is used because the natural inclusion of $X$ in $\R^n$ induces a bi-Lipschitz embedding of $X$, taken with the $\ell_p$ norm, into $\ell_1^n$, with distortion $\Delta_p(X)$.
The distortion/spread property of subspaces with respect to different $\ell_p$ norms has been extensively studied, owing to its connections to width properties in convex geometry~\cite{gluskin,KT07}, embeddings between metric spaces~\cite{Indyk06}, compressed sensing~\cite{Donoho06,CandesRT06,KT07}, error-correction over the reals~\cite{CandesT05, GLW08}, and the restricted isometry (RIP) and dimensionality-reduction/Johnson-Lindenstrauss (JL) properties~\cite{KT07,BDDw07, ZhuGR15}.

 Despite a lot of interest and the abundance of probabilistic constructions, an outstanding question is to construct an \emph{explicit} subspace $X \subseteq \R^n$ of dimension $\Omega(n)$ that is $\ell_2$-spread, or equivalently has $\Delta_2(X) \le O(1)$. By explicit, we mean deterministically constructing a basis for the subspace (or its dual) in $\text{poly}(n)$ time.\footnote{Explicit constructions of $\ell_p$-spread spaces of dimension $\Omega(n)$ are given in \cite{BerindeGI+08} ($p = 1$) and \cite{Karnin11} ($1 \leq p < 2$).}
 This is a counterpart, for subspaces of $\R^n$, of the problem of constructing asymptotically good binary linear codes $C \subseteq \{0,1\}^n$: namely, codes whose dimension and minimum distance are both proportional to $n$. In addition to being a natural and basic challenge in pseudorandomness, explicit constructions are also valuable in applications of spread subspaces such as compressed sensing, as they provide a \emph{guarantee} that the matrix will have the stipulated properties. This is particularly important since there are no known methods to efficiently certify the $\ell_2$-spread (or even $\ell_p$-spread) of random subspaces.

 	%%%%%%%%%%%%%%%%%%%%%%%%%%%%%%%%%%%%%%%%%%%%%%%%%%%%%%%%%%%%%%%%%%%%%%%%%%%%%%%%
	%%%%%%%%%%%%%%%%%%%%%%%%%%%%%%%%%%%%%%%%%%%%%%%%%%%%%%%%%%%%%%%%%%%%%%%%%%%%%%%%
 \subsection{Kernels of sparse matrices} 
 
 In the case of $p = 2$, the best known explicit constructions of subspaces $X \subseteq \R^n$ with $\dim(X) \ge \Omega(n)$, in terms of their distortion $\Delta_2(X)$, are due to \cite{GLR10}. They give a construction analogous to Tanner codes from coding theory~\cite{tanner}, combining appropriately chosen unbalanced bipartite expanders and local subspaces, to produce $X$ with $\dim(X) \ge n-o(n)$ and $\Delta_2(X) \le (\log n)^{O(\log \log \log n)}$ (so almost poly-logarithmic).\footnote{For sublinear dimension, an explicit construction of $X \subseteq \R^n$ with distortion $\Delta_2(X) \le 1+o(1)$ and $\dim(X) \ge n/2^{O((\log \log n)^2)}$ was given in \cite{Indyk07}.}  A simpler construction analogous to Sipser-Spielman codes~\cite{SS96}, using $s$-regular spectral expanders and local well-spread subspaces of $\R^s$, was given in \cite{GLW08} and achieves\footnote{This construction is not explicit except for very small $s$, as the local subspace of $\R^s$ is either constructed by brute force or drawn at random.} $\Delta_2(X) \le n^{O(1/\log s)}$. An alternate probabilistic construction achieving similar parameters to \cite{GLW08} based on tensor products was given in \cite{IndykS10}. The approach of \cite{IndykS10} can further achieve distortion approaching $1$ at the expense of making $\dim(X)$ smaller, but still $\Omega(n)$.

One notable attribute of the constructions above is that the subspace $X$ can be expressed as the kernel of a matrix that is \emph{sparse}. For instance, the construction of \cite{GLW08} picks a matrix where each row is $s$-sparse with $\pm 1$ entries (that are chosen randomly for a probabilistic construction), and the construction in \cite{IndykS10} defines the subspace $X \subseteq \R^n$ as the $k$-fold tensor product of another subspace, and so $X$ can be defined as the kernel of an $n^{1/k}$-sparse matrix. Moreover, known explicit constructions of $\ell_p$-spread subspaces for $1 \leq p < 2$ \cite{BerindeGI+08,Karnin11} are also kernels of sparse matrices.

The sparsity of these constructions is inherited from the ``underlying constructions'' for codes; the constructions of \cite{GLR10, GLW08, IndykS10} come from ``lifting'' constructions of linear codes  (namely, Tanner codes~\cite{tanner}, Sipser-Spielman codes~\cite{SS96}, and tensor product codes, respectively) to this setting, and these constructions (for linear codes) are known to give good low density parity check (LDPC) codes: namely, codes that are the kernels of sparse matrices.

In light of these works, a natural question (and indeed one explicitly posed in \cite{GLW08}), is the following.
\begin{mquestion}
\label{q:sparsesections}
Does there exist an $m \times n$ matrix $A$ with $n - m\ge \Omega(n)$ whose rows are $s$-sparse for $s \le O(1)$ (or even $s \le \polylog(n)$) such that $\Delta_2(\text{ker}(A)) \le O(1)$?
\end{mquestion}
The approaches of \cite{GLW08,IndykS10} show that one can achieve $\Delta_2(\text{ker}(A)) \le \text{exp}(O(1/\delta))$ when $s=n^\delta$. A positive answer to \cref{q:sparsesections}, even via random matrices, would likely yield good progress towards explicit constructions, as $O(1)$-sparse matrices are likely easier to derandomize than dense random ones. A negative answer to \cref{q:sparsesections} would likely rule out explicit constructions based on the current state-of-the-art approaches of \cite{GLR10,GLW08,IndykS10}.

 In addition to exploring the potential of the approaches behind the current best constructions, sparsity is desirable from  a computational efficiency standpoint. Sparse matrices  lead to faster algorithms, for example when used as measurement matrices in compressed sensing or to compute a sparse JL transform for dimensionality-reduction.

Motivated by these considerations,  we study the $\ell_2$-spread, and, more generally, $\ell_p$-spread ($1\le p\le 2$) of subspaces defined as the kernel of \emph{sparse random matrices}. Such subspaces are the continuous analogues of random low density parity check (LDPC) codes.	Random LDPC codes have been studied in coding theory since Gallager's seminal work~\cite{gallager}, with a renaissance since the mid 1990s~\cite{RU-book} due to their fast iterative decoding algorithms and performance close to capacity.

Random LDPC codes are known to achieve rate vs.\ distance trade-offs approaching that of random (dense) linear codes~\cite{gallager}. Recently, even the list-decodability, and indeed any ``local'' property, of random LDPC codes was shown to be similar to that of random linear codes~\cite{MRRSW20}. Given that random subspaces are well-spread and that random LDPC codes achieve similar properties to random (dense) codes, one might naturally expect, by analogy, that the kernels of sparse random matrices are also well-spread.
		
\subsection{Our results}\label{sec:introResults}
Our results paint a precise picture of the $\ell_p$-spread of kernels $X$ of sparse random matrices. Before stating our results, we first define $\ell_p$-spread and state the random matrix model that we use.
			\begin{definition}[$\ell_p$-spread]
					\label{def:spread-intro}
		Fix $p \in [1, \infty]$, $\eps \in [0,1]$ and $k\le n\in \N$. A vector $y\in \R^n$ is \emph{$k$-sparse} if $\inabs{\supp(y)} \le k$. A vector $x\in \R^n \setminus \{0^n\}$ is said to be \emph{$(k,\eps)$-$\ell_p$-compressible} if there exists a $k$-sparse $y\in \R^n$ such that $\smnorm{x-y}_p\le \eps\smnorm{x}_p$. Otherwise, we say that $x$ is \emph{$(k,\eps)$-$\ell_p$-spread}. 
		
				A subspace $X \subseteq \R^n$ is \emph{$(k,\eps)$-$\ell_p$-spread} if every $x \in X \setminus \{0^n\}$ is $(k,\eps)$-$\ell_p$-spread.
	\end{definition}
		
		\parhead{The random matrix model.} 
		    A matrix $A \in \{0, 1, -1\}^{m \times n}$ is said to be \emph{$(s,t)$-biregular} if every row and column of $A$ has exactly $s$ and $t$ nonzero entries, respectively. Let $\cM_{m,n,s,t}$ denote the set of all $(s,t)$-biregular matrices in $\{0,1,-1\}^{m\times n}$.
		    
            All of our theorems for random matrices will be for a matrix $A$ drawn uniformly at random from $\cM_{m,n,s,t}$, where $\alpha = \frac{m}{n} = \frac{t}{s} \in (0,1)$ is a fixed constant and $n \to \infty$; for this exposition, we will use $A$ to denote a random matrix from $\cM_{m,n,s,t}$, and $B$ to denote an arbitrary matrix in $\{0,1,-1\}^{m \times n}$.            
            We additionally assume that $s := s(n) \leq n^c$ for some absolute constant $0 < c < 1$, and $t = \alpha s \geq 3$. An event $\cE$ is said to hold \emph{with high probability} if $\lim_{n\to \infty}\PR{\cE} = 1$. All asymptotic notation refers to the regime of $n\to \infty$ and constant $\alpha$. The constants implied by asymptotic notation are universal, unless stated otherwise. The symbols $c$, $c'$, $c_1$ and $c_2$ always stand for positive universal constants, which may differ across different lemma and theorem statements.
        We use the phrase ``in particular'' in theorem statements to refer to an implication that follows by either of the generic reductions of \cref{prop:companddist} ($\ell_p$-spread  implies $\ell_p$-distortion) or \cref{prop:riptospread} ($\ell_p$-RIP implies $\ell_p$-spread).
		
		\subsubsection{Poor $\ell_2$-spread of sparse random matrices}\label{sec:ell_2results}
		Our first theorem shows that, surprisingly, $\ker(A)$ is, with high probability, \emph{not} $\ell_2$-spread.			
	    \begin{restatable}[Poor $\ell_2$-spread of $\ker(A)$]{mtheorem}{elltwospreadneg}		\label{mthm:ell2spreadneg}
	    
		With high probability over $A$, there exists an $(m^{c},\frac{n^{-\Omega(\log(1/\alpha)/\log s)}}{1-\sqrt \alpha})$-$\ell_2$-compressible vector $x \in \ker(A)$, where $c < 1$ is an absolute constant. In particular, $$\Delta_{2}(\ker(A)) \geq (1-\sqrt\alpha)\cdot n^{\Omega(\log (1/\alpha)/\log s)} \enspace.$$
		Moreover, there is a $\poly(n)$-time algorithm that, on input $A$, outputs such an $x$. 
	    \end{restatable} 
	%The lower bound on $\Delta_2(\Ker(A))$ follows immediately from the existence of $x$ as above by virtue of \cref{prop:companddist}.
	Choosing $s = O(1)$ in \cref{mthm:ell2spreadneg} (and letting $\alpha$ be bounded away from $1$) implies\footnote{Note that since $\alpha s = t \geq 3$, we must have $\log s \geq \log \frac{1}{\alpha}$.} that $\Delta_2(\ker(A)) \geq n^{\Omega(1)}$ with high probability, and choosing $s = \polylog(n)$ implies $\Delta_2(\ker(A)) \geq n^{\Omega(\log (1/\alpha)/\log \log n)}$. We always trivially have $\Delta_2(\ker(A)) \leq \sqrt{n}$, so not only does \cref{mthm:ell2spreadneg} answer \cref{q:sparsesections} in the negative for sparse random matrices, but it also does so in a very strong sense. For instance, when $s = O(1)$, \cref{mthm:ell2spreadneg} shows that $\Delta_2(\ker(A))$ is ``maximally bad'', up to a constant factor in the exponent.
	
	Another point of interest is the choice $s = n^{\delta}$ for some fixed $\delta$. This yields the tradeoff of $\Delta_2(\ker(A)) \geq (\frac{1}{\alpha})^{\Omega(\frac{1}{\delta})}$, which precisely matches the tradeoff (in terms of $\delta$) achieved by both \cite{GLW08,IndykS10}. While our matrix ensemble is ``more random'' compared to those in \cite{GLW08, IndykS10}, \cref{mthm:ell2spreadneg} can nonetheless be interpreted as giving evidence that this $\exp{O(\frac{1}{\delta})}$ tradeoff from \cite{GLW08, IndykS10} is tight and inherent to sparse constructions.
	
	Our proof of \cref{mthm:ell2spreadneg} is \emph{constructive}, in the sense that we give a very simple, efficient algorithm to find such an $x \in \ker(A)$. This moreover shows that for sparse random matrices, one can efficiently \emph{refute} the claim that $\Delta_2(\ker(A)) = O(1)$, as the vector $x$ is a refutation witness. Our algorithm provides an interesting counterpoint to the work of \cite{BarakBHKSZ12}, who gave an algorithm based on the sum-of-squares SDP hierarchy to certify that $\Delta_2(\ker(A)) \leq O(1)$ with high probability for \emph{dense} matrices $A$ where $\dim(\ker(A)) \leq O(\sqrt{n})$. In contrast, our algorithm succeeds when $\dim(\ker(A)) = \Omega(n)$ and the matrix $A$ is \emph{sparse}. The two results taken together suggest an interesting relationship between the density and $\dim(\ker(A))$ of matrices $A$ for which we can efficiently certify or refute bounds on $\Delta_2(\ker(A))$.
	
	We also note that, by the well-known duality formula relating Kolmorogov and Gelfand widths (see [KT07] and the references therein), \cref{mthm:ell2spreadneg} implies that the row span of $A$ is far from approximating the $\ell_2$-sphere in $\ell_\infty$ distance. Concretely, with high probability over $A$ there exists $x\in \R^n$ with $\twonorm x=1$ that is $(1-\sqrt\alpha)\cdot n^{\Omega(\log (1/\alpha)/\log s)}/\sqrt{n}$-far in $\ell_\infty$ norm from all vectors of the form $A^{\top} y$, where $y\in \R^m$.

	The proof of \cref{mthm:ell2spreadneg} requires the following strong bound that we show on the singular values of $A$. 
	\begin{restatable}[Singular value bound]{mtheorem}{singvalue}
		\label{mthm:singvalue}
	    With high probability, the set of singular values $\sigma(A)$ of $A$ satisfy
		\begin{equation*}
			\sigma(A)\subseteq \inbrak{\sqrt{s-1} - (1+o(1))\cdot\sqrt{t-1}, \sqrt{s-1} + (1+o(1))\cdot\sqrt{t-1}} \enspace.
		\end{equation*}
		Moreover, the above bound holds even without our (otherwise global) assumption that $s \leq n^c$ for some absolute constant $c < 1$.
	\end{restatable}

	\cref{mthm:singvalue} should not be surprising, especially given the recent works of \cite{BritoDH18, Bordenave19, BordenaveC19, MohantyOP20a, MohantyOP20b, ODonnellW20, Zhu20}, and indeed our proof follows the same overall blueprint of these works. Most of these papers, however, only handle the case when the degree of the graph is \emph{constant} as $n \to \infty$;\footnote{The exceptions are \cite{MohantyOP20a}, which handles $\polylog(n)$ degree, and \cite{Zhu20}, which handles $n^c$ degree but does not obtain as sharp bounds.} this corresponds to the case of $s = O(1)$ in \cref{mthm:singvalue}. \cref{mthm:singvalue} thus differs as it allows for $s = \omega(1)$, and indeed we can even take $s = n^c$ for some absolute constant $c < 1$. On the other hand, most of the aforementioned works deal with the case of \emph{unsigned} adjacency matrices, whereas we only prove \cref{mthm:singvalue} for \emph{randomly signed} adjacency matrices. We note that proving the analogue of \cref{mthm:singvalue} for unsigned adjacency matrices  (where $\sigma(A)$ now denotes the set of singular values, excluding the trivial value of $\sqrt{st}$) and for all (non-constant) $s,t$ remains open.
	
	The singular value bound in \cref{mthm:singvalue} is challenging to prove because it is so sharp. Indeed, it is not too difficult to show that $\sigma(A) \subseteq[\sqrt{s} - O(\sqrt{t}), \sqrt{s} + O(\sqrt{t})]$ with high probability via black-box applications of known results, e.g., \cite{BandeiraV16}. However, this does not suffice for our use in the proof of \cref{mthm:ell2spreadneg}, as the aforementioned weaker bound would only suffice to prove \cref{mthm:ell2spreadneg} provided that $\alpha \leq c$ for some absolute constant $c$, where $c$ depends on the absolute constant $c'$ hidden in the ``$O(\sqrt{t})$''. We need the sharp bound of \cref{mthm:singvalue} in order to allow for $\alpha$ to be an arbitrary constant in $(0,1)$.
	
	As a counterpart to \cref{mthm:ell2spreadneg}, we give the following partial converse, which shows that $\ker(A)$ \emph{is} $(k,\eps)$-$\ell_2$-spread for a weak choice of parameters $k$ and $\eps$.
	
	\begin{restatable}[Converse to \cref{mthm:ell2spreadneg}]{mtheorem}{elltwospreadpos}
	\label{mprop:ell2spreadpos}
	Assume that $t\ge 9$. Then, with high probability over $A$, the space $\ker(A)$ is $\inparen{\Omega(\alpha^2 n/t^4), \alpha^{O(\log n/\log t)}}$-$\ell_2$-spread.
	\end{restatable}
	We note that in \cref{mprop:ell2spreadpos}, the parameter $k$ is $\Omega(\alpha^2 n/t^4) = m^{\Omega(1)}$ and the parameter $\eps$ is $\alpha^{O(\log n/\log s)}$.\footnote{This follows since $t = \alpha s \leq s$, $m = \alpha n$, and $s \leq n^c$ for some absolute constant $c$.} \cref{mprop:ell2spreadpos} thus shows that the parameters in  \cref{mthm:ell2spreadneg} are tight up to the universal constants in the exponent.
	 Our proof of \cref{mprop:ell2spreadpos} is an adaptation of the proof of \cite[Lemma 3.4]{BasakR17}.
	
	\subsubsection{$\ell_p$-RIP and $\ell_p$-spread for $p < 2$}
	We next focus on the $\ell_p$ norm for $p < 2$. For $p < 2$, there are known explicit constructions of $\ell_p$-spread subspaces \cite{Karnin11}. Because of this, for $p < 2$ we focus on the stronger,\footnote{\label{fn:riptospread}See \cref{prop:riptospread}.} well-studied \emph{Restricted Isometry Property (RIP)}. We also note that the constructions of \cite{Karnin11} are highly structured, and so even though they also come from sparse matrices, they do not tell us anything about the $\ell_p$-spread of sparse \emph{random} matrices. 
	
	We prove that sparse random matrices are not only $\ell_p$-spread, but are also $\ell_p$-RIP, and, moreover, this follows merely from the expansion of the underlying bipartite graph of the random matrix $A$. In particular, we prove that \emph{any} signed adjacency matrix $B$ of a left-regular bipartite expander graph $G$ is $\ell_p$-RIP, provided that the maximum right degree $s_{\max}$ is above a small threshold independent of $n$.

    The RIP is a well-studied property of matrices from the compressed sensing literature, defined as follows.
	\begin{definition}[$\ell_p$-RIP]\label{def:RIP}
	Let $B\in \R^{m \times n}$ be a matrix. We say that $B$ is \emph{$(k,\eps)$-$\ell_p$-RIP} if there exists $K > 0$ such that for every $k$-sparse $x \in \R^n$, it holds that\footnote{We note that the standard definition of RIP typically appears without the normalization factor $K$ above. We include the parameter $K$ for convenience, as the random sparse matrices we consider are not normalized.}
	\begin{equation*}
	K(1 - \eps) \norm{x}_p \leq \norm{Bx}_p \leq K (1 + \eps) \norm{x}_p \enspace.
	\end{equation*}
	\end{definition}
	We note that $\ell_p$-RIP implies $\ell_p$-spread,\cref{fn:riptospread} and in fact it is a strictly stronger property~\cite{KT07}.
	
	RIP matrices have been studied extensively in the context of compressed sensing, as they yield a polynomial-time algorithm based on linear programming for the \emph{robust sparse recovery problem}. Namely, given a ``noisy measurement sketch'' $y = Bx + e$ of a vector $x$, where $B$ is $(k, \eps)$-$\ell_p$-RIP and $\norm{e}_p \leq \eta$, there is a polynomial-time algorithm to recover an estimate $\hat{x}$ for $x$ with the so-called ``$\ell_p$-$\ell_1$ guarantee,'' namely the estimate $\hat{x}$ satisfies $\norm{\hat{x} - x}_p \leq O\inparen{k^{-(1 - \frac{1}{p})}\norm{x - x^*}_1 + \eta}$, where $x^*$ is a $k$-sparse vector minimizing $\norm{x - x^*}_1$ (see Appendix~A in \cite{ZhuGR15} for details). We note that if $B$ is merely $\ell_p$-spread, then $B$ suffices for the (non-robust) sparse recovery problem, i.e., when there is no noise $e$.

    We now turn to formally stating our results. We first recall the definition of a (unique) bipartite expander.

	\begin{definition}[Unique expanders]
	A bipartite graph $G = (V_L = [n], V_R = [m], E)$ is a \emph{$t$-left-regular $(\gamma, \mu)$-unique expander} if \begin{inparaenum}[(1)] \item $\deg(u) = t$ for all $u \in V_L$, and \item for all $S \subseteq V_L$, $\abs{S} \leq \gamma n$, there are at least $t(1 - \mu)\abs{S}$ vertices $v \in V_R$ which each have exactly one neighbor in $S$\end{inparaenum}.
	\end{definition}
	
	A matrix $B \in \{0, 1, -1\}^{m \times n}$ is a \emph{signed adjacency matrix} of a bipartite graph $G=(V_L = [n], V_R = [m], E)$ if 
	$$B_{r,u}\ne 0 \iff (u,r) \in E$$
	for all $u\in V_L$, $r\in V_R$.
	
	\begin{restatable}[$\ell_p$-RIP of expander graphs]{mtheorem}{ExpansionToRIP}\label{mthm:ExpansionToRIP}
        Let $G$ be a bipartite $t$-left-regular $(\gamma,\mu)$-unique expander with maximum right degree $s_{\max}$, and let $B$ be any signed adjacency matrix of $G$. Let $0 < \eps \leq 1$ and $1\le p < 2$ such that $\eps^2 \ge 9\mu s_{\max}^{p-1}$. Then, $B$ is $(\gamma n, \eps)$-$\ell_p$-RIP, i.e., for every $\gamma n$-sparse $x\in \R^n$,
        $$t^{\frac{1}{p}}(1-\eps) \norm{x}_p \le \norm{Bx}_p \le t^{\frac{1}{p}}(1+\eps)\norm{x}_p \enspace.$$
	\end{restatable}
     \cref{mthm:ExpansionToRIP} generalizes a result of \cite{BerindeGI+08}, which shows that any signed adjacency matrix $B$ of $G$ is $\ell_1$-RIP, provided that $G$ is an expander. This is somewhat surprising, as the proof in \cite{BerindeGI+08} makes heavy use of properties specific to the $\ell_1$ norm.\footnote{They also show that their proof for $\ell_1$-RIP extends to $\ell_p$-RIP for $p \leq 1 + O(\frac1{\log n})$, because the ``H\"{o}lder factor'' of $n^{1 - \frac{1}{p}}$ is $O(1)$, but it does not extend to $\ell_p$ for any constant $p > 1$.}
    
    The $\ell_p$-RIP of matrices for general $p$ has been studied in other contexts, most notably in \cite{ZhuGR15}. As is typical when studying RIP matrices, they view the sparsity parameter $k$ as a fixed function of $n$, and determine $m$ as a function of $k,n$. However, the results in \cite{ZhuGR15} are incomparable to ours, as they hold only for the low-sparsity case of $k = O(n^{1/p})$ (so $k = o(n)$ if $p > 1$), but we are concerned with the case of $k = \Omega(n)$, when the sparsity is a small constant fraction of $n$.    
	
	 As a random $t$-left-regular bipartite graph is a good expander with high probability, we obtain the following corollary of \cref{mthm:ExpansionToRIP}, which shows that $\ker(A)$ for $A \gets \cM_{m,n,s,t}$ achieves very good $\ell_p$-spread for every $p \in [1,2)$. Thus, the poor $\ell_2$-spread of $\ker(A)$ is in fact specific to the case of $p = 2$.
		\begin{restatable}[Good $\ell_p$-RIP and $\ell_p$-spread of $A$]{mcorollary}{ellpspreadpos}
		\label{mthm:ellpspreadpos}
		Fix $p \in [1, 2)$, $0 < \eps < \frac 12$, and suppose that $s \ge \inparen{\frac{18}{\alpha \eps^2}}^{\frac 1{2-p}}$.
		Then, with high probability over $A$, the matrix $A$ is $\inparen{\Omega(\gamma n),\eps}$-$\ell_p$-RIP for $\gamma = \frac{\alpha^2}{t^4}$: for every $\Omega(\gamma n)$-sparse $x \in \R^n$, it holds that
		\begin{equation*}
		t^{\frac{1}{p}}(1 - \eps) \norm{x}_p \leq \norm{Ax}_p \leq t^{\frac{1}{p}}(1 + \eps) \norm{x}_p \enspace.
		\end{equation*}
		In particular, the subspace $\ker(A)$ is $\inparen{\Omega(\gamma n), \Omega\inparen{\gamma^{1 - \frac{1}{p}}}}$-$\ell_p$-spread and $\Delta_p(\ker(A)) \leq O\inparen{1/\gamma^{2 - \frac{2}{p}}}$.
	\end{restatable}
    Fixing $p, \alpha,\eps$ to be constants and taking $s$ to be a large enough constant, this shows that $\ker(A)$ is $(\Omega(n), \Omega(1))$-$\ell_p$-spread with high probability, and therefore $\Delta_p(\ker(A)) = O(1)$.   
    Together with \cref{mthm:ell2spreadneg}, this shows that the $\ell_p$-spread property of $\ker(A)$ exhibits an interesting \emph{threshold phenomenon} at $p = 2$.

   We also combine \cref{mthm:ExpansionToRIP} with the explicit constructions of expander graphs of \cite{CapalboRVW02} to obtain the following corollary, which gives an explicit construction of $\ell_p$-RIP matrices for all $p \in [1,p_0)$, where $1 < p_0 < 2$ is an absolute constant. We thus obtain the first explicit construction of a matrix $B$ achieving the ``$\ell_p$/$\ell_1$ guarantee'' for the robust sparse recovery problem, and our matrices are for the regime $k = \Theta(n)$ and any $p \in [1, p_0)$.
   Previously, such constructions were only known for $p\le 1 + O\inparen{\frac 1{\log n}}$ \cite{BerindeGI+08}. Unlike \cref{mthm:ellpspreadpos}, our explicit constructions only extend up to some threshold $p_0 < 2$. This is because the expanders of \cite{CapalboRVW02} achieve weaker expansion than random graphs. Concretely, the ``expansion error'' $\mu$ of the \cite{CapalboRVW02} expanders is $\mu = O(1/t)^{\tau}$ for some constant $\tau < 1$, which yields the threshold of $p_0 = 1 + \tau$, whereas random graphs achieve $\mu = O(1/t)$, allowing for $p_0 = 2$.

        \begin{restatable}[Explicit construction of $\ell_p$-RIP matrices]{mcorollary}{ExplicitConstruction}\label{mthm:ExplicitConstruction}
        Let $0 < \eps < \frac 12$, $\alpha \in (0,1)$, and let $n \in \N$ be sufficiently large.
    For some universal constant $1 < p_0 < 2$, there exists a deterministic algorithm which, given $p\in [1,p_0)$, $\eps$, $\alpha$ and $n$, outputs in time $\poly(n/\delta) + 2^{O(1/\delta)}$ a matrix $B\in \{0,1\}^{m\times n}$, for some $m\le \alpha n$, such that 
 $B$ is $\inparen{\gamma n, \eps}$-$\ell_p$-RIP, for some $\delta, \gamma = \poly(\eps, \alpha)^{\frac{1}{p_0 - p}}$. In particular, $\ker(B)$ is $(\gamma n, \gamma^{1 - \frac{1}{p}})$-$\ell_p$-spread and $\Delta_p(\ker(B)) \leq 1/\gamma^{2 - \frac{2}{p}}$.
	\end{restatable}
 	Note that as $\eps$, $\alpha$ and $p$ are constants, the matrix $B$ in \cref{mthm:ExplicitConstruction} is $(\Omega(n), O(1))$-$\ell_p$-RIP, $\ker(B)$ is $(\Omega(n), \Omega(1))$-$\ell_p$-spread, and $\Delta_p(\ker(B)) \leq O(1)$.
	
	As noted earlier, \cite{Karnin11} gives explicit constructions of $(\Omega(n), \Omega(1))$-$\ell_p$-spread subspaces, for \emph{all} $1 \leq p < 2$. This is incomparable to \cref{mthm:ExplicitConstruction}: on one hand, \cite{Karnin11} obtains the full range of $1 \leq p < 2$, but on the other hand, his matrices only are $\ell_p$-spread and do not satisfy the (strictly stronger) $\ell_p$-RIP. 
	Our construction is moreover the ``simplest'' black-box reduction to expansion: we show that the mere adjacency matrix of a bipartite expander is $\ell_p$-RIP. While the constructions in \cite{Karnin11} are themselves not too complicated, we think that this is nonetheless an interesting conceptual contribution of our work. 

    \medskip
    Finally, we also prove the following partial converse to \cref{mthm:ellpspreadpos}, which shows that when $s^{2-p} \lessapprox \frac{1}{\alpha}$ (i.e., $s^{2-p}$ is a constant factor below the threshold in \cref{mthm:ellpspreadpos}), then $A$ is \emph{not} $\ell_p$-RIP.
    
	\begin{restatable}[Partial converse to \cref{mthm:ellpspreadpos}]{mtheorem}{ellpspreadneg}
	\label{mprop:ellpspreadneg} 
	Let $p\in [1,2)$, $\eps > 0$. If $s-1 \le\inparen{\frac{1}{(1+\eps)\alpha}}^{\frac 1{2-p}}$, then with high probability over $A$, there exists an $n^c$-sparse vector $x \in \R^n\setminus\{0^n\}$ such that $$\frac{\norm{Ax}_p}{\norm{x}_p} \leq t^{\frac{1}{p}} \cdot m^{-\Omega\inparen{\frac{\eps}{\log s}}}\enspace.$$
	\end{restatable}
	Note that $\norm{Ae_1}_p/\norm{e_1}_p = t^{\frac{1}{p}}$ always holds, so \cref{mthm:ellpspreadpos} demonstrates that, given small enough $s$, the ratio $\frac{\norm{Ax}_p}{\norm{x}_p}$ has a large range over different choices of $n^c$-sparse $x$.

	%%%%%%%%%%%%%%%%%%%%%%%%%%%%%%%%%%%%%%%%%%%%%%%%%%%%%%%%%%%%%%%%%%%%%%%%%%%%%%%%
	%%%%%%%%%%%%%%%%%%%%%%%%%%%%%%%%%%%%%%%%%%%%%%%%%%%%%%%%%%%%%%%%%%%%%%%%%%%%%%%%
	%%%%%%%%%%%%%%%%%%%%%%%%%%%%%%%%%%%%%%%%%%%%%%%%%%%%%%%%%%%%%%%%%%%%%%%%%%%%%%%%
	\section{Proof overview}
	\label{sec:techs}
	We outline the proofs of our results. For the purposes of this exposition, we will adopt the same convention as in \cref{sec:introResults} and use $A$ and $B$ to denote a uniformly sampled matrix from $\cM_{m,n,s,t}$ and arbitrary matrix from $\{0, 1, -1\}^{m \times n}$, respectively. Recall that $\cM_{m,n,s,t}$ denotes the set of $(s,t)$-biregular matrices with entries in $\{0,1,-1\}$, and that $\frac mn = \frac ts = \alpha$ for some constant $\alpha$, and $n\to \infty$. For simplicity of this exposition, in this section we restrict ourselves to the regime $s = O(1)$ unless stated otherwise. 
	
	We naturally associate with $B$ the bipartite graph $G = G_B = (V_L,V_R,E)$ with $n = |V_L|$ left vertices, $m = |V_R|$ right vertices, and an edge between $u\in V_L$ and $r\in V_R$ if $B_{r,u}\ne 0$. We view the rows and columns of $B$ as indexed by $V_R$ and $V_L$, respectively, and identify $\R^n$ with $\R^{V_L}$, and $\R^m$ with $\R^{V_R}$. In addition, we define the function $\sign = \sign_B :E\to \{1,-1\}$, which maps an edge $\{u,r\}$ as above to $B_{r,u}$. We note that the combination of $G_B$ and $\sign_B$ completely describes $B$. 
	
    %%%%%%%%%%%%%%%%%%%%%%%%%%%%%%%%%%%%%%%%%%%%%%%%%%%%%%%%%%%%%%%%%%%%%%%%%%%%%%%%
	%%%%%%%%%%%%%%%%%%%%%%%%%%%%%%%%%%%%%%%%%%%%%%%%%%%%%%%%%%%%%%%%%%%%%%%%%%%%%%%%
	\subsection{\cref{mthm:ell2spreadneg}: $\ker(A)$ in not $\ell_2$-spread}\label{sec:negtechinques}
	For simplicity, we only sketch here why $\ker(A)$ is likely to contain an $(o(n),o(1))$-compressible vector, and leave the more refined parameter setting for the actual proof of \cref{mthm:ell2spreadneg} in \cref{sec:negative}.
	
	The proof of \cref{mthm:ell2spreadneg} consists of two steps.
	In the first step we find an $o(n)$-sparse vector $x\in \R^n$ with $\twonorm x \ge 1$ and $\twonorm{Ax} \le o(1)$ (\cref{lem:treevector}). In the second step we find a vector $y\in \ker(A)$ with $\twonorm{y-x} \le o(1)\cdot \twonorm y$ (\cref{lem:singvalrounding,mthm:singvalue}). In particular, $y$ is $\inparen{o(n),o(1)}$-compressible, so $\ker(A)$ cannot be $\ell_2$-spread. 
	
    Below, we outline these two steps. It is straightforward to see, given the construction described below, that both $x$ and $y$ can be computed in polynomial time given $A$.
    
    We also note that an $\ell_p$ analog of Step 1 is the main technical component in the proof of \cref{mprop:ellpspreadneg}, and is also proven in \cref{sec:negative}.
	
	\parhead{Step 1: constructing a sparse $x$ with small $\twonorm {Ax}$}
	To obtain the vector $x$, we first prove (\cref{prop:goodGraph}) that $G$ is highly likely to contain a vertex $v^*\in V_L$ such that the ball of radius $2\ell+1$ about $v^*$, for some $\ell \le O(\log n)$, contains no cycles. That is, the radius-$\ell$ neighborhood of $v^*$ is a complete $(t,s)$-biregular tree $T$ rooted at $v^*$. Recall that a rooted tree is $(t,s)$-biregular if the even depth (resp.\ odd depth) inner vertices have degree $t$ (resp.\ s). The existence of such a vertex $v^*$ is the only random property of $A$ needed in this step of the proof. In particular, assuming that $G$ has the aforementioned property, our construction of $x$ is always possible, regardless of the $\sign$ function.
	
	To describe the construction of $x$ itself, we assume for simplicity that $\sign(e) = 1$ for all $e\in E$. Namely, all the non-zero entries of $A$ are $1$. In this setting, let $v\in V_L\cap T$ be a vertex of depth $2k$ in the tree for some $k\ge 0$ (note that a vertex in $V_L$ must have even depth), and set $x_v = (-(s-1))^{-k}$. For any $x\in V_L\setminus T$, set $x_v = 0$. Note that $\supp (x) \subseteq T$. We choose $\ell$ above to be as large as possible, i.e., $O(\log n)$, so that the size of $T$ is roughly $n^c$ for some $c < 1$. In particular, $x$ is $\approx n^c$-sparse. Also, note that $\norm{x}_2^2 \ge x_{v^*}^2 = 1$.  We informally refer to the vector $x$ produced by this construction as a \emph{tree vector}.
	
	Our construction guarantees that $\inparen{Ax}_r = 0$ for every internal node $r\in T\cap V_R$. Indeed, suppose that $r$ is of depth $2k+1$. Then, it has one neighbor of depth $2k$, and $s-1$ neighbors of depth $2k+2$. As $\inparen{Ax}_r$ is the sum of $x_v$ over neighbors $v$ of $r$, we have $\inparen{Ax}_r = (-(s-1))^{-k} + (s-1)\cdot(-(s-1))^{-(k+1)} = 0$. 
	
	To compute $\twonorm{Ax}$, it thus suffices to compute $\inabs{(Ax)_r}$ when $r$ is one of the $t(t-1)^{\ell}(s-1)^\ell$ leaves of $T$. It is not hard to see that in this case $\inabs{\inparen{Ax}_r} = (s-1)^{\ell}$, and so $$\twonorm{Ax}^2 = t(t-1)^{\ell}(s-1)^\ell\cdot (s-1)^{-2\ell} = e^{-\Omega(\ell)} = o(1)\enspace.$$
	
	We note that our tree vector construction is similar in spirit to a  construction by Noga Alon \cite[Theorem 8]{GLW08}, which demonstrates the limitations of expander-based analysis of the spread property. In \cite{GLW08}, however, they \emph{choose} their graph $G$ so that (their analog of) the tree vector $x$ will lie in (their analog of) $\ker(A)$ \emph{by design}.
 Our graph is random and not up to our choice, so we cannot simply orchestrate the graph so that our tree vector $x$ to belong to $\ker(A)$.
This necessitates that we perform the nontrivial step of rounding $x$ to some $y \in \ker(A)$, which we discuss next.
	
	\parhead{Step 2: finding $y\in \ker(A)$ close to $x$}
	Our main goal in this step is to establish the following lemma:
	\begin{lemma}[Informal]\label{lem:informalRounding}
        	With high probability over $A$, it holds that every $x\in \R^n$ with $\twonorm{Ax}\le o(1)$ is $o(1)$-close to some vector $y\in \ker(A)$.
	\end{lemma}
	Indeed, let $x$ be the tree vector constructed in Step 1. Then $x$ is $o(n)$-sparse with $\twonorm{x} \ge 1$ and $\twonorm{Ax} \le o(1)$. By \cref{lem:informalRounding}, there exists a vector $y\in \ker(A)$, which is $o(1)$-close to $x$. This vector $y$ is $\inparen{o(n),o(1)}$-compressible, which yields \cref{mthm:ell2spreadneg} in the present parameter setting.
	
    One may naively try to prove \cref{lem:informalRounding} by locally perturbing $x$ to try to make $Ax = 0^m$, e.g.\ by designing a greedy algorithm for this task. This approach, however, seems difficult to execute, especially given that \cref{lem:informalRounding} is in fact not true in general. For example, it could be the case that $x$ is a (unit norm) right singular vector of $A$ with singular value $o(1)$. Then, $\norm{Ax}_2 = o(1)$, but $\norm{x - y}_2 \geq 1$ for all $y \in \ker(A)$, and in fact the closest vector in $\ker(A)$ to $x$ is $0^n$. 
    
    Instead, we set $y$ to be the orthogonal projection of $x$ onto $\ker(A)$. In hindsight, this is the obvious choice for $y$, as then $y \in \ker(A)$ is the vector that minimizes $\norm{x - y}_2$. How large can $\norm{x - y}_2$ be? Intuitively, we would like to say that $\twonorm{Ax}$ being small implies that $\twonorm{x-y}$ is small as well. As the earlier example shows, this is not true for a general matrix $A$, as $A$ could have small singular values. However, the implication \emph{does} hold provided that all singular values of $A$ are $\Omega(1)$.\footnote{Technically, what matters is the minimum \emph{nonzero} singular value. However, with high probability the matrix $A$ will be full rank (i.e., rank $m$), so that $\sigma_{\min}(A) > 0$. Indeed, this is trivially implied by \cref{mthm:singvalue}.} Indeed, the singular value decomposition of $A$  implies that
	\begin{equation*}
	\norm{Ax}_2 = \norm{A(x - y)}_2 \geq \sigma_{\min}(A) \norm{x - y}_2 \enspace,
	\end{equation*}
	where $\sigma_{\min}(A)$ is the minimum singular value of $A$ and the inequality holds as $x-y$ is orthogonal to $\ker(A)$. Hence, $\twonorm{x-y} \le \frac{\twonorm{Ax}}{\sigma_{\min}(A)}$.
	
	The main technical component of Step 2 is therefore the lower bound on $\sigma_{\min}(A)$, given by \cref{mthm:singvalue}. As we have argued above, the crude lower bound of $\sigma_{\min}(A) \geq \Omega(1)$ suffices to yield \cref{lem:informalRounding}. Indeed, if $\sigma_{\min}(A) \geq \Omega(1)$, then $\norm{x - y}_2 \leq o(1)$, and so $\ker(A)$ contains an $\inparen{o(n),o(1)}$-compressible vector. The precise high-probability lower bound on $\sigma_{\min}(A)$ established in \cref{mthm:singvalue} implies a finer quantitative version of \cref{lem:informalRounding}, which yields the full \cref{mthm:ell2spreadneg}. The latter gives a much sharper bound on the $o(1)$ term, and also applies to sparsity all the way up to $O\inparen{n^c}$ for some $c > 0$.
	
    We remark that one can easily show that $\sigma_{\min}(A) \geq \sqrt{s} - O(\sqrt{t})$ via ``off-the-shelf'' methods, such as \cite{BandeiraV16}. However, this would only allow us to prove \cref{mthm:ell2spreadneg} provided that $\alpha \leq c$ for some absolute constant $c < 1$ (related to the $O(1)$ factor in front of $\sqrt{t}$ above), and thus would not allow us to take $\alpha$ to be \emph{any} constant in $(0,1)$, e.g., $\alpha = 0.999$. Our sharper bound also highlights the difficulty in lower bounding the minimum singular value when $\alpha = m/n$ is close to $1$.
	
	We postpone our discussion of the proof of \cref{mthm:singvalue} to \cref{sec:techSingularValues}, and turn next to our positive result for $\ell_p$-spread for $p < 2$.
	
	%%%%%%%%%%%%%%%%%%%%%%%%%%%%%%%%%%%%%%%%%%%%%%%%%%%%%%%%%%%%%%%%%%%%%%%%%%%%%%%%
	%%%%%%%%%%%%%%%%%%%%%%%%%%%%%%%%%%%%%%%%%%%%%%%%%%%%%%%%%%%%%%%%%%%%%%%%%%%%%%%%
	\subsection{\cref{mthm:ExpansionToRIP}: $\ell_p$-RIP for $p < 2$ from vertex expansion}
	We sketch the proof of \cref{mthm:ExpansionToRIP}.
    For simplicity, we will assume that $B \in \cM_{m,n,s,t}$, i.e., that the bipartite graph $G_B$ is $t$-left-regular and $s$-right-regular (and hence $s_{\max} = s$) and also that $G_B$ is a $(\gamma, \mu)$-unique expander. For this exposition, we only discuss the claimed lower bound on $\norm{Bx}_p$ stated in the theorem, namely,
    \begin{equation}
    \label{eq:RIPlb}
    \norm{Bx}_p \geq t^{\frac{1}{p}}(1 - \eps) \norm{x}_p
    \end{equation}
    for all $\gamma n$-sparse $x \in \R^n$, as the upper bound is obtained via a variation on the same method. 
	
	\parhead{\cref{mthm:ExpansionToRIP} for tree vectors}As a warm-up for the proof of \cref{mthm:ExpansionToRIP}, we show why the $o(n)$-sparse tree vector $x$ constructed in \cref{sec:negtechinques} does not yield a counterexample to \cref{eq:RIPlb}. Let $x^{(k)}$ ($0\le k\le \ell$) denote the restriction of $x$ to the vertices in the $2k$-th level of the tree $T$. Then, $$\norm{x^{(k)}}_p^p = t(t-1)^{k-1}(s-1)^{(1-p)k}\enspace .$$ For $p = 2$, this expression decreases exponentially in $k$, and thus the $\ell_2$-mass of $x$ is concentrated at the top of the tree. For $p < 2$, however, $\norm{x^{(k)}}_p^p$ actually \emph{grows} exponentially in $k$ provided that $s$ is large enough (concretely, one needs $s^{2-p}\gtrapprox \frac 1\alpha$).\footnote{And, indeed, if instead $s^{2-p} \lessapprox \frac{1}{\alpha}$, then we have $\norm{Bx}_p = o(1)$, and this gives us \cref{mprop:ellpspreadneg}.} In this case, the $\ell_p$-mass is concentrated towards the bottom of the tree. Moreover, one can take $s$ large enough so that all but an $\frac \eps2$-fraction of the mass lies in the bottom layer. Then,
    $$\frac{\norm{Bx}_p^p}{\norm x_p^p} \ge\inparen{1-\frac \eps2}\cdot \frac{\norm{Bx}_p^p}{\norm {x^{(\ell)}}_p^p} = \inparen{1-\frac\eps2}\cdot\frac{t(t-1)^\ell(s-1)^{(1-p)\ell}}{t(t-1)^{\ell-1}(s-1)^{(1-p)\ell}} = \inparen{1-\frac \eps 2}\cdot (t-1) \ge (1-\eps)t\enspace.$$
    Hence, $x$ is not a counterexample to \cref{eq:RIPlb}.	
	
	\parhead{\cref{mthm:ExpansionToRIP} for general vectors with tree-shaped support}
    Fix a set $S\subseteq V_L$ such that the subgraph induced by $T:= S\cup N(S)$ is a $(t,s)$-biregular tree. We generalize the above discussion of tree vectors by explaining why \cref{eq:RIPlb} holds for any vector $x\in \R^n$ supported on $S$. 
    
	Given $r \in N(S)$, let $v_r$ denote the parent of $r$ in the tree $T$. In an overly optimistic scenario, if we could show that $\abs{(Bx)_r} \approx \abs{v_r}$ for all $r\in N(S)$, then we would be done, as each vertex $v\in S$ has $t-1$ children.\footnote{Except for the root, which has $t$ children.} Each of these children then contributes $\approx \abs{x_v}^p$ mass to $\norm{Bx}_p^p$, so that $\norm{Bx}_p^p \approx (t-1)\cdot\sum_{v\in S}\inabs{x_v}^p = (t-1)\cdot\norm{x}_p^p$, implying \cref{eq:RIPlb}.
    As the tree vector case shows, one cannot, in fact, hope to guarantee $\abs{(Bx)_r} \approx \abs{v_r}$ for all $r \in N(S)$. Indeed, for a tree vector $x$ we have $(Bx)_r=0$ for any non-leaf $r\in N(S)$. Thus, a more delicate analysis is required. 
	
	For intuition, let us consider the viewpoint of an adversary seeking to construct an $x$ supported on $S$ such that $\norm{Bx}_p^p$ is small. We shall think of the adversary as assigning values to $\{x_v\}_{v \in S}$ starting from the root, and then moving down the tree.
	
	For each non-leaf $r \in N(S)$, let $W_r$ denote the set of $s-1$ children of $r$. Recall that $v_r$ is the parent of $r$. Note that $\inabs{(Bx)_r} \ge \inabs{x_{v_r}} - \sum_{u\in W_r}\inabs{x_u}$
	 due to the triangle inequality. Hence,	when assigning values to the vertices in $W_r$, the adversary morally has two choices: \begin{inparaenum}[(1)]
	 \item either make $\sum_{u \in W_r} \abs{x_u} \ll \abs{x_{v_r}}$, in which case $\abs{(Bx)_r} \approx \abs{x_{v_r}}$, or
	\item make $\abs{(Bx)_r} \approx 0$, in which case $\sum_{u \in W_r} \abs{x_u} \approx \abs{x_{v_r}}$.\footnote{Note that the adversary has the third choice of setting $\sum_{u \in W_r} \abs{x_u} \gg \abs{x_{v_r}}$, but this is worse for the adversary.}
	\end{inparaenum}
	Let us fix some $\beta < 1$, and suppose that the adversary chooses values for $\{x_u\}_{u \in W_r}$ such that $\sum_{u \in W_r} \abs{x_u} \leq \beta \abs{x_v}$, i.e., the adversary chooses Case (1).
	We then have that
	\begin{equation*}
	\abs{(Bx)_r} \geq \Big(\abs{x_{v_r}} - \sum_{u \in W_r} \abs{x_u}\Big) \geq (1 - \beta) \abs{x_{v_r}} \enspace,
	\end{equation*}
	so $\abs{(Bx)_r} \approx \abs{x_{v_r}}$, which is what we wanted. Next, suppose that the adversary makes $\sum_{u \in W_r} \abs{x_u} \geq \beta \abs{x_v}$, i.e., the adversary chooses Case (2). Then, $\abs{(Bx)_r}$ can be small, but applying H\"{o}lder's inequality, we have
	\begin{equation*}
	\sum_{u \in W_r} \abs{x_u}^p \geq \frac{\beta^p}{(s-1)^{p-1}} \cdot \abs{x_v}^p \enspace.
	\end{equation*}
	Now, suppose that all children $r$ of $v_r$ have this property. Then, the total $\ell_p^p$ mass of all of the grandchildren of $v_r$ must be at least $\frac{(t-1)\beta^p}{(s-1)^{p-1}} \cdot \abs{x_v}^p \gg \abs{x_v}^p$. We thus see that, intuitively, the adversary has merely pushed its task down to the grandchildren of $v_r$, and in doing so has not made any progress towards its overall goal.	
	Indeed, this is precisely what happens in the case of a tree vector!
	
    The above informal argument shows that the adversary does not ``win'' in either case.
	We can concretely capture this intuition via the following potential function:
	\begin{equation*}
	a_r(x) \defeq \abs{(Bx)_r}^p + \Theta(1)\cdot  \frac{(s-1)^{p-1}}{\eps^{p-1}}\sum_{u \in W_r} \abs{x_u}^p \enspace.
	\end{equation*}
	In the actual proof, this choice of $a_r(x)$ allows us to cleanly express the intuition that either $\abs{(Bx)_r}$ is large or $\sum_{u \in W_r} \abs{x_u}$ is large, and further extends beyond the ``toy case'' of tree-supported vectors.
	
    \parhead{Using expansion when $S\cup N(S)$ is not a tree.} We now turn to the general case, where the subgraph induced by $S\cup N(S)$ (where $S = \supp(x)$) is not necessarily a tree. We observe that above, we are only using the tree structure to show that the rooted tree $S \cup N(S)$ trivially has a $1$-to-$(t-1)$ ``matching'' with the following properties:
    \begin{inparaenum}[(1)]\item every vertex $r \in N(S)$ is matched with exactly one vertex $v \in S$, and \item every vertex $v \in S$ is matched with at least $t-1$ vertices in $N(S)$\end{inparaenum}. Indeed, when $S\cup N(S)$ is a tree, such a matching exists by matching each vertex $r\in N(S)$ with its parent $v_r$.
    
    To generalize the above, we use the (unique) expansion of $G$ to construct a similar matching that suffices for the proof.
    Recall that $G$ is a $(\gamma, \mu)$-unique expander, meaning that every set $S \subseteq V_L$ of size $\leq \gamma n$ has at least $t(1 - \mu)\abs{S}$ unique neighbors, i.e., neighbors of a unique element of $S$.
    We construct the matching by ``peeling off'' vertices one at a time from $S$, each time matching a vertex with $\geq t(1 - \mu)$ vertices in $N(S)$, namely its neighbors that are not neighbors of any of the remaining ``unpeeled'' vertices in $S$.
    
    The above step can be viewed as extracting a ``tree-like'' subgraph from $S \cup N(S)$, where each vertex $v \in S$ has at least $t(1 - \mu)$ ``children'' (the vertices it was matched with), and at most $\mu t$ ``parents'' (its neighbors that it was \emph{not} matched with). Each vertex $r \in N(S)$ still has exactly one ``parent'' and $\leq s - 1$ ``children''. Once we have the above ``tree-like'' subgraph, the argument for trees goes through with only minor modifications, so this finishes the proof.

    We note that the existence of this ``tree-like'' subgraph for any set $S$ with $\abs{S} \leq \gamma n$ immediately implies that $G$ is a $(\gamma, \mu)$-vertex expander, and hence a $(\gamma, 2 \mu)$-unique expander. Thus, the existence of such a subgraph for every $S$ of size at most $\gamma n$ is equivalent to unique expansion, up to a factor of $2$ loss in the parameter $\mu$.

	\parhead{Comparison with \cite{BerindeGI+08}.} We briefly summarize the proof in \cite{BerindeGI+08} for the case of $p = 1$, and explain why their proof does not extend to the case of $p > 1$. 
	
	The proof in \cite{BerindeGI+08} proceeds as follows. For a vector $x$ supported on $S$, let $E_0$ denote the set of edges between $S$ and $N(S)$. First, they match each $r \in N(S)$ to its neighbor $v \in S$ with $\abs{x_v}$ maximized. Let $E_1$ be the set of edges in this matching, and let $E_2 = E_0 \setminus E_1$. For any $r \in N(S)$, it then follows that $\abs{(Bx)_r} \geq \abs{x_v} - \sum_{u \in W_r} \abs{x_u}$, where $(v,r) \in E_1$ and $W_r = \{u : (u,r) \in E_2\}$. Hence, $\norm{Bx}_1 \geq \sum_{(v,r) \in E_1} \abs{x_v} - \sum_{(u,r) \in E_2} \abs{x_u}$. We observe that this step of the proof is specific to the $\ell_1$ norm, and does not generalize to larger $\ell_p$ norms.
	
	The main step in the proof is to argue that $\sum_{ (v,r) \in E_2} \abs{x_v} \leq  t \eps\norm{x}_1$ using expansion. With this in hand, it immediately follows that $\sum_{(v,r) \in E_1} \abs{x_v} \geq t(1 -  \eps)\norm{x}_1$, because $\sum_{(v,r) \in E_0} \abs{x_v} = t \norm{x}_1$ by regularity.
	It then follows that $\norm{Bx}_1 \geq t(1 - 2 \eps) \norm{x}_1$. Note that the upper bound $\norm{Bx}_1 \leq t \norm{x}_1$ is trivial, so this shows that $B$ is $\ell_1$-RIP.
	
    One may attempt to generalize this proof to $p > 1$ by replacing $\abs{x_v}$ with $\abs{x_v}^p$. For example, using expansion it follows that $\sum_{(v,r) \in E_2} \abs{x_v}^p \leq  t \eps\norm{x}_p^p$, and as $\sum_{(v,r) \in E_0} \abs{x_v}^p = t \norm{x}_p^p$, we then have $\sum_{(v,r) \in E_1} \abs{x_v}^p \geq t(1 -  \eps)\norm{x}_p^p$. But this is not enough to complete the proof, as it does \emph{not} follow that $\norm{Bx}_p^p \geq \sum_{(v,r) \in E_1} \abs{x_v}^p - \sum_{(v,r) \in E_2} \abs{x_v}^p$. Indeed, this is a fundamental barrier, and is the reason why our analysis for $p \geq 1$ proceeds by analyzing the ``local'' potential function $a_r(x)$, rather than the two ``global'' sums over $E_1$ and $E_2$ above.

		%%%%%%%%%%%%%%%%%%%%%%%%%%%%%%%%%%%%%%%%%%%%%%%%%%%%%%%%%%%%%%%%%%%%%%%%%%%%%%%%
	%%%%%%%%%%%%%%%%%%%%%%%%%%%%%%%%%%%%%%%%%%%%%%%%%%%%%%%%%%%%%%%%%%%%%%%%%%%%%%%%
	\subsection{\cref{mthm:singvalue}: bounds on the singular values of $A$}\label{sec:techSingularValues}
	We give a brief overview of the proof of \cref{mthm:singvalue}. First, we observe that in order to bound the singular values of $A$, it suffices to bound the spectrum of $M := A A^{\top} - s \cdot \Id$, as each singular value of $A$ is the square root of an eigenvalue of $A A^{\top}$. Note that $M$ is a square matrix with an all-$0$ diagonal, by regularity of $A$.
	
	\parhead{Step 1: reducing to the nomadic walk matrix via a modified Ihara--Bass formula.}
	The first step in the proof is to relate bounds on the spectrum of $M$ to the spectral radius (i.e., maximum eigenvalue in absolute value) $\rho(B)$ of $B$, the \emph{nomadic walk matrix} introduced in \cite{MohantyOP20b}.\footnote{We note that one could most likely also prove \cref{mthm:singvalue} using the standard nonbacktracking walk matrix and Ihara--Bass formula, e.g., with similar methods as in \cite{BritoDH18}.} The nomadic walk matrix $B$ is indexed by pairs of edges\footnote{In \cite{MohantyOP20b}, the nomadic walk matrix is indexed by \emph{directed} edges. In our context, this is equivalent to a length~$2$ oriented walk $e_1 \to e_2$, which is equivalent to a pair $(e_1, e_2)$ of \emph{undirected} edges, as the ordering in the pair gives the unique orientation $e_1 \to e_2$ in the walk.} $(e_1, e_2)$ in $G$ that form a length $2$ non-backtracking walk in $G$, and its $((e_1,e_2),(e'_1,e'_2))$-th entry is $\sign(e'_1)\sign(e'_2)$ if $e_1 \to e_2 \to e'_1 \to e'_2$ forms a non-backtracking walk of length $4$ in $G$, and is $0$ otherwise. Note that $B$ is \emph{not} symmetric.
	
	\begin{theorem}[Modified Ihara--Bass formula, Theorem 3.1 of \cite{MohantyOP20b}, informal]
	If $\rho(B) \leq (1 + o(1)) \sqrt{(s-1)(t-1)}$, then the spectrum $\Spec(M)$ of $M$ satisfies:
	\begin{equation*}
	\Spec(M) \subseteq [t - 2 - 2(1+o(1)) \sqrt{(s-1)(t-1)}, t - 2 + 2(1 + o(1))\sqrt{(s-1)(t-1)}]\enspace.
	\end{equation*}
	\end{theorem}
	The above theorem thus shows that it suffices to prove that $\rho(B) \leq (1+o(1)) \sqrt{(s-1)(t-1)}$ with high probability.
	
	We remark that bounds on the spectra of matrices of the form of $M$ were studied in \cite{MohantyOP20b} for the case of $s = O(1)$. Unfortunately, this is insufficient to prove \cref{mthm:singvalue}, as we wish to allow $s$ to be any function of $n$ (provided that $s \leq n^c$ for some absolute constant $c$). However, \cite[Theorem 3.1]{MohantyOP20b} is a general statement that holds regardless of $s$, so we can make use of it in our setting.
	
	\parhead{Step 2: bounding $\rho(B)$ via the trace method by counting hikes.}
	The natural approach to bound $\rho(B)$ is by applying the trace method. As the matrix $B$ is not symmetric, we compute:
	\begin{equation*}
	T := \EE_{\sign}[\tr(B^{\ell} (B^{\top})^{\ell})] \enspace,
	\end{equation*}
	where the expectation is taken over the function $\sign$ that determines the signs of the entries of $A$. By carefully expanding this expectation, one can show that the nonzero contributions to $T$ roughly come from length $4(\ell-1)$ closed walks in $G$ where \begin{inparaenum}[(1)]
	\item each edge in the walk appears an even number of times, and
	\item the walk is non-backtracking, except possibly at the middle step in the walk
	\end{inparaenum}.
	Such walks (of length $4\ell)$ are commonly referred to as $(2\ell)$-\emph{hikes} \cite{MohantyOP20a}.
	
	To finish the proof of \cref{mthm:singvalue}, we thus turn to obtaining a careful bound on the number of such  walks.

	Counting these walks requires extra care in our setting as our graph is bipartite, and so the bound needs to be sensitive to the difference in right/left degree.
	The counting of such hikes also differs greatly depending on whether $s \leq \polylog(n)$ or $s = \omega(\polylog(n))$.
	
	\parhead{Step 3: counting the number of hikes when $s \leq \polylog(n)$.}
	
	\cite[Section 3]{MohantyOP20a} counts the number of such hikes when $s = O(1)$, provided that $G$ is bicycle-free at radius $O(\log n)$. Namely, any vertex $v$ participates in at most one cycle of length $O(\log n)$. 
	By repeating their proof, one can show that their bounds can be extended to the case when $s \leq \polylog(n)$. However, we still cannot use their bound on the number of such hikes naively, as their counting is for non-bipartite graphs and thus yields a bound of $m(1 + o(1))^{\ell}(s-1)^{2\ell} $, simply because it treats left and right vertices the same, and the maximum degree of a vertex is $s$. We refine their approach to ensure that right and left vertices contribute roughly equally, which will yield the desired bound. One may, at first glance, be tempted to assume that this is trivial because a closed walk in a bipartite graph has an equal number of left and right vertices, but this is not the case, as we shall see.
	
	We adopt the bookkeeping approach of \cite{MohantyOP20a}. We think of a hike as discovering the graph $G$ as one traverses the hike. A step in a hike is \emph{fresh} if uses an edge for the first time and ends at a previously undiscovered vertex; it is \emph{boundary} if it uses an edge for the first time but ends up at an old vertex; finally, it is \emph{stale} if it uses an old edge.
	
	Because each edge must appear an even number of times, a hike can have at most $2\ell$ fresh steps. Each fresh step ``pays'' a factor of $(s-1)$ (if we move from a right to a left vertex) or $(t-1)$ (if we move from a left to a right vertex) in our bound in the number of hikes, as this is the number of choices for the next vertex that the hike moves to.
	\cite[Theorem 2.13]{MohantyOP20a} implies that, since $G$ is bicycle-free, the number of boundary steps is $\ll \ell$; they also show that one can bound the ``number of choices'' for the stale steps by some small factor, which we ignore here.
	
	We need to augment the argument of \cite{MohantyOP20a} with the following addition: if a hike has $c$ fresh steps, then the number of fresh steps $c_R$ that start at a right vertex is $\approx \frac{c}{2}$, and similarly for $c_L$ for left vertices. Note that by definition, $c = c_R + c_L$.	
	The key observation is that a fresh step from a right (resp.\ left) vertex must be followed by either a fresh step from a left (resp.\ right) vertex, or by a boundary step. Indeed, after we take a fresh step we are at a previously unvisited vertex, so the next step must use a new edge; in particular, it cannot be \emph{stale}. This implies that the deviation of each of $c_L,c_R$ from $\frac{c}{2}$ is bounded by the number of boundary steps, which is $\ll \ell$.
	
	This implies a bound of $m (1 + o(1))^{\ell}(s-1)^{\ell}(t-1)^{\ell}$ on the number of hikes, provided that $s \leq \polylog(n)$. The $m$ comes from the number of start vertices in the hike, and the $(1+o(1))^{\ell}$ comes from the stale and boundary steps, as well as our new deviation term. This yields the desired bound for sparse $s$.
	
	\parhead{Step 4: counting the number of hikes when $s = \omega(\polylog(n))$.}
	For $s$ this large, the graph $G$ is ``dense'', and so it will not be bicycle-free at radius $O(\log n)$. This rules out the approach of \cite{MohantyOP20a}, which relies on $G$ being bicycle-free. 
	Instead, we adapt a standard counting approach (for bounding the operator norm of a random $n \times n$ Gaussian matrix) given in \cite[Section 2.3.6]{Tao12} to our bipartite setting. Our crucial observation here is to note that any hike can have at most $\ell$ distinct left vertices. As we pay a factor of $(s-1)$ every time we move to a left vertex, it then follows that the ``power'' of $(s-1)$ in our bound can only be at most $\ell$. The standard counting argument of \cite[Section 2.3.6]{Tao12} for Gaussian matrices then goes through, yielding the desired bound.
	
		%%%%%%%%%%%%%%%%%%%%%%%%%%%%%%%%%%%%%%%%%%%%%%%%%%%%%%%%%%%%%%%%%%%%%%%%%%%%%%%%
	%%%%%%%%%%%%%%%%%%%%%%%%%%%%%%%%%%%%%%%%%%%%%%%%%%%%%%%%%%%%%%%%%%%%%%%%%%%%%%%%
	\subsection{Organization}
	The rest of the paper is organized as follows. \cref{sec:prelims} introduces notation and definitions that we use in our proofs. In \cref{sec:RandomGraphProperties}, we state and prove several claims about  properties of a certain random bipartite graph naturally associated with the random matrix $A$. In \cref{sec:negative}, we prove our negative results, namely, \cref{mthm:ell2spreadneg,mprop:ellpspreadneg}. In \cref{sec:ellpspreadpos} we prove \cref{mthm:ExpansionToRIP,mthm:ellpspreadpos,mthm:ExplicitConstruction}, our positive results about $\ell_p$ for $p<2$. Our proofs in \cref{sec:negative} rely on the singular value bounds stated in \cref{mthm:singvalue}, which we prove in \cref{sec:singvalue}. Finally, in \cref{sec:ell2spreadpos} we prove \cref{mprop:ell2spreadpos}, our positive result for $\ell_2$-spread.

	%%%%%%%%%%%%%%%%%%%%%%%%%%%%%%%%%%%%%%%%%%%%%%%%%%%%%%%%%%%%%%%%%%%%%%%%%%%%%%%%
	%%%%%%%%%%%%%%%%%%%%%%%%%%%%%%%%%%%%%%%%%%%%%%%%%%%%%%%%%%%%%%%%%%%%%%%%%%%%%%%%
	%%%%%%%%%%%%%%%%%%%%%%%%%%%%%%%%%%%%%%%%%%%%%%%%%%%%%%%%%%%%%%%%%%%%%%%%%%%%%%%%
	\section{Preliminaries}
	\label{sec:prelims}
	
	%%%%%%%%%%%%%%%%%%%%%%%%%%%%%%%%%%%%%%%%%%%%%%%%%%%%%%%%%%%%%%%%%%%%%%%%%%%%%%%%
	%%%%%%%%%%%%%%%%%%%%%%%%%%%%%%%%%%%%%%%%%%%%%%%%%%%%%%%%%%%%%%%%%%%%%%%%%%%%%%%%
	\subsection{Basic notation}
	For an integer $k \in \N$, let $[k] \defeq \{1, \dots, k\}$.
	
	For a vector $x \in \R^n$, we let $\supp(x) := \{i \in [n] : x_i \ne 0\}$. We say $x$ is $k$-sparse if $|\supp(x)| \le k$ and $x \ne 0^n$.
	
	Given a graph $G$, a vertex $u$, and $\ell\in \N$, we let $\ball_G(u,\ell)$ denote the induced subgraph on vertices of distance at most $\ell$ from $u$. For a set of vertices $S$, we let $\ball_G(S,\ell) \defeq \bigcup_{u\in S}\ball_G(S,\ell)$.
	
	Let $\sigma(A)$ denote the set of singular values of a matrix $A$. We also denote $\sigma_{\min}(A) = \min{\sigma(A)}$ and $\sigma_{\max} = \max{\sigma(A)}$. If $A$ is square, we denote the set of its eigenvalues by $\Spec(A)$, and its spectral radius by $\rho(A) \defeq \max_{\lambda \in \Spec(A)} \abs{\lambda}$. We remark that unlike in \cref{sec:introResults,sec:techs}, we will now use $A$ to denote arbitrary $m \times n$ matrices, not just uniformly random ones from $\cM_{m,n,s,t}$.
	
	%%%%%%%%%%%%%%%%%%%%%%%%%%%%%%%%%%%%%%%%%%%%%%%%%%%%%%%%%%%%%%%%%%%%%%%%%%%%%%%%
	%%%%%%%%%%%%%%%%%%%%%%%%%%%%%%%%%%%%%%%%%%%%%%%%%%%%%%%%%%%%%%%%%%%%%%%%%%%%%%%%
	\subsection{Biregular bipartite graphs and sparse matrices}\label{sec:SparseMatrices}

	Let $G = (V_L, V_R, E)$ denote a bipartite graph with $n = \abs{V_L}$ left vertices and $m = \abs{V_R}$ right vertices. We say that $G$ is \emph{$t$-left-regular} if $\deg(u)=t$ for all $u\in V_L$. If, additionally, $\deg(v)=s$ for all $v\in V_R$, we say that $G$ is \emph{$(t,s)$-biregular}.
	
	For a vertex $v$, we let $N_G(v)$ denote the set of neighbors of $v$ in $G$. For a set of vertices $S$, we let $N_G(S) := \cup_{v \in S} N_G(v)$ denote the set of neighbors of $S$ in $G$, and we let $U_G(S)$ denote the set of unique neighbors of $S$, i.e., $U_G(S)$ is the set of $v \in N_G(S)$ such that $v \in N_G(u)$ for some $u \in S$ and $v \notin N(u')$ for all other $u' \in S$ with $u \ne u'$. We omit the subscript $G$ when it is clear from the context.
    For $S\subseteq V_L$ and $T\subseteq V_R$, we let $E(S,T) = E_G(S,T) = \inset{\inset{u,v}\mid u\in S\wedge v\in T}$.
	
	Let $\cG_{m,n,s,t}$ denote the set of all $(t,s)$-biregular graphs with $\abs{V_L} = n$ and $\abs{V_R} = m$. Namely, $\cG_{m,n,s,t} = \inset{G_A \mid A\in \cM_{m,n,s,t}}$ (see the beginning of \cref{sec:techs} for a reminder of the definition of $G_A$).
	We note that sampling a uniformly random element of $\cM_{m, n,s,t}$ is equivalent to sampling a uniformly random graph $G_A$ from $\cG_{m,n,s,t}$ and then sampling a uniformly random edge signing $\sign_A \colon E \to \{\pm 1\}$.
	
    We also define the following properties.
	\begin{definition}\label{def:bicycleFree}
    A graph $G$ is said to be \emph{bicycle-free} at radius $\ell \in \N$ if for all $v\in V$, the ball $\ball_G(v,\ell)$ contains at most one cycle.
    \end{definition}
    
    \begin{definition}[Expansion and unique expansion]\label{def:uniqueExpansion}
        Let $G$ be a $t$-left-regular graph. Fix $\gamma, \mu\in (0,1)$. We say that $G$ is:
        \begin{enumerate}[(1)]
        \item a \emph{$(\gamma, \mu)$-vertex expander} if for all $S \subseteq V_L$ with $\abs{S} \leq \gamma n$, it holds that $\abs{N(S)} \geq t(1 - \mu) \abs{S}$, 
        \item a \emph{$(\gamma, \mu)$-unique expander} if for all $S \subseteq V_L$ with $\abs{S} \leq \gamma n$, it holds that $\abs{U(S)} \geq t(1 - \mu) \abs{S}$.
        \end{enumerate}
	\end{definition}
	
	We note the following proposition, which shows that expansion implies unique expansion, up to a factor $2$ in $\mu$.
	\begin{proposition}[Expansion implies unique expansion]
	\label{prop:exptouniqueexp}
	Let $G$ be a $t$-left-regular $(\gamma, \mu)$-vertex expander. Then, $G$ is a $(\gamma, 2\mu)$-unique expander.
	\end{proposition}
	\begin{proof}
	Fix $S \subseteq V_L$ with $\abs{S} \leq \gamma n$. Let $T = N(S)$, $T_1 = U(S)$, and $T_2 = T \setminus T_1$. We have $\abs{T_1} + 2 \abs{T_2} \leq \abs{E(S,T)}$, as each $v \in T$ must have at least two edges to $S$, and $\abs{E(S,T)} \leq t \abs{S}$, as $G$ is $t$-left-regular. As $\abs{T_1} + \abs{T_2} = \abs{T}$, we thus have $2\abs{T} - \abs{T_1} \leq t \abs{S}$. As $\abs{T} \geq t(1 - \mu) \abs{S}$, it follows that $\abs{T_1} \geq t(1 - 2\mu) \abs{S}$, which finishes the proof.
	\end{proof}
	
	 We recall the formal guarantees of the explicit expander construction of \cite{CapalboRVW02}.
     \begin{theorem}[{\cite[Theorem 7.1]{CapalboRVW02}}, rephrased]\label{thm:CRVW}
	        For some universal $c \geq 1$, there is a deterministic algorithm which, given $\beta \in (0,1)$ an inverse power of $2$, $\mu > 0$ and $n\in \N$ a sufficiently large power of $2$, outputs in  time $\poly\inparen{n, \frac1\beta, \frac1\mu} + \poly\inparen{2^{\frac{1}{\beta}}, 2^{\frac{1}{\mu}}}$ 
	        a bipartite graph $G=(V_L,V_R,E)$ with $\abs{V_L} = n$, $\abs{V_R} = m := \beta n$, such that $G$ is a $t$-left-regular $(\Omega(\frac{\mu m}{t}), \mu)$-unique expander, for $t \leq O\Big(\big(\frac{1}{\beta\mu}\big)^c\Big)$.
	    \end{theorem}
	 We note that \cite{CapalboRVW02} only constructs vertex expanders, not unique expanders. However, as \cref{prop:exptouniqueexp} shows, these are equivalent up to a factor of $2$ in the parameter $\mu$, so this does not matter.

	%%%%%%%%%%%%%%%%%%%%%%%%%%%%%%%%%%%%%%%%%%%%%%%%%%%%%%%%%%%%%%%%%%%%%%%%%%%%%%%%
	%%%%%%%%%%%%%%%%%%%%%%%%%%%%%%%%%%%%%%%%%%%%%%%%%%%%%%%%%%%%%%%%%%%%%%%%%%%%%%%%
	\subsection{$\ell_p$ norms, compressibility, spread subspaces, and distortion}
    Let $p \in \R$ satisfy $1 \leq p \leq \infty$. Given a matrix $A\in \R^{m\times n}$ we denote its $\ell_p\to\ell_p$ operator norm by
    $$\norm{A}_p = \sup_{x\in \R^n\setminus \inset 0}\frac{\norm{Ax}_p}{\norm x_p} \enspace.$$
    We recall the following two facts about $\ell_p$ norms.
    
    \begin{lemma}[H\"{o}lder's inequality]
	\label{lem:holder}
    Let $x, y, z \in \R^n$ be vectors, with $z_i = x_i y_i$ for all $i \in [n]$. Then $\norm{z}_1 \leq \norm{x}_p \smnorm{y}_{\frac{p}{p - 1}}$.

    Furthermore,  let $1 \leq q < p$. Then $\norm{x}_p \leq \norm{x}_q \leq \norm{x}_p n^{1/q - 1/p}$ for any $x \in \R^n$.
    \end{lemma}

        Next, we recall the definition of compressible and spread vectors, as well as spread subspaces.
        
			\begin{definition}[$\ell_p$-spread]
		Fix $p \in [1, \infty]$, $\eps \in [0,1]$ and $k\le n\in \N$. A vector $y\in \R^n$ is \emph{$k$-sparse} if $\inabs{\supp(y)} \le k$. A vector $x\in \R^n \setminus \{0^n\}$ is said to be \emph{$(k,\eps)$-$\ell_p$-compressible} if there exists a $k$-sparse $y\in \R^n$ such that $\smnorm{x-y}_p\le \eps\smnorm{x}_p$. Otherwise, we say that $x$ is \emph{$(k,\eps)$-$\ell_p$-spread}. 
		
				A subspace $X \subseteq \R^n$ is \emph{$(k,\eps)$-$\ell_p$-spread} if every $x \in X \setminus \{0^n\}$ is $(k,\eps)$-$\ell_p$-spread.
	\end{definition}
  		
	We note that $\ell_p$-spread implies $\ell_q$-spread, for any $q \leq p$, up to a small change in parameters. Formally, the following proposition holds.
	\begin{restatable}[$\ell_p$-spread implies $\ell_q$-spread]{proposition}{pspreadtoqspread}
    \label{prop:pspreadtoqspread}
    Suppose that $X \subseteq \R^n$ is $(2k, \eps)$-$\ell_p$-spread. Then for every $1 \leq q < p$, $X$ is $(k, \eps_q)$-$\ell_q$-spread for $\eps_q = \eps^2 \left(\frac{k}{n}\right)^{\frac{1}{q}}$.

    In particular, if $X$ is $(\Omega(n), \Omega(1))$-$\ell_p$-spread, then $X$ is also $(\Omega(n), \Omega(1))$-$\ell_q$-spread for every $1 \leq q < p$.
    \end{restatable}
    We prove \cref{prop:pspreadtoqspread} in \cref{sec:pspreadtoqspread}.
    
    Being $\ell_p$-RIP (\cref{def:RIP}) is a stronger property than being $\ell_p$-spread. This is formalized by the following proposition, which we prove in \cref{sec:riptospread}.
    \begin{restatable}[$\ell_p$-RIP implies $\ell_p$-spread]{proposition}{riptospread}
    \label{prop:riptospread}
    Let $p \in [1, \infty]$, and let $A \in \R^{m \times n}$ be a $(k,\eps)$-$\ell_p$-RIP matrix. Then, $\ker(A)$ is $(k, \eps')$-$\ell_p$-spread for $\eps' = \frac{1 - \eps}{2 + \eps(1 + (\frac{2n}{k})^{1 - \frac{1}{p}})}$.
    \end{restatable}
	
	\begin{definition}[Distortion of a vector]
		Let $1 \leq q < p$. The \emph{$(\ell_q,\ell_p)$-distortion} of a vector $x \in \R^n \setminus \{0\}$ is
		\begin{equation*}
	    \Delta_{q, p}(x) \defeq \frac{\norm{x}_p \cdot n^{1/q - 1/p}}{\norm{x}_q} \enspace.
	    \end{equation*}
	\end{definition}
	By \cref{lem:holder}, $1 \le \Delta_{q,p}(x)\le n^{1/q - 1/p}$ always holds. Note that $\Delta_{q,p}(x)=1$ iff all entries of $x$ are equal, and $\Delta_{q,p}(x) = n^{1/q - 1/p}$ iff $x$ is supported on a single coordinate.  So the distortion is a measure of the well-spreadness of the vector, with smaller values corresponding to more spread.

	\begin{definition}[Distortion of a subspace]
		Given $n\in \N$ and a subspace $X\subseteq \R^n$, the \emph{$(\ell_q,\ell_p)$-distortion} of $X$ is
		\begin{equation*}
		\Delta_{q,p}(X) \defeq \sup\{\Delta_{q,p}(x) \mid x\in X\setminus\{0\}\} \enspace.
		\end{equation*}
	\end{definition}
	We remark that $\Delta_{p}(X)$ (as defined in \cref{sec:intro}) is simply $\Delta_{1,p}(X)$.

    Finally, we observe the following equivalence between $\ell_p$-spread and $(\ell_q,\ell_p)$-distortion, which we prove in \cref{sec:companddist}.
	\begin{restatable}[Compressibility and distortion]{proposition}{companddist}
	\label{prop:companddist}
		The following holds for all $1\le q<p$, $k\in \N$ and $x\in \R^n$.
		\begin{enumerate}
		\itemsep=0ex
			\item Let $\eps > 0$. If $x$ is $(k, \eps)$-$\ell_p$-compressible then $\Delta_{q,p}(x) \ge \frac{1}{\inparen{\frac{k}{n}}^{\frac{1}{q} - \frac{1}{p}} + \eps}$. 
			\item The vector $x$ is $\Bigl(k,\frac{\inparen{\frac {n}{k}}^{\frac{1}{q}}}{\Delta_{q,p}(x)}\Bigr)$-$\ell_p$-compressible.
		\end{enumerate}
		In particular, if a subspace $X$ is $(k, \eps)$-$\ell_p$-spread, then $\Delta_{q,p}(X) \leq  \frac{1}{\eps} \left(\frac{n}{k}\right)^{\frac{1}{q}}$ for all $1 \leq q < p$.
	\end{restatable}

	%%%%%%%%%%%%%%%%%%%%%%%%%%%%%%%%%%%%%%%%%%%%%%%%%%%%%%%%%%%%%%%%%%%%%%%%%%%%%%%%
	%%%%%%%%%%%%%%%%%%%%%%%%%%%%%%%%%%%%%%%%%%%%%%%%%%%%%%%%%%%%%%%%%%%%%%%%%%%%%%%%
	%%%%%%%%%%%%%%%%%%%%%%%%%%%%%%%%%%%%%%%%%%%%%%%%%%%%%%%%%%%%%%%%%%%%%%%%%%%%%%%%
	\section{Properties of random biregular graphs}\label{sec:RandomGraphProperties}

	In this section, $G^* = (V_L,V_R,E^*)$ is a random graph sampled uniformly from $\cG_{m,n,s,t}$. We state \cref{prop:goodGraph,prop:uniqueExpansion} and \ref{prop:uniqueExpansion}, which show that $G^*$ satisfies certain properties with high probability.
	
    \begin{proposition}[Scarcity of cycles] \label{prop:goodGraph}
		Suppose that $1\le \ell = \ell(n) \le c\cdot \frac{\log (m)}{\log (ts)}$ for a small enough universal constant $c > 0$. Then, with high probability, the graph ${G^*}$ has the following properties.
		\begin{enumerate}
			\item There exists $v^*\in V_L$ such that $\ball_{G^*}(v^*,2\ell+1)$ does not contain a cycle.\label{enum:NoCycle}
			\item The graph $G^*$ is bicycle-free at radius $2\ell+1$. (See \cref{def:bicycleFree}) \label{enum:NoBicycle}
		\end{enumerate}
		%Here, $c > 0$ is a universal constant.
	\end{proposition}
	Note that the above proposition is vacuously true when $ts > m^c$.
	
	\begin{proposition}[Unique expansion] \label{prop:uniqueExpansion}
	    For some $c > 0$, the graph $G^*$ is a $\inparen{\frac{c\alpha^2}{t^4}, \frac 2t}$-unique expander, with high probability.
	 \end{proposition}

	\subsection{A negative association lemma}
	Our main technical tool for proving Propositions \ref{prop:goodGraph} and \ref{prop:uniqueExpansion} is \cref{lem:NAedges}. This lemma applies in the setting where a certain subset of edges $F\subseteq E^*$ of the random graph $G^*$ has been revealed. We fix some vertex $u$ and sample a random (yet unrevealed) edge, originating from $u$ to some random vertex $v\in N_{G^*}(u)$. The claim gives a natural lower bound for the probability that $\deg_F(v)=0$, namely, that $v$ is ``new,'' in the sense that it does not touch any of the edges in $F$.
	
	In order to state \cref{lem:NAedges} we introduce some notation. Given a set of edges $F\subseteq \binom{V_L\cup V_R}{2}$ and a vertex $v\in V_L\cup V_R$, we denote $N_F(v) = \{u\in V_L\cup V_R\mid \{v,u\}\in F\}$ and $\deg_F(v) = \inabs{N_F(v)}$.

	\begin{definition}
		A set of edges $F\subseteq \binom{V_L\cup V_R}{2}$ is called \emph{viable} if it can be completed to an $(s,t)$-biregular bipartite graph, namely, if $F\subseteq E(G)$ for some $G\in \cG_{m,n,s,t}$.
	\end{definition}

	\begin{lemma}[A Negative Association like property]
		\label{lem:NAedges}
		The following holds for every viable set of edges $F$.
		\begin{enumerate}
			\item Let $u\in V_L$ such that $\deg_F(u) < t$. Let $v$ be uniformly sampled from $N_{G^*}(u)\setminus N_F(u)$. Then,
			\begin{equation*}
			\PROver{G^*,v}{\deg_F(v)>0\mid F\subseteq E^*} \le \frac{\inabset{x\in V_R\mid \deg_F(x) > 0}}{\abs{V_R}} \enspace.
			\end{equation*}
			\item Let $u\in V_R$ such that $\deg_F(u) < s$. Let $v$ be uniformly sampled from $N_{G^*}(u)\setminus N_F(u)$. Then,
			\begin{equation*}
			\PROver{G^*,v}{\deg_F(v)>0\mid F\subseteq E^*} \le \frac{\inabset{x\in V_L\mid \deg_F(x) > 0}}{\abs{V_L}} \enspace.
			\end{equation*}
		\end{enumerate}
	\end{lemma}
	\begin{proof}
		We only prove the first claim. The second claim follows by symmetry.
		
		Let $V_R' = V_R\setminus N_F(u)$ and $V_R'' = \{x\in V_R\mid \deg_F(x)=0\}$. Fix two vertices $z\in V_R'$, $y\in V_R''$. We claim that 
		\begin{equation} \label{eq:NAEdges}
			\PROver{G^*}{y\in N_{G^*}(u)\mid F\subseteq E^*} \ge \PROver{G^*}{z\in N_{G^*}(u)\mid F\subseteq E^*} \enspace.
		\end{equation}
		This suffices to prove the Lemma, since
		\begin{align*}\PROver{G^*,v}{\deg_F(v)=0\mid F\subseteq E^*} &= \sum_{x\in V_R''}\PROver{G^*,v}{v=x\mid F\subseteq E^*} \\ &= \frac{1}{t-\deg_F(u)}\cdot \sum_{x\in V_R''}\PROver{G^*}{x\in N_{G^*}(u)\mid F\subseteq E^*} \\ &\ge  \frac{1}{t-\deg_F(u)}\cdot \frac{\inabs {V_R''}}{\inabs {V_R'}}\cdot \sum_{x\in V_R'}\PROver{G^*}{x\in N_{G^*}(u)\mid F\subseteq E^*}  \\ &=\frac{\inabs {V_R''}}{\inabs {V_R'}}   \ge \frac{\inabs {{V_R''}}}{\inabs{V_R}} \enspace,
		\end{align*}
		where the first inequality follows from \cref{eq:NAEdges}.
		
		We turn to proving \cref{eq:NAEdges}. For every set $W\subseteq V_L$ such that $u\in W$, and for every $1\le k\le \frac{|W|}2$, fix some injective map
		\begin{equation*}
		g_{W,k} \colon \{U\subseteq W \mid |U|=k \text{ and }u\in U\}\to \{U\subseteq W \mid |U|=k \text{ and }u\notin U\} \enspace.
		\end{equation*}
		Since $k\le \frac{|W|}2$, the right set is at least as large as the left one, and thus such a map exists.
		
		Let 
		\begin{equation*}\cY = \inset{G\in\cG_{m,n,s,t}\mid F\cup \{\{u,y\}\}\subseteq E(G)}
		\end{equation*}
		and 
		\begin{equation*}
		\cZ = \inset{G\in \cG_{m,n,s,t}\mid F\cup \{\{u,z\}\} \subseteq E(G)} \enspace.
		\end{equation*}
		Since the ratio between the left and right sides of \cref{eq:NAEdges} is $\frac{|\cY|}{|\cZ|}$, it suffices to show that $|\cY| \ge |\cZ|$. We do so by proving the existence an injection $f \colon \cZ\setminus \cY\to \cY\setminus \cZ$.
		
		We define the injection $f$ by fixing a graph $G = (V_L,V_R,E)\in \cZ\setminus \cY$ and describing its image $f(G)$. Let $S = N_G(y) \cap N_G(z)$ and $T = N_F(z) \setminus S$. Denote $U = N_G(z) \setminus (S\cup T)$,  $V = N_G(y) \setminus S$,  and $W = (N_G(y) \cup N_G(z)) \setminus (S\cup T)$.
		We can write the relevant neighbor sets as disjoint unions:
		\begin{equation*}
		N_G(y) = S\sqcup V \enspace, \quad 
		N_G(z) = S\sqcup T\sqcup U \enspace, \quad \text{and} \quad
		W = U\sqcup V \enspace.
		\end{equation*}
		Our assumption that $G\in \cZ\setminus \cY$ yields $u\in N_G(z)\setminus N_G(y) = U\sqcup T$.  Since $u\notin N_F(z) \supseteq T$, it follows that $u\in U$. Let $k = |U| = s - |S| - |T|$. Note that $|V| = s - |S| \ge k$, and so $|W| = |U|+|V| \ge 2k$.  Hence, $g_{W,k}(U)$ is well-defined.
		
		To describe the graph $f(G)$ we specify the neighbor sets, with regard to $f(G)$, of every vertex $x\in V_R$. Let
		\begin{equation*}
		N_{f(G)}(x) = \begin{cases}
			S \sqcup T\sqcup g_{W,k}(U) 
			&\text{if }x=z\\
			S \sqcup \inparen{W\setminus g_{W,k}(U)} &\text{if }x=y\\
			N_G(x) &\text{if }x\notin \{y,z\} \enspace.
		\end{cases}
		\end{equation*}
		
		Note that $G$ and $f(G)$ agree on all edges, except for some of the edges between $W$ and  $\{y,z\}$. In particular, observe that no edge of $F$ connects $W$ to $\{y,z\}$. Indeed, $U\cap N_F(z)$ by definition, and $\deg_F(y)=0$ due to our assumption that $y \in V_R''$. Thus, $F\subseteq E(f(G))$.
		
		Note that, in both $G$ and $f(G)$, every vertex in $W$ is connected to exactly one of $y$ and $z$. Hence, the left degrees of $G$ and $f(G)$ are identical. Since $\inabs{g_{W,k}(U)} = \inabs{U}$, the vertices $y$ and $z$ also have the same degree under $G$ and $f(G)$, implying that all the right degrees are identical. Consequently, $f(G)\in \cG_{m,n,s,t}$. By definition of $g_{W,k}$, we have $u\in W\setminus g_{W,k}(U)\subseteq N_{f(G)}(y)\setminus N_{f(G)}(z)$, and so $f(G)\in \cY\setminus \cZ$. 
		
		It remains to prove that $f$ is an injection. Indeed, suppose that $f(G) = f(G')$. It suffices to show that $G=G'$. Clearly, $G$ and $G'$ agree on every edge except, perhaps, for edges between $W$ and $\{y,z\}$. In particular, the set
		\begin{equation*}
		S = N_G(y)\cap N_G(z) = N_{G'}(y)\cap N_{G'}(z) = N_{f(G)}(y)\cap N_{f(G)}(z)
		\end{equation*}
		is well-defined in terms of $f(G)$. Since $g_{W,k}$ is injective, we have
		\begin{equation*}
		N_G(z)\cap W = N_{G'}(z)\cap W = g_{W,k}^{-1}\inparen{N_{f(G)}(z)\setminus (S\sqcup T)}
		\end{equation*}
		and
		\begin{equation*}
		N_G(y)\cap W = N_{G'}(y)\cap W = W\setminus N_G(z) \enspace,
		\end{equation*}
		so that $G=G'$.
	\end{proof}
	
	\subsection{Proof of \cref{prop:goodGraph}: scarcity of cycles}
	To prove \cref{prop:goodGraph}, we show the following lemma.
	
	\begin{lemma}
		\label{lem:treedepth}
		Let $\ell \in \N$ and $v^*\in V_L$. Then, 
		\begin{equation}\label{eq:prContainsCycle}
			\PR{\ball_{G^*}(v^*,2\ell+1)\text{ contains a cycle}} \le \frac{8\ell^{2}t^{2\ell+1}s^{2\ell}}{m} \enspace.
		\end{equation}
		and
		\begin{equation} \label{eq:prContains2Cycles}
			\PR{\ball_{G^*}(v^*,2\ell+1)\text{ contains two (possibly intersecting) cycles}} \le\frac{\ell^3 (ts)^{O(\ell)}}{m^2} \enspace.
		\end{equation}
	\end{lemma}
    \cref{prop:goodGraph} follows readily from \cref{lem:treedepth}. Indeed, the first property holds with high probability for some arbitrarily chosen vertex $v^*$ due to  \cref{eq:prContainsCycle}, provided that $c$ is small enough. The second property holds by \cref{eq:prContains2Cycles}, via a union bound over all vertices in $V_L$.
	
	\begin{proof}[Proof of \cref{lem:treedepth}]
		\item \paragraph{Proof of \cref{eq:prContainsCycle}:} Consider the following randomized algorithm.
		\begin{algorithm}[H]\caption{Detect a cycle near $v^*$\label{alg:DetectCycle}} \begin{algorithmic}[1]
				\State Set $v_0 = v^*$ and set $v_1$ be a uniformly random neighbor of $v_0$.
				\For{$i=2,\dotsc, 4\ell+1$}
				\State Sample $v_i$ uniformly from $N_{G^*}(v_{i-1})\setminus \{v_{i-2}\}$
				\If {$v_i \in \{v_0,\dotsc, v_{i-3}\}$}
				\State accept.
				\EndIf
				\EndFor
				\State reject.
			\end{algorithmic}
		\end{algorithm}
		
		Let $T_i$ denote the event that Algorithm \ref{alg:DetectCycle} reaches the $i$-th step and accepts on that step. Let $T = \bigcup_{i=1}^{4\ell+1} T_i$ denote the event that the algorithm eventually accepts.
		
		Lemma \ref{lem:NAedges}, applied to the edge set $\inset{\inset{v_{j-1},v_j}\mid 1\le j\le i-1}$ and to the vertex $v_{i-1}$, yields $$\PR{T_i\mid \bigcap_{j=2}^{i-1}\overline{T_j}} \le \begin{cases}
			\frac{i}{2|V_L|} &\text{if }$i$ \text{ is even}\\
			\frac{i-1}{2|V_R|} &\text{if }$i$ \text{ is odd} \enspace.\\
		\end{cases}
		$$ 
		Thus,
		$$\PR{T} \le \sum_{i=2}^{4\ell+1} \PR{T_i\mid \bigcap_{j=2}^{i-1}\overline{T_j}} \le  \frac{4\ell^{2}}{|V_L|} + \frac{4\ell^2}{|V_R|} \enspace.$$
		
		Let $C$ denote the event that $\ball_{G^*}(v^*,2\ell+1)$ contains a cycle. If such a cycle exists, at least one run of Algorithm \ref{alg:DetectCycle} will detect it. Since there are at most $t^{2\ell+1}\cdot s^{2\ell}$ possible runs,
		$$\PR{T\mid C} \ge \frac{1}{t^{2\ell+1}s^{2\ell}}\enspace.$$
		Therefore,
		$$\frac{4\ell^{2}}{|V_L|} + \frac{4\ell^2}{|V_R|} \ge \PR{T} \ge \PR{C}\cdot \PR{T\mid C} \ge \frac{\PR{C}}{t^{2\ell+1}s^{2\ell}}\enspace,$$
		and so,
		$$\PR{C} \le \ell^{2}t^{2\ell+1}s^{2\ell}\cdot\inparen{\frac 4{m} + \frac{4}{n}} \le \frac{8\ell^{2}t^{2\ell+1}s^{2\ell}}{m}\enspace. $$
		
		\paragraph{Proof of \cref{eq:prContains2Cycles}:}
		We prove \cref{eq:prContains2Cycles} along similar lines to \cref{eq:prContainsCycle}. We define another algorithm, which attempts to detect two cycles near $v^*$.
		\begin{algorithm}[H]\caption{Detect two cycles near $v^*$\label{alg:Detect2Cycles}} \begin{algorithmic}[1]
				\State Use Algorithm \ref{alg:DetectCycle} to seek a simple cycle in $\ball_{G^*}(v^*, 2\ell+1)$. \label{line:CallToInner}
				\If{a cycle $u_0,\dotsc,u_{r-1},u_0$ ($r < 4\ell + 1$) has been found}
				\State Sample a random $0\le j\le r-1$.
				\State Let $w_0 = u_j$. 
				\State Uniformly sample $w_1$ from $N_{{G^*}}(u_j)\setminus \{u_{j-1},u_{j+1}\}$ (where $j\pm 1$ is taken $\mod r$).
				\State In a similar fashion to Algorithm \ref{alg:DetectCycle}, seek a path $w_0,w_1,\dotsc, w_q$ ($q\le 4\ell +1$), where $w_q \in \{w_0,\dotsc, w_{q-3}\}\cup \{u_0,\dotsc, u_{r-1}\}$.
				\If{such a path has been found}
				\State accept.
				\Else
				\State reject.
				\EndIf
				\Else
				\State reject. 
				\EndIf
			\end{algorithmic}
		\end{algorithm}
		
		As before, we denote by $T$ the event that the call to Algorithm \ref{alg:DetectCycle} in Line \ref{line:CallToInner} accepts. We denote by $T'$ the event that Algorithm \ref{alg:Detect2Cycles} accepts.
		
		We bound $\PR{T'\mid T}$ similarly to our bound of $\PR{T}$. By Lemma \ref{lem:NAedges}, the probability that $w_i\in \{w_0,\dotsc, w_{i-3} \cup \{u_0,\dotsc, u_{r-1}\}$ is at most $\frac{i+r}{2|V_L|}$ (if $w_{i-1}\in V_R$) or $\frac{i+r}{2|V_R|}$ (if $w_{i-1} \in V_L$). Hence,
		$$\PR{T'\mid T} \le \frac{32\ell}{|V_L|} + \frac{32\ell}{|V_R|}\enspace,$$
		and so
		$$\PR{T'} = \PR{T'\mid T}\cdot \PR{T} \le \inparen{\frac{32\ell}{|V_L|} + \frac{32\ell}{|V_R|}} \cdot \inparen{\frac{4\ell^2}{|V_L|}+\frac{4\ell^2}{|V_R|}}\le O\inparen{\frac{\ell^3}{|V_L|^2}}\enspace.$$
		
		Let $C'$ denote the event that $\ball_{{G^*}}(v^*,2\ell+1)$ contains two cycles. As in the proof of \cref{eq:prContainsCycle}, $C'$ implies that at least one run of Algorithm \ref{alg:Detect2Cycles} accepts. Thus,
		$$\PR{T'\mid C'} \ge \frac{1}{(4\ell+1)\cdot t^{4\ell+2}s^{4\ell+2}}\enspace.$$
		We conclude that
		\begin{equation*}
		\PR{C'} \le \frac{\PR{T'}}{\PR{T'\mid C'}} \le \frac{\ell^3 (ts)^{O(\ell)}}{m^2} \enspace .\qedhere
		\end{equation*}
	\end{proof}

	\subsection{Proof of \cref{prop:uniqueExpansion}: unique expansion}
	To prove \cref{prop:uniqueExpansion}, we show the following lemma.
	\begin{lemma}\label{lem:ProbSetNotExpanding}
	    Let $S\subseteq V_L$ and denote $k = |S|$. Then
	    $$\PROver{G^*}{S\text{ has at most $(t-2)|S|$ unique neighbors}} \le \inparen{\frac{et^2k}{2m}}^{2k}.$$
	\end{lemma}
	Our proof of \cref{lem:ProbSetNotExpanding} is an adaptation of the proof of~\cite[Theorem 4.4]{Vadhan12}.

		\cref{prop:uniqueExpansion} then follows from \cref{lem:ProbSetNotExpanding} by the union bound. Let $\gamma = \frac{c\alpha^2}{t^4}$. Then,
	\begin{align*}
	    &\PROver{G^*}{G^*\text{ is not a $\inparen{\gamma, \frac 2t}$-unique expander}} \\ \le &\sum_{k=1}^{\gamma n} \sum_{\substack{S\subseteq [n]\\ |S|=k}}\PROver{G^*}{S\text{ has less than $(t-2)|S|$ unique neighbors}} \\ 
	    \le &\sum_{k=1}^{\gamma n} \binom nk \cdot \inparen{\frac{et^2k}{2m}}^{2k} \le \sum_{k=1}^{\gamma n} \inparen{\frac{ne}{k}}^{k}\cdot \inparen{\frac{et^2k}{2m}}^{2k} = \sum_{k=1}^{\gamma n} \inparen{\frac{e^3 t^4 k }{\alpha^2 n}}^k \enspace. \numberthis \label{eq:probNotExpandingSum}
	\end{align*}
	
	Taking $\gamma \le \frac{\alpha^2}{2e^3t^4}$, we have $\frac{e^3 t^4 k }{\alpha^2 n} < \frac 12$ whenever $k\le \gamma n$. Hence, the sum on the right-hand side of \cref{eq:probNotExpandingSum} is dominated by its first term, namely, it is $O\inparen{\frac{t^4}{\alpha^2 n}}$.

	\begin{proof}[Proof of \cref{lem:ProbSetNotExpanding}]
	Write $S = \{v_1,\dotsc, v_k\}$. Consider a process in which the $tk$ edges touching $|S|$ are revealed in $tk$ steps, where the first $t$ steps reveal the neighbors of $v_1$, the next $t$ steps reveal those of $v_2$, and so on. 
	
	Let $V_i\subseteq V_R$ ($0\le i\le tk$) denote the set of neighbors of $S$ revealed by the $i$-th step. Note that $V_i=V_{i-1}$ if and only if the neighbor revealed in the $i$-th step has already been revealed in a previous step. Otherwise, $|V_i| = |V_{i-1}|+1$. By \cref{lem:NAedges}, the probability of the event $V_i= V_{i-1}$, conditioned on all previous steps, is at most $\frac{\inabs{V_{i-1}}}m \le \frac{tk}m$. Thus,
	\begin{flalign*}
	   &\PR{S\text{ has at most $(t-2)k$ unique neighbors}} = \PR{\text{at least $2k$ steps $i$ have }V_i=V_{i-1}} \\
	   &\le \binom{tk}{2k} \inparen{\frac{tk}m}^{2k} \le \inparen{\frac{etk}{2k}}^{2k}\cdot \inparen{\frac {tk}m}^{2k} = \inparen{\frac{et^2k}{2m}}^{2k} \enspace . \qedhere 
	\end{flalign*}
	\end{proof}

	%%%%%%%%%%%%%%%%%%%%%%%%%%%%%%%%%%%%%%%%%%%%%%%%%%%%%%%%%%%%%%%%%%%%%%%%%%%%%%%%
	%%%%%%%%%%%%%%%%%%%%%%%%%%%%%%%%%%%%%%%%%%%%%%%%%%%%%%%%%%%%%%%%%%%%%%%%%%%%%%%%
	%%%%%%%%%%%%%%%%%%%%%%%%%%%%%%%%%%%%%%%%%%%%%%%%%%%%%%%%%%%%%%%%%%%%%%%%%%%%%%%%
	\section{Limitations of $\ell_2$-spread}\label{sec:negative}
	Our main goal in this section is to prove \cref{mthm:ell2spreadneg} given \cref{mthm:singvalue}. We recall \cref{mthm:ell2spreadneg}.
	\elltwospreadneg*
	
    The proof of \cref{mthm:ell2spreadneg} relies on the following two lemmas.

	\begin{lemma}
		\label{lem:treevector}
		Let $A\in \cM_{m,n,s,t}$. Suppose that $G_A$ satisfies Property \ref{enum:NoCycle} of \cref{prop:goodGraph} with regard to some $\ell\in \N$. Then, there is a vector $x \in \R^n\setminus\{0^n\}$ with $\abs{\supp(x)} \leq 1 + 2 t(t-1)^{\ell-1}(s-1)^{\ell}$ and
		$$\frac{\norm{Ax}_p^p}{\norm{x}_p^p} \le t(t-1)^\ell(s-1)^{(1-p)\ell}$$
		for all $p \ge 1$.
		In particular,
		\begin{equation}\label{eq:treeVectorell2UpperBound}
		    \frac {\norm{Ax}_2^2}{\norm{x}_2^2} \le \frac{t(t-1)^\ell}{(s-1)^\ell}.
		\end{equation}
		Furthermore, $x$ can be computed in polynomial time given $A$.
	\end{lemma}

	\begin{lemma}
		\label{lem:singvalrounding}
		Let $A \in \R^{m \times n}$ be a matrix, and let $x\in \R^n$ with $\twonorm x = 1$. Then, $\ker(A)$ contains a $\inparen{\inabs{\supp(x)}, \frac{\eps}{1-\eps}}$-$\ell_2$-compressible vector $y$, where $\eps = \frac{\twonorm{Ax}}{\sigma_{\min}(A)}$. Furthermore, $y$ can be computed in polynomial time given $A$ and $x$.
	\end{lemma}
	
	\cref{lem:treevector,lem:singvalrounding} are proven later in this section. We first use them to prove \cref{mthm:ell2spreadneg}.
	\begin{proof}[Proof of \cref{mthm:ell2spreadneg} given \cref{mthm:singvalue} and \cref{lem:treevector,lem:singvalrounding}]
		Let $c$ denote the constant from \cref{prop:goodGraph}. Let $\ell = \floor{c'\frac{\log (m)}{\log (ts)}-1}$, where $c' = \min(c,\frac 12)$.
		
		With high probability, we sample a matrix $A$ such that the events in \cref{prop:goodGraph} and \cref{mthm:singvalue} all occur. For the remainder of the proof, we shall view $A$ as fixed, assuming that these events hold.
		
		In particular, our assumption implies the hypothesis of \cref{lem:treevector}. Hence, there is a vector $x\in \R^n$, computable in polynomial time, which satisfies \cref{eq:treeVectorell2UpperBound}. Applying \cref{lem:singvalrounding} to $\frac x{\twonorm x}$ yields a $\inparen{2(ts)^\ell,\frac{\eps}{1-\eps}}$-$\ell_2$-compressible vector in $y\in \ker(A)$, where $\eps = \frac{\twonorm{Ax}}{\sigma_{\min}(A)}\le \frac{\sqrt t\cdot \alpha^{\frac \ell2}}{\sigma_{\min}(A)}$ due to \cref{eq:treeVectorell2UpperBound}. As the event in \cref{mthm:singvalue} occurs, $$\sigma_{\min}(A) \ge \sqrt{s-1}-\inparen{1+o(1)}\cdot \sqrt{t-1}\ge  \sqrt{s} -\inparen{1+o(1)}\sqrt t = \sqrt{s}\cdot \inparen{1-\sqrt\alpha -o(1)},$$
		and so
		$$\eps \le \frac{\alpha^{\frac {\ell+1}2}}{1-\sqrt\alpha - o(1)} \le \frac{n^{-\Omega\inparen{\frac{\log (1/\alpha)}{\log s}}}}{1-\sqrt\alpha}\enspace.$$
		Also, $(ts)^\ell \le m^{c'}$, so $y$ is $\inparen{m^{c'}, \frac{n^{-\Omega\inparen{\frac{\log (1/\alpha)}{\log s}}}}{1-\sqrt\alpha}}$-$\ell_2$-compressible. By \cref{prop:companddist}, 
		$$\Delta_{1,2}\inparen{\ker A} \ge \frac{1}{n^{-\Omega(1)}+ \frac{n^{-\Omega\inparen{\frac{\log (1/\alpha)}{\log s}}}}{1-\sqrt\alpha}}\ge (1-\sqrt \alpha)\cdot n^{\Omega\inparen{\frac{\log(1/\alpha)}{\log s}}}\enspace,$$
		where we use that $n^{-\Omega(1)} \leq n^{-\Omega\inparen{\frac{\log (1/\alpha)}{\log s}}}$ as $t = \alpha s \geq 3$.		
		Finally, the last part of the theorem follows since $y$ can be computed from $A,x$ in polynomial time.
	    \end{proof}
	
	    In addition, \cref{lem:treevector} suffices to prove \cref{mprop:ellpspreadneg}.
	    \ellpspreadneg*
	    
	\begin{proof}[Proof of \cref{mprop:ellpspreadneg} given \cref{lem:treevector}]
	    Let $\ell = \floor{c\frac{\log m}{\log(ts)}}$, where $c$ is as in \cref{prop:goodGraph}. The graph $G_A$ satisfies Property \ref{enum:NoCycle} of \cref{prop:goodGraph} with high probability, and so we proceed assuming that this is the case. \cref{lem:treevector} now yields a vector $x\in \R^n\setminus \{0\}$ with
	    \begin{align*}
	        \frac{\norm{Ax}_p}{\norm{x}_p} &\le \inparen{t(t-1)^\ell(s-1)^{(1-p)\ell}}^{\frac 1p} = t^{\frac 1p}\inparen{\frac{t-1}{s-1}\cdot (s-1)^{2-p}}^{\frac \ell p}\le t^{\frac 1p}\inparen{\alpha\cdot (s-1)^{2-p}}^{\frac \ell p} \\ &\le t^{\frac 1p}\inparen{1+\eps}^{-\frac \ell p}\le  t^{\frac 1p} \cdot m^{-\Omega\inparen{\frac{\eps}{\log s}}} \ . \qedhere
	    \end{align*} 
	\end{proof}
	%%%%%%%%%%%%%%%%%%%%%%%%%%%%%%%%%%%%%%%%%%%%%%%%%%%%%%%%%%%%%%%%%%%%%%%%%%%%%%%%
	%%%%%%%%%%%%%%%%%%%%%%%%%%%%%%%%%%%%%%%%%%%%%%%%%%%%%%%%%%%%%%%%%%%%%%%%%%%%%%%%
	\subsection{Proof of \cref{lem:treevector}: constructing $x$}
	Write $G_A = \inparen{V_L,V_R,E}$ as in \cref{sec:SparseMatrices}. By our assumption of \cref{prop:goodGraph} -- Property \ref{enum:NoCycle}, there is a vertex $v^*\in V_L$ such that $\ball_G(v^*,2\ell+1)$ is isomorphic to a $(t,s)$-biregular tree of depth $2\ell+1$, rooted at $v^*$. Recall that a $(t,s)$-biregular tree is a tree in which a non-leaf vertex has degree $t$ (resp. $s$) if it belongs to $V_L$ (resp. $V_R$).
	
	We define the vector $x\in \R^n$ as follows. Given $v\in V_L$, if $v\notin \ball_G\inparen{v^*,2\ell+1}$, let $x_v = 0$. Otherwise, we have $v\in V_L\cap \ball_G(v^*,2\ell+1)$. Let $v^*=u_0,r_1, u_1, r_2, u_2, \dotsc,r_k, u_{k}=v$ ($0\le k\le \ell$) denote the unique simple path in $\ball_G(v^*,2\ell+1)$ from $v^*$ to $v$. Let
	$$x_v = \prod_{i=1}^{k} \inparen{\sign(u_{i-1},r_i) \cdot \sign(u_i,r_i) \cdot \frac {-1}{s-1}}.$$
	
	Clearly, $x$ is computable in polynomial time given $A$.
	
	We compute $(Ax)_r = \sum_{v\in N_G(r)}\sign(v,r)\cdot x_{v}$ for every $r\in V_R$ by considering three cases:
	\begin{itemize}
		\item If $r$ is an inner vertex in $\ball_G(v^*,2\ell+1)$, denote its parent in the tree by $v'\in V_L$ and its children by $v_1,\dotsc, v_{s-1}\in V_L$. Note that $$x_{v_i} = \frac{-1}{s-1} \cdot x_{v'} \cdot \sign(v',r)\cdot \sign(v_i,r)$$ for $1\le k\le s-1$. Hence,
		\begin{align*}
			(A x)_r &= \sign(v',r)\cdot x_{v'} + \sum_{k=i}^{s-1}\sign(v_i,r)\cdot x_{v_i} \\ &= \sign(v',r)\cdot x_{v'} + \sum_{i=1}^{s-1} \frac{-1}{s-1}\cdot x_{v'}\cdot  \sign(v_i,r)^2 \cdot \sign(v',r) = 0,
		\end{align*}
		where we used the fact that $\sign(v_i,r)\in \{1,-1\}$, and so $\sign(v_i,r)^2=1$.
		\item If $r$ lies outside of $\ball_G(v^*,2\ell+1)$ then $(Ax)_r=0$ since every neighbor $v\in V_L$ of $r$ has $x_v=0$.
		\item Finally, if $r$ is a leaf of the tree $\ball_G(v^*,2\ell+1)$, then $r$ has a unique neighbor $v\in v_R$ such that $x_v\ne 0$ (namely, $v$ is the parent of $r$ in the tree). This vertex $v$ is distance $2\ell$ from $v^*$, so $|x_v|=(s-1)^{-\ell}$. Thus, $|(Ax)_r| = (s-1)^{-\ell}.$
	\end{itemize}
	
	Since the tree has $t\cdot (s-1)^\ell\cdot (t-1)^\ell$, leaves, 
	$$\norm{Ax}_p^p = t\cdot (s-1)^\ell\cdot (t-1)^\ell \cdot (s-1)^{-p\ell} = t\cdot(s-1)^{(1-p)\ell}\cdot (t-1)^\ell.$$
	Clearly, since $x_{v^*} = 1$, we have $\norm{x}_p^p \ge 1$, and the lemma follows.
	
	\subsection{Proof of \cref{lem:singvalrounding}: rounding $x$ to $y \in \ker(A)$}
	
	Let $\Pi$ be the orthogonal projection onto $\ker(A)$, and let $\Pi^{\perp}$ be the projection onto the subspace orthogonal to $\ker(A)$. Let $y = \Pi x \in \ker(A)$ be the vector such that $\smnorm{x - y}_2 = \min_{z \in \ker(A)} \smnorm{x - z}_2$.
     By using the SVD of $A$, we observe that $\smnorm{Ax}_2 \geq \sigma_{\min}(A) \cdot \smnorm{\Pi^{\perp} x}_2 = \sigma_{\min}(A) \cdot \smnorm{x - y}_2$, as $\Pi^{\perp} x = x - \Pi x = x - y$. Hence, $\smnorm{x - y}_2 \leq \frac{\twonorm {Ax}}{\sigma_{\min}(A)} =: \eps$. Therefore,
	$$\smnorm{x-y}_2 \le \frac{\eps}{\smnorm{y}_2}\cdot \smnorm{y}_2 \le \frac{\eps}{\smnorm{x}_2 - \smnorm{x-y}_2}\cdot \smnorm{y}_2 \le \frac{\eps}{\smnorm{x}_2 - \eps}\cdot \smnorm{y}_2 \enspace.$$ 
	We conclude that $y$ is $\inparen{\inabs{\supp(x)}, \frac{\eps}{\smnorm{x}_2-\eps}}$-$\ell_2$-compressible.

		%%%%%%%%%%%%%%%%%%%%%%%%%%%%%%%%%%%%%%%%%%%%%%%%%%%%%%%%%%%%%%%%%%%%%%%%%%%%%%%%
	%%%%%%%%%%%%%%%%%%%%%%%%%%%%%%%%%%%%%%%%%%%%%%%%%%%%%%%%%%%%%%%%%%%%%%%%%%%%%%%%
	%%%%%%%%%%%%%%%%%%%%%%%%%%%%%%%%%%%%%%%%%%%%%%%%%%%%%%%%%%%%%%%%%%%%%%%%%%%%%%%%
	\section{Positive results for $\ell_p$-RIP and $\ell_p$-spread for $p \in [1,2)$}
	\label{sec:ellpspreadpos}
	In this section, we prove \cref{mthm:ExpansionToRIP,mthm:ellpspreadpos,mthm:ExplicitConstruction}.
    We restate \cref{mthm:ExpansionToRIP} below.
    \ExpansionToRIP*
    The proof of \cref{mthm:ExpansionToRIP} relies on the following technical lemma, which we prove in \cref{sec:ellprip}.
       
	\begin{lemma}
	\label{lem:ellprip}
	Let $A\in \{0,1,-1\}^{m\times n}$ such that $G_A$ is a $t$-left-regular, $(\gamma, \mu)$-unique expander. Let $s_{\max}$ denote the maximum degree of a right vertex of $G_A$. Then for any $p \geq 1$, $\delta_1 > 0$, $\delta_2 \in (0,1)$, and $\gamma n$-sparse $x \in \R^n$,
	\begin{flalign}
	&\norm{Ax}_p^p \geq \left(\frac{t (1 - \mu)}{\left(1 + \delta_1 \right)^{p-1}}  - \frac{\mu t}{\delta_1^{p-1}} (s_{\max} - 1)^{p - 1}\right) \norm{x}_p^p \label{eq:ellpriplowerbound}\\
	&\norm{Ax}_p^p \leq \left(\frac{t}{\left(1 - \delta_2\right)^{p-1}} + \frac{\mu t}{\delta_2^{p-1}} (s_{\max}-1)^{p-1}\right)\norm{x}_p^p \label{eq:ellpripupperbound}\enspace.
	\end{flalign}
	\end{lemma}
    Note that for $p \geq 2$, \cref{eq:ellpriplowerbound} is trivial, as the right hand side is always $\leq 0$.
    
    We now proceed with the proofs as follows. Below, we prove \cref{mthm:ExpansionToRIP} from \cref{lem:ellprip}. We then prove  \cref{mthm:ellpspreadpos} in \cref{sec:mthmellpspreadpos}, \cref{mthm:ExplicitConstruction} in \cref{sec:mthmExplicitConstruction}, and finally \cref{lem:ellprip} in \cref{sec:ellprip}.
        \begin{proof}[Proof of \cref{mthm:ExpansionToRIP}]
        Suppose without loss of generality that $\norm x_p = 1$. Take $\delta_1=\delta_2 = \frac \eps 3$.  By \cref{eq:ellpriplowerbound},
	    $$
	        \norm{Ax}_p^p \ge \frac{t (1 - \mu)}{\left(1 + \delta_1 \right)^{p-1}}  - \frac{\mu t}{\delta_1^{p-1}} (s_{\max} - 1)^{p - 1} \enspace.
	    $$
	    Our lower bound on $\eps$ implies, in particular, that $\mu \le \frac \eps 3$. Thus,
	    by our choice of $\delta_1$,
	    $$\frac{t(1-\mu)}{(1+\delta_1)^{p-1}} \ge \frac{1-\frac \eps 3}{(1+\frac \eps 3)^{p-1}}\cdot t \ge \frac{1-\frac \eps 3}{1+\frac \eps 3}\cdot t\ge \inparen{1-\frac{2\eps}{3}} t \enspace.$$
        We also have
	    $$
	    \frac{\mu t}{\delta_1^{p-1}}(s_{\max}-1)^{p-1}\le \frac{\eps^2t}{9\delta_1} = \frac{\eps t}{3} \enspace.
	    $$
	    We conclude that
	    $$\norm{Ax}_p^p \ge (1-\eps) t \enspace.$$
	    
	    We similarly obtain an upper bound on $\norm{Ax}_p^p$. By \cref{eq:ellpripupperbound},
	    $$\norm{Ax}_p^p \le\frac{t}{\left(1 - \delta_2\right)^{p-1}} + \frac{\mu t}{\delta_2^{p-1}} (s_{\max}-1)^{p-1}\enspace.$$
	    Now,
	    $$\frac{t}{\left(1 - \delta_2\right)^{p-1}}\le \frac{t}{1-\frac \eps 3} \le \inparen{1+\frac{2\eps}{3}}t$$
	    since $\eps < \frac 12$. Also,
	    $$\frac{\mu t}{\delta_2^{p-1}} (s_{\max}-1)^{p-1} = \frac{\mu t}{\delta_1^{p-1}} (s_{\max}-1)^{p-1} \le \frac{\eps t}{3} \enspace,$$
	    and so,
	    $$\norm{Ax}_p^p \le (1+\eps) t \enspace.$$
        Finally, we observe that $(1 + \eps)^{1/p} \leq 1 + \eps$ and $(1 - \eps)^{1/p} \geq 1 - \eps$, which finishes the proof.
    \end{proof}

	%%%%%%%%%%%%%%%%%%%%%%%%%%%%%%%%%%%%%%%%%%%%%%%%%%%%%%%%%%%%%%%%%%%%%%%%%%%%%%%%
	%%%%%%%%%%%%%%%%%%%%%%%%%%%%%%%%%%%%%%%%%%%%%%%%%%%%%%%%%%%%%%%%%%%%%%%%%%%%%%%%
	\subsection{Proof of \cref{mthm:ellpspreadpos}: $\ell_p$-RIP and $\ell_p$-spread of sparse random matrices}
	\label{sec:mthmellpspreadpos}
	We now prove \cref{mthm:ellpspreadpos}, which we restate below.\ellpspreadpos*
	 \begin{proof}
	 Fix $\eps$, and let $s$ satisfy $s \geq (\frac{18}{\alpha \eps^2})^{\frac{1}{2 - p}}$. Note that this implies $\eps^2 \geq \frac{18}{t} \cdot s^{p-1}$, as $t = \alpha s$.
	 
	 By \cref{prop:uniqueExpansion}, with high probability it holds that $G_A$ is a $(\gamma, \mu)$-unique expander with $\gamma = \Omega(\alpha^2/t^4)$ and $\mu = \frac{2}{t}$. Hence, we have $\eps^2 \geq 9\mu s^{p-1} = \frac{18}{t} \cdot s^{p-1}$, and $s_{\max} = s$ as $G_A$ is $s$-right-regular. By \cref{mthm:ExpansionToRIP}, we thus have that $A$ is $(\gamma n, \eps)$-$\ell_p$-RIP. Applying \cref{prop:riptospread}, we conclude that $\ker(A)$ is $(\gamma n,\eps')$-$\ell_p$-spread for $\eps' = \frac{1 - \eps}{2 + \eps(1 + (\frac{2}{\gamma})^{1 - \frac{1}{p}})}$. As $\eps < \frac{1}{2}$, it follows that $\eps' = \Omega(\gamma^{1 - \frac{1}{p}}) = \Omega((\frac{\alpha^2}{t^4})^{1 - \frac{1}{p}})$. 
	 Then, applying \cref{prop:companddist}, we conclude that $\Delta_p(\ker(A)) \leq \frac{1}{\gamma \eps'} = O\big(1/\gamma^{2 - \frac{1}{p}}\big) = O\big((\frac{t^4}{\alpha^2})^{2 - \frac{1}{p}}\big)$.
	 \end{proof}
	
	%%%%%%%%%%%%%%%%%%%%%%%%%%%%%%%%%%%%%%%%%%%%%%%%%%%%%%%%%%%%%%%%%%%%%%%%%%%%%%%%
	%%%%%%%%%%%%%%%%%%%%%%%%%%%%%%%%%%%%%%%%%%%%%%%%%%%%%%%%%%%%%%%%%%%%%%%%%%%%%%%%
	\subsection{Proof of \cref{mthm:ExplicitConstruction}: explicit construction of $\ell_p$-RIP matrices}
	\label{sec:mthmExplicitConstruction}
    In this subsection, we prove \cref{mthm:ExplicitConstruction}, which we restate below.
    \ExplicitConstruction*
    
\begin{proof}
    We prove \cref{mthm:ExplicitConstruction} by combining \cref{mthm:ExpansionToRIP} with the explicit constructions of vertex expanders due to \cite{CapalboRVW02}. We note that these expanders do not necessarily have bounded right degree, which is necessary to use \cref{mthm:ExpansionToRIP}. Because of this, we first give a simple preprocessing algorithm to convert an expander to a new graph with similar expansion and bounded right degree.
    
    \begin{lemma}\label{lem:BoundDegrees}
        Let $G= (V_L,V_R,E)$ be a bipartite $t$-left regular $(\gamma,\mu)$-unique expander where $n:=|V_L|\ge |V_R|$.
        Then, there exists a bipartite $t$-left regular $(\gamma,\mu)$-unique expander $G' = (V_L',V_R',E')$, with $V_L' = V_L$ and $|V_R'| \le 3|V_R|$, such that every vertex in $V_R'$ has degree at most $\frac{tn}{|V_R|}$. Futhermore, $G'$ can be computed from $G$ in $\poly(|V_L|)$ time. 
    \end{lemma}
    \begin{proof}
        Write $\beta = \frac{|V_R|}{|V_L|}\le 1$. We modify $G$ to create the new graph $G'$ via the following algorithm: While the graph contains a right vertex $r$ with $\deg(r) > \frac{t}{\beta}$, create a new vertex $r'$. Let $D$ be some arbitrary subset of $N(r)$ such that $|D| = \floor{\frac{t}{\beta}}$.        
        Remove all edges between $r$ and $D$, and instead create edges between $r'$ and every vertex in $D$. 
	        
	    Clearly, when this process terminates, the maximum right degree of the graph is at most $\frac{t}{\beta}$, and the left-degrees remain as before. Also, observe that the neighbor set of any set of left vertices can only increase in size, so our modification to the graph does not hurt its unique expansion.
	        
	    It is left to show that we did not add too many right vertices. Indeed, the number of new right vertices in the graph is equal to the number of iterations of our algorithm. Let $J_i$ denote the number of edges in graph that touch a right vertex with degree larger than $\frac{t}\beta$, after the $i$-th iteration. Clearly, $J_0 \le |E| = nt$, and $J_{i+1} \le J_{i} - \floor{\frac{t}{\beta}}$. Hence, the number of iterations is at most $J_0/\floor{\frac{t}{\beta}} \leq J_0 \cdot \frac{2 \beta}{t} \leq 2 \abs{V_R}$. 
    \end{proof}
	    
	 We now turn to the proof of \cref{mthm:ExplicitConstruction}.
	 Let $c$ be the constant in \cref{thm:CRVW}, and let $p \geq 1$ satisfy $p < p_0 := 1 + \frac{1}{c}$. Let $\eps \in (0, \frac{1}{2})$, and let $\alpha \in (0,1)$. Let $\beta = \frac{\alpha}{3}$. We shall assume without loss of generality that $\beta$ is an inverse power of $2$, as otherwise we can simply decrease $\alpha$ until this holds, and this will only lose a factor of $2$ in $\alpha$.
	 
	 Let $\mu = ( \alpha^{(1 + c)(p - 1)} \eps^2 C)^{\frac{1}{1 - c(p-1)}}$, where $C$ is a universal constant. Note that $1 - c(p-1) > 0$ as $p < 1 + \frac{1}{c}$.
	 
	 Now, let $n$ be sufficiently large. We shall assume that $n$ is a power of $2$ without loss of generality, as otherwise we can simply increase $n$ until this holds.
	 By \cref{thm:CRVW}, we can construct in $\poly(n)$-time (as $\beta, \mu$ are constants) a bipartite graph $G = (V_L, V_R, E)$ with $\abs{V_L} = n$, $\abs{V_R} = \beta n$, such that $G$ is a $t$-left-regular $(\Omega(\frac{\mu \beta n}{t} ), \mu)$-unique expander and $t = O\Big(\big(\frac{1}{\beta\mu}\big)^c\Big)$. Applying \cref{lem:BoundDegrees}, we thus construct in $\poly(n)$-time a bipartite graph $G' = (V'_L, V'_R, E')$ with $\abs{V'_L} = n$, $\abs{V'_R} = 3 \beta n = \alpha n$, and $G'$ is a $t$-left-regular $(\Omega(\frac{\mu \alpha n}{t}), \mu)$-unique expander with max right degree $s_{\max} \leq \frac{t n}{\beta n} = \frac{3 t}{\alpha}$.
	    
	 We now observe that $\eps^2 \geq 9 \mu s_{\max}^{p-1}$. Indeed, this is because
	 \begin{equation*}
	9\mu s_{\max}^{p-1} = 9\mu \cdot \left(\frac{3 t}{\alpha}\right)^{p - 1} \leq O(1) \cdot \mu \cdot \left( \frac{1}{\alpha (\alpha \mu)^c} \right)^{p-1} = O(1) \cdot \frac{1}{\alpha^{(1 + c)(p-1)}} \cdot \mu^{1 - c(p-1)} \leq \eps^2 \enspace,
	 \end{equation*}
	 by our choice of $\mu$. Thus, by \cref{mthm:ExpansionToRIP}, we conclude that the adjacency matrix $B \in \{0,1\}^{m \times n}$ of $G'$, defined by $B_{r,u} = 1$ if $(u,r) \in E'$ and $0$ otherwise, is $(\Omega(\gamma n), \eps)$-$\ell_p$-RIP, where $\gamma = \frac{\mu \alpha }{t}$.
	 
	 By \cref{prop:riptospread}, we thus have that $\ker(B)$ is $(\gamma n, \eps')$-$\ell_p$-spread, where $\eps' = \Omega(\gamma^{1 - \frac{1}{p}})$, and by \cref{prop:companddist}, we have $\Delta_p(\ker(B)) \leq O(\frac{1}{\eps' \gamma}) = O(1/\gamma^{2 - \frac{2}{p}})$.
	 
	 We finish the proof by simply observing that $\frac{1}{1 - c(p-1)} = \frac{p_0 - 1}{p_0 - p} = O(\frac{1}{p_0 - p})$, and so $\gamma = \poly(\mu, \alpha) = \poly(\eps, \alpha)^{\frac{1}{p_0 - p}}$.
    \end{proof}

		%%%%%%%%%%%%%%%%%%%%%%%%%%%%%%%%%%%%%%%%%%%%%%%%%%%%%%%%%%%%%%%%%%%%%%%%%%%%%%%%
	%%%%%%%%%%%%%%%%%%%%%%%%%%%%%%%%%%%%%%%%%%%%%%%%%%%%%%%%%%%%%%%%%%%%%%%%%%%%%%%%
	\subsection{Proof of \cref{lem:ellprip}: $\ell_p$-RIP from unique expansion}
	\label{sec:ellprip}
	 Let $x \in \R^n$ be $\gamma n$-sparse, and let $S = \supp(x)\subseteq V_L$. Let $E(S, N(S)) := \{\{v,r\} \in E : v \in S, r \in N(S)\}$.
		The following claim asserts the existence of a certain many-to-one matching between $S$ and $N(S)$.
	
	\begin{claim}\label{claim:MatchingExists}
	    There exists a set of edges $M \subseteq E(S,N(S))$ with the following properties:
	    \begin{enumerate}
	        \item Every $r\in N(S)$ touches exactly one edge in $M$.
	        \item Every $v\in S$ touches at least $t(1-\mu)$ edges in $M$.
	    \end{enumerate}
	\end{claim}
	\begin{proof}
        We construct $M$ via an iterative algorithm. The algorithm records a set of \emph{processed} left vertices $P \subseteq S \subseteq V_L$. It also maintains the edge set $M$. The algorithm is initialized with $P \defeq \emptyset$ and $M\defeq \emptyset$. 
        
        There are $|S|$ iterations. On each iteration, the algorithm picks a vertex $v \in S \setminus P$ such that the set $U_v := N(v) \setminus (N(S \setminus P))$ of unique neighbors of $v$ has size at least $t(1 - \mu)$. Note that by unique expansion of the set $S\setminus P$, there is always such a vertex.
	    For each $r \in U_v$, the algorithm adds the edge $\{v,r\}$ to $M$. The vertex $v$ is added to $P$.
	    
	    We turn to analyzing this algorithm. It is straightforward to observe that when the algorithm terminates we have $P = S$, and every vertex $v \in S$ touches at least $t(1 - \mu)$ edges of $M$. We claim that every vertex $r\in N(S)$ touches exactly one edge of $M$. Indeed, let $h = \abs{N(r) \cap S}$ denote the number of neighbors of $r$ in $S$. Let $u_1, \dots, u_h$ be these neighbors, ordered so that $u_1$ is the first neighbor of $r$ added to $P$, $u_2$ is the second one, etc. We observe that for $i < h$, the edge $\{u_i, r\}$ cannot be in $M$. Indeed, when $u_i$ is added to $P$, $u_h$ is not in $P$, and thus $r$ is cannot be a unique neighbor of $u_i$. Next, we observe that when $u_h$ is added to $P$, the edge $\{u_h,r\}$ is added to $M$. This is because all $u_i \in P$ for all $i < h$, and so at this iteration of the algorithm $r$ has only one neighbor not in $P$, namely $u_h$. This finishes the proof.
	\end{proof}
	
	We are now ready to establish the lower and upper bounds on $\norm{Ax}_p$ claimed in  \cref{eq:ellpriplowerbound,eq:ellpripupperbound}. We begin with \cref{eq:ellpriplowerbound}.

	\begin{proof}[Lower bound on $\norm{Ax}_p$]
%	Proof of \cref{eq:ellpriplowerbound}]
    Let $\kappa_1 = \frac{1}{\delta_1} > 0$. Let $M$ be the set of edges obtained from \cref{claim:MatchingExists}. For each $r\in N(S)$, let $v_r$ denote the unique vertex in $S$ such that $\{v_r,r\}\in M$, and let $W_r = (N(r)\cap S)\setminus \{v_r\}$ denote the rest of the neighbors of $r$ in $S$. Define
    \begin{equation*}
    a_r = \inabs{(Ax)_r}^p + \kappa_1^{p-1} (s_{\max}-1)^{p-1}\sum_{u\in W_r}\inabs{x_u}^p \enspace .
    \end{equation*}
    Now note that
    \begin{align*}
        \sum_{r\in N(S)} a_r &= \norm{Ax}_p^p + \kappa_1^{p-1}(s_{\max}-1)^{p-1}\sum_{r\in N(S)}\sum_{u\in W_r} \inabs{x_u}^p \\ &\le \norm{Ax}_p^p + \kappa_1^{p-1}(s_{\max}-1)^{p-1} \mu t \sum_{u\in S} \inabs{x_u}^p \\ &= \norm{Ax}_p^p + \kappa_1^{p-1}(s_{\max}-1)^{p-1} \mu t \norm{x}_p^p, \numberthis \label{eq:suma_rUpperBound}
    \end{align*}
    where the inequality is due to each $u\in S$ belonging to at most $\mu t$ sets $W_r$ ($r\in N(S)$). Indeed, the number of such sets for a given $u$ is $t - \inabset{r\in N(S)\mid u=v_r} \le t- (1-\mu)t = \mu t$.
    
    Next, for any fixed $r\in N(S)$, we claim that 
    \begin{equation}\label{eq:a_rLoweround}
        a_r \ge \frac{1}{\left(1 + \frac{1}{\kappa_1}\right)^{p-1}}\abs{x_{v_r}}^p \enspace.
    \end{equation}
    Denote $z = \sum_{u\in W_r} \inabs{x_u}$. Note that $|W_r| = \inabs{(N(r)\cap S)\setminus \{v_r\}} \le \deg(r)-1 \le s_{\max}-1$, so H\"{o}lder's inequality (\cref{lem:holder}) yields 
    \begin{equation*}
    a_r = \inabs{\inparen{Ax}_r}^p + \kappa_1^{p-1}(s_{\max}-1)^{p-1}\sum_{u\in W_r}\inabs{x_u}^p \ge \inabs{\inparen{Ax}_r}^p + \kappa_1^{p-1} z^p \enspace.
    \end{equation*}
    To deduce \cref{eq:a_rLoweround} from the above, we distinguish two cases:
    \begin{itemize}
	    \item If $z \geq \abs{x_{v_r}}$, we can bound $$a_r \ge \kappa_1^{p-1} z^p \geq \kappa_1^{p-1}\inabs{x_{v_r}}^p\ge \frac{1}{\left(1 + \frac{1}{\kappa_1}\right)^{p-1}}\inabs{x_{v_r}}^p \enspace.$$
    	\item If $z \leq |x_{v_r}|$, we let $z = \beta \abs{x_{v_r}}$ for some $\beta \in [0,1]$. We then have $\inabs{(Ax)_r}^p \geq (\inabs{x_{v_r}} - z)^p \geq (1 - \beta)^p \abs{x_{v_r}}^p$, and conclude that 
    	$$
    	a_r \ge \inabs{\inparen{Ax}_r}^p + \kappa_1^{p-1} z^p \ge (1 - \beta)^p \abs{x_{v_r}}^p + \kappa_1^{p-1}\beta^p \abs{x_{v_r}}^p  
    	\geq \frac{1}{\left(1 + \frac{1}{\kappa_1}\right)^{p-1}} \abs{x_v}^p  \enspace,
    	$$
    	where we use that $\min_{\beta \in [0,1]} (1 - \beta)^p + \kappa_1^{p-1} \beta^p = \left(1 + \frac{1}{\kappa_1}\right)^{1 - p}$ for $\kappa_1 > 0$.
	\end{itemize}
	
	Now, \cref{eq:a_rLoweround} yields
	\begin{flalign*}
	&\sum_{r\in N(S)}a_r \ge \frac{1}{\left(1 + \frac{1}{\kappa_1}\right)^{p-1}}\sum_{v\in S}\inparen{\inabs{x_v}^p \cdot \inabset{r\in N(S)\mid v = v_r}} \\
	&\ge \frac{1}{\left(1 + \frac{1}{\kappa_1}\right)^{p-1}}\sum_{v\in S}\inabs{x_v}^p \cdot (1-\mu)t = \frac{1}{\left(1 + \frac{1}{\kappa_1}\right)^{p-1}}(1-\mu)t \norm{x}_p^p\enspace.
	\end{flalign*}
	\cref{eq:ellpriplowerbound} follows from the above and \cref{eq:suma_rUpperBound}, as $\kappa_1 = \frac{1}{\delta_1}$.
	\end{proof}
	
	\begin{proof}[Upper bound on $\norm{Ax}_p$]
	%[Proof of \cref{eq:ellpripupperbound}]
		We now prove \cref{eq:ellpripupperbound} in a similar way.
    Let $\kappa_2 = \frac{1}{\delta_2} > 1$, and let 
	$$b_r = \inabs{\inparen{Ax}_r}^p - \kappa_2^{p-1}(s_{\max}-1)^{p-1}\cdot \sum_{u\in W_r}\inabs{x_u}^p \enspace.$$
	Note that
	\begin{align*}
		\sum_{r\in N(S)} b_r &\ge \norm{Ax}_p^p - \kappa_2^{p-1}(s_{\max}-1)^{p-1}\sum_{r\in N(S)}\sum_{u\in W_r} \inabs{x_u}^p \\ &\ge \norm{Ax}_p^p - \mu t\kappa_2^{p-1}(s_{\max}-1)^{p-1}  \sum_{u\in S} \inabs{x_u}^p \\ &= \norm{Ax}_p^p -\mu t \kappa_2^{p-1}(s_{\max}-1)^{p-1} \norm{x}_p^p \enspace. \numberthis \label{eq:sumb_rUpperBound}
	\end{align*}

	We next seek an upper bound on $b_r$. Fix some $r\in N(S)$, and write $ \sum_{u\in W_r} \inabs{x_u} = \beta \inabs{x_{v_r}}$ for some $\beta\ge 0$. Note that $$\inabs{\inparen{Ax}_r}^p \le \big|\abs{x_{v_r}}+\sum_{u\in W_r}\abs{x_u}\big|^p \le (1+\beta)^p \inabs{x_{v_r}}^p \enspace.$$
	Also, by H\"{o}lder's inequality (\cref{lem:holder}),
	$$\sum_{u\in W_r}\inabs{x_u}^p \ge \inabs{W_r}^{1 - p} \inparen{\sum_{u\in W_r}\inabs{x_u}}^{p} = \inabs{W_r}^{1-p} \cdot \inparen{\beta \inabs{x_{v_r}}}^p \ge (s_{\max}-1)^{1-p}\cdot \inparen{\beta \inabs{x_{v_r}}}^p\enspace.$$
	It follows that 
	$$b_r \le \inparen{(1+\beta)^p - \kappa_2^{p-1}\cdot \beta^p}\inabs{x_{v_r}}^p \le \frac{1}{\left(1 - \frac{1}{\kappa_2}\right)^{p-1}}\inabs{x_{v_r}}^p,$$
	since $\sup_{\beta \geq 0}\{(1+\beta)^p - \kappa_2^{p-1}\cdot \beta^p\} = \frac{1}{\left(1 - \frac{1}{\kappa_2}\right)^{p-1}}$ for $\kappa_2 > 1$.
	
	By \cref{eq:sumb_rUpperBound} and the above,
	$$\norm{Ax}_p^p \le \sum_{r\in N(S)}b_r + \kappa_2^{p-1} (s_{\max}-1)^{p-1}\mu t\norm{x}_p^p \le \left(\frac{t}{\left(1 - \frac{1}{\kappa_2}\right)^{p-1}} +\kappa_2^{p-1} (s_{\max}-1)^{p-1}\mu t\right)\norm{x}_p^p \enspace,$$
    since each $u \in S$ can appear as $v_r$ for some $r \in N(S)$ at most $t$ times. As $\kappa_2 = \frac{1}{\delta_2}$, this yields \cref{eq:ellpripupperbound}.
	\end{proof}

	%%%%%%%%%%%%%%%%%%%%%%%%%%%%%%%%%%%%%%%%%%%%%%%%%%%%%%%%%%%%%%%%%%%%%%%%%%%%%%%%
	%%%%%%%%%%%%%%%%%%%%%%%%%%%%%%%%%%%%%%%%%%%%%%%%%%%%%%%%%%%%%%%%%%%%%%%%%%%%%%%%
	%%%%%%%%%%%%%%%%%%%%%%%%%%%%%%%%%%%%%%%%%%%%%%%%%%%%%%%%%%%%%%%%%%%%%%%%%%%%%%%%
	\section{Singular values of random biregular matrices}
	\label{sec:singvalue}
	In this section, we prove \cref{mthm:singvalue}, which is restated below.
	\singvalue*
	
	For the entire section, we will let $A \in \cM_{m,n,s,t}$, and let $G \defeq G_A=(V_L,V_R,E)$ and $\sign = \sign_A$. For convenience, let us define $a\pm b:= [a -b , a+ b]$.
	
	We begin our proof of \cref{mthm:singvalue} with the following simple claim.
	\begin{claim}\label{clm:MtoA}
		Let $M = AA^{\top} - s \cdot \Id$. Suppose that \begin{equation}\label{eq:MGoodSpectrum}\Spec(M) \subseteq t-2 \pm (2+\eps)\sqrt{(s-1)(t-1)}\end{equation} 
		for some $\eps = o(1)$. 
		Then, \begin{equation}\label{eq:sigmaBounded}\sigma(A) \subseteq \sqrt{s-1} \pm (1+\eps')\sqrt{t-1}\end{equation}
		for some $\eps' = o(1).$
	\end{claim}
	\begin{proof}
		We immediately have that $$\Spec(AA^{\top}) \subseteq (s-1) + (t-1) \pm (2 + \eps) \sqrt{(s-1)(t-1)}.$$ The singular values of $A$ are obtained by taking the square root of the eigenvalues of $AA^{\top}$. Hence, \cref{eq:sigmaBounded} holds whenever $\eps'$ satisfies
		$$2\eps' - (2\eps'+\eps'^2)\sqrt{\frac {t-1}{s-1}}\ge \eps.$$
		Because $\frac{t-1}{s-1} \le \frac{t}{s} = \alpha$ is a constant $< 1$, we can take $\eps' \le O_\alpha(\eps)$.
	\end{proof}
	
	Thus, we focus on showing that \cref{eq:MGoodSpectrum} holds with high probability when $A$ is sampled uniformly from $\cM_{m,n,s,t}$. To do this, we use the following definitions and theorem from \cite{MohantyOP20b}.	
	\begin{definition}[Nomadic walk matrix]
	    A \emph{nomadic pair} is a $2$-tuple $(e,e')$ of (undirected) edges in $E$, where $e = \{u,v\}$ and $e' = \{v,w\}$ with $u \in V_R$, $v \in V_L$, and $w \in V_R$ with $w \ne u$. Let $$\nomadicv{G} = \inset{\inparen{e_1,e_2}\mid \inparen{e_1,e_2}\text{ is a nomadic pair in $G$}}\enspace.$$
	    
		The \emph{nomadic walk matrix} of $A$ is the matrix $B\in \R^{\nomadicv{G}\times \nomadicv{G}}$, such that $B[(e_1, e_2), (e_3, e_4)] = \sign(e_3)\cdot \sign(e_4)$ if $\inparen{e_1, e_2, e_3,  e_4}$  forms a non-backtracking\footnote{A non-backtracking walk of length $\ell$ is a sequence of vertices $v_0, \dots, v_{\ell}$ where $(v_{i-1}, v_{i}) \in E$ and $v_{i-1} \ne v_i$ for all $i \in [\ell]$.} walk of length $4$ in $G$. Otherwise, $B[(e_1, e_2), (e_3, e_4)]=0$.
	\end{definition}
	
	\begin{theorem}[A modified Ihara-Bass formula {\cite[Theorem 3.1]{MohantyOP20b}}]\label{thm:IharaBass}
		Let $A\in \cM_{m,n,s,t}$ and let $M = AA^{\top} - s \cdot \Id$. Let $L(z) = \Id - z \cdot M + z \cdot (t - 2)  \Id + z^2 \cdot (s-1)(t-1) \Id$. Then,
		\begin{equation}\label{eq:IharaBass}
			(1 - z)^{m \cdot \frac{s(t-1)}{t} - 1} (1 + (t-1)z)^{m \cdot \frac{s}{t} - 1} \det(L(z)) = \det(\Id - B z) \enspace,
		\end{equation}
		where $B$ is the \emph{nomadic walk} matrix of $A$.
	\end{theorem}
	We note that the above theorem is stated in \cite{MohantyOP20b} in more generality: for simplicity we only state the version specific to our application.
	
	We use \cref{thm:IharaBass} to connect the spectrum of $M$ with the spectrum of the nomadic walk matrix $B$.
	
	\begin{claim}\label{clm:BtoM}
	    Let $M=AA^\top-s\cdot\Id$ and let $B$ denote the nomadic walk matrix of $A$. 
		Suppose that \begin{equation}\label{eq:opNormBBound}\rho(B) \leq (1 + \eps) \sqrt{(s-1)(t-1)}\end{equation} for some $\eps \leq \frac 12$. Then 
		\begin{equation*}
		\Spec(M) \subseteq t-2 \pm \inparen{2+4\eps^2} \sqrt{(s-1)(t-1)} \enspace.
		\end{equation*}
		In particular, if $\eps = o_{n\to\infty}(1)$, then \cref{eq:MGoodSpectrum} holds.
	\end{claim}
	\begin{proof}
		Let $\lambda$ be an eigenvalue of $M$ so that there is a nonzero vector $x$ with $M x = \lambda x$. Let $z\in \R$ be such that \begin{equation}\label{eq:zQuadratic}1 + z (t - 2 - \lambda) + z^2 (s-1)(t-1) = 0.\end{equation} Note that $z \ne 0$, and so $\lambda = \frac{(1 + z (t - 2) + z^2 (s-1)(t-1))}{z}$. Then, rearranging, we have that $L(z) x = 0$, so that $\det(L(z)) = 0$, and therefore $\det(\Id - z B) = 0$ due to \cref{eq:IharaBass}. It then follows that $\frac{1}{z}$ is an eigenvalue of $B$.
		
		Now, suppose that there exists an eigenvalue $\lambda$ of $M$ such that either $\lambda = t - 2 - (2 + \delta) \sqrt{(s-1)(t-1)}$ or $t - 2 + (2 + \delta) \sqrt{(s-1)(t-1)}$ for some $\delta \geq 0$. In either case, by choosing $z$ to be the solution to \cref{eq:zQuadratic} that minimizes $\abs{z}$, we see that $z = \pm \frac{1 + \frac{\delta}{2} - \frac{1}{2} \sqrt{2 \delta + \delta^2}}{\sqrt{(s-1)(t-1)}}$. Hence, $B$ has an eigenvalue $\mu$ with $$\rho(B)\ge \abs{\mu} \geq \frac{\sqrt{(s-1)(t-1)}}{1 + \frac{\delta}{2} - \frac{1}{2} \sqrt{2 \delta + \delta^2}}\enspace.$$ 
		If $\delta \geq 1$, then this implies that $\rho(B) > \frac{3}{2} \cdot \sqrt{(s-1)(t-1)}$, which contradicts \cref{eq:opNormBBound}, as $\eps \leq \frac 12$. If $\delta < 1$, then we observe that $\frac{1}{1 + \frac{\delta}{2} - \frac{1}{2} \sqrt{2 \delta + \delta^2}} \geq 1 + \frac{1}{2}\sqrt{\delta}$. Hence, $1 + \eps \geq \frac{\rho(B)}{\sqrt{(s-1)(t-1)}}\ge 1 + \frac{1}{2}\sqrt{\delta}$, which implies that $\delta \leq 4\eps^2$.
	\end{proof}
	
	In light of \cref{clm:MtoA,clm:BtoM}, in order to finish the proof of \cref{mthm:singvalue}, it suffices to prove the following lemma.
	\begin{lemma}\label{lem:nomadicspec}
		Let $A$ be uniformly sampled from $\cM_{m,n,s,t}$, and let $B$ be the respective nomadic walk matrix. Then, $\rho(B) \leq (1 + \eps)\sqrt{(s-1)(t-1)}$ with high probability, for some $\eps \le o(1)$.
	\end{lemma}
	The rest of this section is devoted to the proof of \cref{lem:nomadicspec}. This requires some definitions.
	\begin{definition}[Hikes of various kinds]
		A $2\ell$-\emph{hike} $\cH$ in $G$ is a closed walk of length exactly $4\ell$, starting in $V_R$, which is non-backtracking except possibly between the $2\ell$-th and $(2\ell+1)$-th steps. We let $\sign(\cH)$ denote the product of $\sign(e)$ over all edges in $\cH$, with multiplicity. We say that a hike is \emph{even} if all edges are traversed with even multiplicity, and \emph{singleton-free} if no edge is traversed exactly once. Finally, we say that a hike with edges $\inparen{e_1, e_2, \dotsc,  e_{4 \ell}}$ is \emph{special} if $e_1 = e_{4\ell}$, $e_2 = e_{4\ell - 1}$, $e_{2\ell} = e_{2 \ell + 1}$, and $e_{2\ell - 1} = e_{2\ell + 2}$, namely, the last two steps are the reverse of the first two, and the $(2\ell+1)$-th and $(2\ell+2)$-th steps are the reverse of the two steps preceding them. Note that an odd-numbered step is always from $V_R$ to $V_L$, and vice-versa for an even numbered step.
	\end{definition}
	\begin{proof}[Proof of \cref{lem:nomadicspec}]
		Fix $\ell \in \N$ to be determined later. Let $T \defeq \tr\inparen{B^{\ell} \cdot (B^{\top})^{\ell}}$, and note that
		\begin{equation}\label{eq:traceBoundsOp}
		    \rho(B)^{2\ell} \leq \rho(B^{\ell})^2 \leq \norm{B^{\ell}}_2^2 = \norm{B^{\ell} (B^{\top})^{\ell}}_2 \le T \enspace.
		\end{equation}
		We turn to bounding $T$. First, note that
		\begin{align*}
		    \inparen{B^\ell}\inbrak{(f_1,f_2),(g_1,g_2)} &= \sum_{\substack{(e_1,e_2),(e_3,e_4)\dotsc,(e_{2\ell+1},e_{2\ell+2})\in \nomadicv{G}\\(e_1,e_2)=(f_1,f_2),~(e_{2\ell+1},e_{2\ell+2})=(g_1,g_2)}} \prod_{i=1}^\ell B\inbrak{(e_{2\ell-1},e_{2\ell}),(e_{2\ell+1},e_{2\ell+2})} 
		    \\ &= \sum_{\substack{\inparen{e_1, e_2, \dotsc,  e_{2\ell+2}}\\\text{ is a non-backtracking walk in }G\text{ and}\\ e_1=f_1, e_2=f_2, e_{2\ell+1}=g_1,e_{2\ell+2}=g_2}} \prod_{i=3}^{2\ell+2}\sign\inparen{e_i} \enspace,
		\end{align*}
		for $(f_1,f_2),(g_1,g_2)\in  \nomadicv{G}$. Similarly,
		\begin{align*}
		    \inparen{(B^{\top})^\ell}\inbrak{(g_1,g_2),(f_1,f_2)}  &= \sum_{\substack{\inparen{e_1, e_2, \dotsc,  e_{2\ell+2}}\\\text{ is a non-backtracking walk in }G\text{ and}\\ e_1=g_2, e_2=g_1, e_{2\ell+1}=f_2,e_{2\ell+2}=f_1}} \prod_{i=1}^{2\ell}\sign\inparen{e_i} \enspace.
		\end{align*}		
		
		Hence,
		\begin{align*}
		    T &= \sum_{(f_1,f_2),(g_1,g_2)\in \nomadicv{G}} \inparen{B^\ell}\inbrak{(f_1,f_2),(g_1,g_2)}\inparen{(B^{\top})^\ell}\inbrak{(g_1,g_2),(f_1,f_2)}  \\
		    &= \sum_{(f_1,f_2),(g_1,g_2)\in \nomadicv{G}} \sum_{\substack{\inparen{e_1, e_2, \dotsc , e_{2\ell+2}}\\ \text{and }\\ \inparen{e'_1, e'_2, \dotsc , e'_{2\ell+2}}\\ \text{ are non-backtracking walks in }G\text{ and}\\ e_1=e'_{2\ell+2}=f_1, e_2=e'_{2\ell+1}=f_2,\\ e_{2\ell+1}=e'_{2}=g_1, e_{2\ell+2}=e'_{1}=g_2}}\prod_{i=3}^{2\ell+2}\sign(e_i)\prod_{i=1}^{2\ell}\sign(e'_i) \\
		    &= \sum_{\substack{\inparen{e_1, e_2, \dotsc , e_{2\ell+2}}\\ \text{and }\\ \inparen{e'_1, e'_2, \dotsc , e'_{2\ell+2}}\\ \text{ are non-backtracking walks in }G\text{ and}\\ e_1=e'_{2\ell+2}, e_2=e'_{2\ell+1},\\ e_{2\ell+1}=e'_{2}, e_{2\ell+2}=e'_{1}}}\prod_{i=3}^{2\ell+2}\sign(e_i)\prod_{i=1}^{2\ell}\sign(e'_i) \enspace.
		\end{align*}
		 Observe that the each sequence $(e_1, e_2, \dots e_{2\ell+2}, e'_1, e'_2, e'_3, \dots, e'_{2\ell+2})$ in the above sum is a special $(2\ell + 2)$-hike $\cH$ in $G$. Moreover, 
		 \begin{flalign*}
		 &\sign(\cH) = \sign(e_1)\sign(e_2)\sign(e'_{2\ell+1}) \sign(e'_{2\ell+2})\left(\prod_{i=3}^{2\ell+2}\sign(e_i)\prod_{i=1}^{2\ell}\sign(e'_i)\right) = \prod_{i=3}^{2\ell+2}\sign(e_i)\prod_{i=1}^{2\ell}\sign(e'_i) \enspace,
		 \end{flalign*}
        as $e_1 = e'_{2\ell+2}$ and $e_2 = e_{2\ell+1}$. Hence, 
        \begin{equation*}
        T = \sum_{\cH\text{ is a special $2(\ell+1)$-hike in $G$}} \sign(\cH) \enspace.
        \end{equation*}

        We proceed by conditioning on $G$. When $G$ is fixed, each $\sign(e)$ for $e \in E$ is drawn independently from $\Fits$. Thus, we see that 
        $\EE_{\sign}[\sign(\cH) \mid G] = 1$ if every edge in $\cH$ appears with even multiplicity, and otherwise the expectation is $0$. So, we have
        \begin{equation*}
        \EE_{\sign}[T \mid G] = \text{\# of even special $2(\ell+1)$-hikes in $G$} \enspace.
        \end{equation*}
		We can upper bound the latter by $(st)^2 \cdot \inabs {\ehikes {2(\ell-1)}(G)}$, where $\ehikes{2(\ell-1)}(G)$ is the set of even $2(\ell-1)$-hikes in $G$. Indeed, this is because every even special $2(\ell + 1)$-hike can be formed from an even $2(\ell - 1)$ hike by attaching $2$ steps at the beginning and their reverse at the end (at most $st$ choices in total) and by attaching $2$ steps and their reverse in the middle (at most $st$ choices again). Let $\eps = o_{n\to\infty}(1)$ be a function to be chosen later, and let $a= \inparen{1+\eps}\sqrt{(s-1)(t-1)}$. By \cref{eq:traceBoundsOp} and Markov's inequality,
		\begin{flalign*}
		&\PROver{\sign}{\rho(B)\ge a\mid G}  \le \PROver{\sign}{T\ge a^{2\ell}\mid G} \\ 
		&\le \frac{\Eover{\sign}{T\mid G}}{a^{2\ell}} \le \frac{(st)^2\cdot \inabs{\ehikes{2(\ell-1)}(G)}}{a^{2\ell}} \enspace.
		\end{flalign*}
		Thus, to prove the lemma, it suffices to show that, for some $\ell$ and $\eps$ of our choice, $$(st)^2\cdot \inabs{\ehikes{2(\ell-1)}(G)}\le o\inparen{a^{2\ell}}$$ with high probability over the choice of $G$. Therefore, the following lemma implies \cref{lem:nomadicspec}.
		\begin{lemma}\label{lem:hikecount}
			Let $\ell = \floor{\log_2^2 m}$. Then with high probability over the choice of $G$, the number of even $2\ell$ hikes in $G$ is at most $m (s-1)^{\ell} (t-1)^{\ell} (1 + o(1))^{\ell}$.
		\end{lemma}
	We prove \cref{lem:hikecount} in the next two subsections. We break the proof into two cases, according to whether the sparsity $s$ of $G$ is $O(\log^c n)$ for some constant $c$, or $\omega(\log^c n)$ for every constant $c$. This then finishes the proof of \cref{lem:nomadicspec}, and thus also \cref{mthm:singvalue}.
	\end{proof}

	%%%%%%%%%%%%%%%%%%%%%%%%%%%%%%%%%%%%%%%%%%%%%%%%%%%%%%%%%%%%%%%%%%%%%%%%%%%%%%%%
	%%%%%%%%%%%%%%%%%%%%%%%%%%%%%%%%%%%%%%%%%%%%%%%%%%%%%%%%%%%%%%%%%%%%%%%%%%%%%%%%
	\subsection{Counting hikes when $s \leq \polylog(n)$}
	\label{sec:hikecountsparse}
	We prove \cref{lem:hikecount} when $s \leq \log^c n$ for some $c > 0$ by showing the following lemma.
	\begin{lemma}
	\label{lem:hikecountsparse}
		Fix $\ell \in \N$. Suppose that $G$ is bicycle-free at radius $r$ for some $r \geq 20 \ln \ell$. Then the number of singleton-free $2\ell$-hikes in $G$ is at most 
		\begin{equation*}
			O(\ell^4 m) [(s-1)(t-1)]^{\ell} (2r\ell (s-1)(t-1))^{O(\frac{\log \ell}{r}) \cdot \ell} \enspace.
		\end{equation*}
	\end{lemma}
	
	\begin{proof}[Proof of \cref{lem:hikecount} for $s \leq \polylog(n)$ from \cref{lem:hikecountsparse}]
	Let $c>0$ so that $s \leq \log^c n$. By \cref{prop:goodGraph}, it holds with high probability that $G$ is bicycle-free at radius $r$ for $r = c' \log n/\log s$ for some absolute constant $c'$. Assume that this event holds. As $\ell \leq \log_2^2 m \leq \log_2^2 n$, it follows that $20 \ln \ell = O(\log \log n)$ and $r = \Omega(\log n/\log \log n)$, so that $r \geq 20 \ln \ell$ for $n$ sufficiently large. Hence, by \cref{lem:hikecountsparse}, we conclude that the number of singleton-free $2\ell$-hikes in $G$, and in particular the number of even $2\ell$-hikes in $G$, is at most
	\begin{flalign*}
				O(\ell^4 m) [(s-1)(t-1)]^{\ell} (2r\ell (s-1)(t-1))^{O(\frac{\log \ell}{r}) \cdot \ell} \enspace.
	\end{flalign*}
	We clearly have that $O(\ell^4) = (1+o(1))^{\ell}$. We have shown that $\log \ell/r = O((\log ''\log n)^2/\log n)$,  and we have $2 r \ell (s-1)(t-1) \leq \log^{c''} n$ for some constant $c''$. We thus have that
	\begin{flalign*}
	(2r\ell (s-1)(t-1))^{O(\frac{\log \ell}{r})} = \exp{\poly(\log\log n)/\log n} = 1 + o(1) \enspace,
	\end{flalign*}
	which shows that the number of even $2\ell$-hikes is at most 
	\begin{equation*}
	m (1+o(1))^{\ell} [(s-1)(t-1)]^{\ell} \enspace,
	\end{equation*}
	as required.
	\end{proof}
	
	\begin{proof}[Proof of \cref{lem:hikecountsparse}]
	We now turn to proving \cref{lem:hikecountsparse}. Recall that since $G$ is bicycle-free at radius $r$, the subgraph $\ball_G(v, r)$ of $G$ contains at most one cycle for every $v \in V_L$. We will count the number of even $2\ell$-hikes via an encoding argument.
	
	Let us first consider a fixed even $2\ell$-hike $\cH$. For this hike, we let $G_{\cH}$ denote the subgraph of $G$ consisting of the edges traversed by $\cH$. For intuition, one should think of the graph $G_{\cH}$ as being ``discovered'' by the decoding algorithm as it follows the hike $\cH$.
	
	We classify each step of $\cH$ as either \emph{fresh}, \emph{stale}, or \emph{boundary}, as follows:
	\begin{itemize}
	    \item If a step of $\cH$ traverses a new edge and it steps into a previously unexplored vertex, it is \emph{fresh}.
	    \item If a step of $\cH$ traverses a new edge but it steps into a previously explored vertex, then it is \emph{boundary}.
	    \item If a step of $\cH$ traverses a previously used edge, then it is \emph{stale}.
	\end{itemize}
	We note that a singleton-free $2\ell$-hike must have at most $2\ell$ fresh steps, as every edge must be traversed at least twice. Next, we observe that the number of boundary steps is exactly $\abs{E(G_{\cH})} + 1 - \abs{V(G_{\cH})}$. As $\abs{V(G_{\cH})} \leq 4 \ell$, we have $r \geq 10 \ln (4 \ell)$. By \cite[Theorem 2.13]{MohantyOP20a}, we thus have that the number of boundary steps is at most
	\begin{flalign*}
	\abs{E(G_{\cH})} + 1 - \abs{V(G_{\cH})} \leq \frac{1 + \ln \abs{V(G_{\cH})}}{r} \cdot \abs{V(G_{\cH})} \leq O(\frac{\log \ell}{r}) \cdot \ell \enspace.
	\end{flalign*}
	
	We group the stale steps into ``stretches'' of contiguous stale steps. First, we can group the stale steps into maximal contiguous blocks of stale steps. Note that we have at most $O(\frac{\log \ell}{r}) \cdot \ell$ blocks, as each block of stale steps must follow a boundary step or the midpoint of the hike. If the midpoint of the hike is contained in one block, we split the block into two. This makes the sequence of stale steps in each block non-backtracking. Next, we can divide the blocks so that they each have size at most $r$. As there are trivially at most $4 \ell$ stale steps, this can only increase the number of blocks by at most $\frac{4\ell}{r} = O(\frac{\log \ell}{r}) \cdot \ell$ (as the ``worst case'' is when all blocks have size exactly $r + 1$). We have thus argued that we can divide the stale steps into $O(\frac{\log \ell}{r}) \cdot \ell$ blocks of contiguous stale steps, each of size at most $r$, and all steps in each block are non-backtracking.
	
	We now describe the encoder and decoder simultaneously. For any hike $\cH$, we let the \emph{shape} of $\cH$ be the sequence $\sigma \in \{\fresh,\bound,\stale\}^{4\ell}$ that indicates which steps are ``fresh'', ``boundary'', and ``stale''. The first part of the encoding is the following sequence of numbers:
	\begin{itemize}
	\item $c_{\fresh}$, the number of times $\fresh$ appears in $\sigma$,
	\item $c_{\bound}$, the number of times $\bound$ appears in $\sigma$,
	\item $c_{\stale}$, the number of ``stale stretches''. That is, the number of blocks of contiguous stale steps where each block has size at most $r$ and the midpoint of the hike does not occur inside the block.
	\end{itemize}
	There are naively at most $O(\ell^3)$ choices for these numbers.
	
	Next, we let $\sigma'$ denote the shape $\sigma$, only we compress a stale stretch into one symbol $\stale$. We call $\sigma'$ the \emph{compressed shape} of $\cH$. The next part of the encoding is $\sigma'$. Note that there are naively $O(\ell) \cdot \ell^{O(\frac{\log \ell}{r}) \cdot \ell}$ choices of $\sigma'$ (once $c_{\fresh}$, $c_{\bound}$ and $c_{\stale}$ are chosen), as we can specify its length ($O(\ell)$ choices), and then the locations of the $\bound$ and $\stale$ symbols ($\ell$ choices for each symbol, and there are $O(\frac{\log \ell}{r}) \cdot \ell$ symbols).
	
	Now, the encoder specifies the start vertex of $\cH$ ($m$ choices). For each $\fresh$ and $\bound$ step, the encoder specifies which neighbor\footnote{We fix some arbitrary order for the neighbors of each vertex.} the hike moves to ($(s-1)$ choices if we start at a right vertex, and $(t-1)$ if we start at a left vertex\footnote{Note that there are $s$ choices for the first step of the hike, but we will count this as $s-1$ and instead pay an extra factor of $\frac{s}{s-1} \leq 2$ in the final bound.}). A stale stretch is encoded as follows. Let $u$ be the start vertex of the stretch, and let $v$ be the end vertex of the stretch. Let $K = \ball_G(u, r)$, which is bicycle-free. The stale stretch is a path of length $\leq r$ from $u$ to $v$ in $K$, so it can be specified by: \begin{inparaenum}[(1)]
	\item specifying the end vertex $v$ (at most $4\ell$ choices, and is specified by giving the first step\footnote{Note that since the step is \emph{stale}, the hike must have reached $v$ at some earlier point.} at which $v$ is reached in the hike),
	\item specifying the number of the times the (unique) cycle in $K$ is traversed by the path (which is an integer $\leq \frac{r}{4}$), and
	\item specifying which direction\footnote{The cycle can only be traversed in one direction, as $K$ has only one cycle and each stale stretch is non-backtracking.} the cycle is traversed in (at most $2$ choices)
	\end{inparaenum}.
	For each stale stretch, the encoder specifies the above information.
	
	Above, we have argued that the information given by the encoder is enough to reconstruct the hike $\cH$ uniquely. This implicitly defines the decoder. It therefore remains to upper bound the total number of strings that can be outputted by the encoder.
    	
	Let $\cH$ be a singleton-free $2\ell$-hike with shape $\sigma$, and let $c_{\fresh}$, $c_{\bound}, c_{\stale}$ be the counts defined earlier (which are fixed by the shape $\sigma$). We observe that the number of singleton-free $2\ell$-hikes with shape $\sigma$ is at most
	\begin{flalign*}
	 2m (s-1)^{c_{\fresh}^{(L)} + c_{\bound}^{(L)}} \cdot (t-1)^{c_{\fresh}^{(R)} + c_{\bound}^{(R)}} \cdot (2 r \ell)^{c_{\stale}}\enspace,
	\end{flalign*}
	where $c_{\fresh}^{(L)}, c_{\bound}^{(L)}$ are the number of fresh/boundary steps taken that start at a left vertex, and similarly for $c_{\fresh}^{(R)}, c_{\bound}^{(R)}$. Note that this is well-defined, as $c_{\fresh}^{(L)}, c_{\bound}^{(L)}, c_{\fresh}^{(R)}, c_{\bound}^{(R)}$ are determined by the shape $\sigma$.
	
	We next observe that every $c_{\fresh}^{(R)}$ step is followed by either a $c_{\fresh}^{(L)}$ or $c_{\bound}^{(L)}$ step, so $c_{\fresh}^{(R)} \leq c_{\fresh}^{(L)} + c_{\bound}^{(L)}$, and similarly $c_{\fresh}^{(L)} \leq c_{\fresh}^{(R)} + c_{\bound}^{(R)}$. Thus, we have
	\begin{flalign*}
		&\abs{c_{\fresh}^{(R)} - c_{\fresh}^{(L)}} \leq  c_{\bound}^{(L)} + c_{\bound}^{(R)} = c_{\bound} \\
		&\implies \abs{c_{\fresh}^{(R)}  - \frac{c_{\fresh}}{2}} \leq \frac{c_{\bound}}{2} \text{ and } \abs{c_{\fresh}^{(L)}  - \frac{c_{\fresh}}{2}} \leq \frac{c_{\bound}}{2} \enspace.
	\end{flalign*}
	Hence, the number of singleton-free $2\ell$-hikes with shape $\sigma$ is at most 
	\begin{equation}
	 2m (s-1)^{\frac{1}{2}c_{\fresh} + \frac{3}{2} c_{\bound}} \cdot (t-1)^{\frac{1}{2}c_{\fresh} + \frac{3}{2}c_{\bound}} \cdot (2 r \ell)^{c_{\stale}}\label{eq:hikecountsparsebound}\enspace.
	\end{equation}
	As the above bound is determined completely by the compressed shape $\sigma'$, we see that it upper bounds the number of singleton-free $2\ell$-hikes with compressed shape $\sigma'$.
	
	Recall that the number of valid compressed shapes $\sigma'$ is at most $O(\ell^4) \cdot \ell^{O(\frac{\log \ell}{r}) \cdot \ell}$. Because \cref{eq:hikecountsparsebound} is an increasing function in $c_{\fresh}$, $c_{\bound}$ and $c_{\stale}$, we can upper bound the total number of singleton-free $2\ell$-hikes by $(\text{\# of choices of $\sigma'$}) \cdot (\text{maximum value of \cref{eq:hikecountsparsebound}})$. We thus have that the number of singleton-free $2\ell$-hikes is at most
	\begin{flalign*}
	 &O(\ell^4) \ell^{O(\frac{\log \ell}{r}) \cdot \ell} \cdot 2m (s-1)^{\ell + O(\frac{\log \ell}{r}) \cdot \ell} \cdot (t-1)^{\ell + O(\frac{\log \ell}{r}) \cdot \ell} \cdot (2 r \ell)^{O(\frac{\log \ell}{r}) \cdot \ell}\\
	 &= O(m \ell^4) [(s-1) (t-1)]^{\ell} \cdot (2 r \ell(s-1)(t-1))^{O(\frac{\log \ell}{r}) \cdot \ell} \enspace,
	\end{flalign*}
	which finishes the proof of \cref{lem:hikecountsparse}.
	\end{proof}

	%%%%%%%%%%%%%%%%%%%%%%%%%%%%%%%%%%%%%%%%%%%%%%%%%%%%%%%%%%%%%%%%%%%%%%%%%%%%%%%%
	%%%%%%%%%%%%%%%%%%%%%%%%%%%%%%%%%%%%%%%%%%%%%%%%%%%%%%%%%%%%%%%%%%%%%%%%%%%%%%%%
	\subsection{Counting hikes when $s = \omega(\polylog(n))$}
	\label{sec:hikecountdense}
	\begin{lemma}
	\label{lem:hikecountdense}
	Let $s = \omega(\log^c n)$ for every constant $c$, let $\ell = \floor{\log_2^2 m}$, and let
	$G$ be an arbitrary $(t,s)$-biregular graph. Then the number of even $2\ell$-hikes in $G$ is at most
	\begin{equation*}
	m \cdot (s-1)^{\ell} \cdot (t-1)^{\ell} \cdot \left(1 + \frac{\ell^{O(1)}}{t}\right) \enspace.
	\end{equation*}
	\end{lemma}
	Note that as $\ell = O(\log^2 m) = O(\log^2 n)$ and $t = \Omega(s) = \omega(\log^c n)$ for every constant $c$, it follows that $\frac{\ell^{O(1)}}{t}  = o(1)$, so \cref{lem:hikecountdense} implies \cref{lem:hikecount} when $s = \omega(\polylog(n))$.
	
	\begin{proof}[Proof of \cref{lem:hikecountdense}]
	We begin by introducing some terminology. Let $\cH$ be an even $2\ell$- hike in $G$, and let $G_{\cH}$ be the subgraph of $G$ induced by the edges in $\cH$. As in the case for $s \leq \polylog(n)$, one should think of the graph $G_{\cH}$ as being ``discovered'' as we traverse the hike $\cH$.
	
	For an edge $e$, we let $a_e$ denote the multiplicity of the edge $e$ in $\cH$, i.e., the number of times that it is traversed (in either direction). We say that $e$ is high multiplicity if $a_e  > 2$. 
	
	A step in $\cH$ is a traversal of an edge in $\cH$. That is, a step is an edge in $\cH$, but we count different occurrences of an edge in $\cH$ as distinct steps. For a step in $\cH$, we say that the endpoint of the step is the vertex where the hike ``ends up'' after this step.	
	
	We say that a step in the hike is high multiplicity (labeled as $\highmult$) if it uses a high multiplicity edge. We say that a non-$\highmult$ step in the hike is fresh (labeled as $\fresh$) if it traverses a new edge for the first time and the endpoint is a previously unvisited vertex. If the endpoint is instead a vertex that has been previously visited, we call it boundary (labeled as $\bound$). We note that any non-high multiplicity edge is traversed exactly twice in $\cH$.  We say that a non-$\highmult$ step is forced (labeled as $\forced$) if at that point in the hike, the current vertex has exactly one non-high multiplicity neighboring edge that has not yet been traversed twice. All other steps in the hike are unforced (labeled as $\unforced$).
	
	For an edge $e$ in $\cH$ with $a_e = 2$, we call the edge ``fresh'' if when $e$ is traversed by $\cH$ for the first time, the step is fresh, and otherwise we call $e$ ``boundary''.
	
    \begin{remark}
    We note that the definitions above differ slightly from those in \cref{sec:hikecountsparse}. For instance, here we separate out the $\highmult$ steps into their own class, and only split non-$\highmult$ steps into $\fresh$ and $\bound$ steps.
    \end{remark}
	
	For a hike $\cH$, we define the following quantities:
	\begin{itemize}
		\item $E(\cH)$, the number of distinct edges in $\cH$,
		\item $n(\cH)$, the number of distinct left vertices in $\cH$,
		\item $m(\cH)$, the number of distinct right vertices in $\cH$,
		\item $h(\cH)$, the number of distinct high multiplicity edges in $\cH$,
		\item $f(\cH)$, the number of distinct fresh edges in $\cH$,
		\item $b(\cH)$, the number of distinct boundary edges in $\cH$,
		\item $u(\cH)$, the number of unforced steps in $\cH$.
		\item $\tau(\cH)$, the type of $\cH$, which the tuple $(h(\cH), f(\cH), b(\cH), u(\cH))$, along with a string $\{\highmult, \fresh, \bound, \forced, \unforced\}^{4 \ell}$ that specifies the type of each step in $\cH$.
	\end{itemize}
	Finally, we say that a type $\tau$ is \emph{valid} if there exists a hike $\cH$ of type $\tau$.
	
	We first prove some basic facts about the above quantities.
	\begin{claim}
		For all $\cH$, $n(\cH) \leq \ell$ and $m(\cH) \leq \ell + 1$.
	\end{claim}
	\begin{proof}
		Consider a vertex $v$ traversed in the hike $\cH$ that is not the midpoint. We claim that $v$ must appear at least twice in the hike. (The start and end of the hike, which are the same vertex, are each counted separately.) Indeed, this is because any edge $e = (u,v)$ that is first traversed as $u \to v$ must be traversed at least twice, and cannot be traversed via the path $u \to v \to u$ because $\cH$ is non-backtracking. We note that since $\cH$ has length $4 \ell$, the midpoint of $\cH$ is a right vertex, and $\cH$ traverses $2\ell$ left vertices and $2\ell + 1$ right vertices (counting multiplicities). Since each vertex counted in $n(\cH)$ appears at least twice, it follows that $2 n(\cH) \leq 2 \ell$. Since each vertex except for one (the midpoint) counted in $m(\cH)$ appears at least twice, it follows that $1 + 2(m(\cH) - 1) \leq 2 \ell + 1$, and so $m(\cH) \leq \ell + 1$.
	\end{proof}
	
	\begin{claim}
		\label{claim:tedious}
		$f(\cH) + h(\cH) + b(\cH) = E(\cH)$, $h(\cH) \leq 2 \ell - E(\cH)$, and $u(\cH) \leq 4 \ell - 2 f(\cH)$.
	\end{claim}
	\begin{proof}
		We observe that each distinct edge is counted exactly once in $f(\cH),h(\cH),b(\cH)$, so the first equality holds.
		
		For the first inequality, we observe that $4h(\cH) \leq \sum_{e \in \cH : a_e > 2} a_e = 4 \ell - \sum_{e : \in \cH : a_e  = 2} a_e = 4 \ell - 2 (f(\cH) + b(\cH))$, where we use the fact that $a_e$ is always even. This implies that $h(\cH) \leq 2 \ell - (f(\cH) + b(\cH) + (\cH)) = 2 \ell - E(\cH)$.
		
		To show the final inequality, we first make the following definition. We call a ``return step'' a step that is either $\unforced$ or $\forced$, and we call a non-high multiplicity edge ``available'' if it has been traversed exactly once.
		
		Let $v$ be any vertex in $\cH$ that is not the start vertex. We keep track of a counter $c_v$, that counts the number of available edges at $v$ every time $\cH$ is at $v$.
		When the hike first arrives at $v$, it does so via a $\fresh$ or $\highmult$ step, and so $v$ has at most one available edge. So, $c_v$ starts at either $0$ or $1$. Every time the hike leaves $v$ via a return step, it decreases the number of available edges by $1$. So, $c_v$ decreases by $1$ in this case, and moreover the step must be $\forced$ if $c_v = 1$ prior to the return step.
		Every time the hike leaves $v$ via a non-return step that is also a non-high multiplicity edge, it adds $1$ to the number of available edges adjacent to $v$, and every time the hike returns to $v$ it will remove an available edge, unless it returns to $v$ via a high multiplicity or boundary edge. This means that the $c_v$ increases by at most $1$ (resp.\ $2$) every time $\cH$ returns to $v$ using a high multiplicity edge (resp.\ boundary edge), and cannot increase otherwise. 
		
		By the above, it follows that the number of $\unforced$ steps from $v$ is at most $-1 + \#$ of times $c_v$ increases. We note that $\#$ of times $c_v$ increases is $\leq \#$ of times a high multiplicity edge enters $v$ $+ 2 \cdot \#$ of times a boundary edge enters $v$. Summing over all vertices, we see that $u(\cH) \leq \sum_{e : a_e > 2} a_e + 2 b(\cH)$. Since $\sum_{e : a_e > 2} = 4 \ell - 2 (f(\cH) + b(\cH))$, we are done.
		\end{proof}

	We are now ready to bound the number of even $2\ell$-hikes $\cH$. Fix ordering on the vertices of $G$, so that if we have a vertex $u$, then the neighbors of $u$ are numbered uniquely from $1, \dots, s-1$ (if $u$ is a left vertex) or $1, \dots, t-1$ (if $u$ is a right vertex).
	
	First, we consider the case when $E(\cH) = 2\ell$ and $b(\cH) = 0$. This will be the dominant term. By \cref{claim:tedious}, it follows that $h(\cH) = 0$, so $\cH$ has no high multiplicity edges. Since $b(\cH) = 0$, we see that $f(\cH) = E(\cH) = 2 \ell$. Now, we consider the sequence of steps made by $\cH$. The first step is $\fresh$, and a $\fresh$ step must always be followed by a $\highmult$, $\fresh$, or $\bound$ step, except at the midpoint of the hike. Because $h(\cH) = b(\cH) = 0$, this means that all steps must be $\fresh$ until the midpoint. But then this means that there can be no more $\fresh$ steps, as the number of $\fresh$ steps is at most $2\ell$, and so all the remaining steps must be $\forced$. We can count the number of such hikes by \begin{inparaenum}[(1)] \item picking the start vertex in $V_R$ ($m$ choices), and \item for every fresh step, picking the neighbor of the current vertex to move to\end{inparaenum}. We see that there are exactly $\ell$ fresh steps starting from a right vertex and $\ell$ fresh steps starting from a left vertex, so the total number of such hikes is at most
	\begin{equation}
	m (s-1)^{\ell} (t-1)^{\ell} \enspace. \label{eq:dominantterm}
	\end{equation}
	
	Now, we assume that $E(\cH) < 2\ell$ or $b(\cH) > 0$. Fix a type $\tau$ with either $E < 2 \ell$ or $b > 0$. We bound the number of hikes $\cH$ where $\tau(\cH) = \tau$ via an encoding argument. 
	
	The encoding of $\cH$ is as follows.
	\begin{itemize}
		\item For each high multiplicity edge $(u,v)$ (where the edge is first traversed as $u \to v$ in $\cH$), specify which neighbor $v$ of $u$ is the endpoint of the edge.
		\item Specify the start vertex.
		\item For each $\highmult$ step, specify which of the $h(\cH)$ edges is the high multiplicity edge being used at this step.
		\item For each fresh step $\fresh$, specify which neighbor of the current vertex is the endpoint of the edge.
		\item For each boundary step $\bound$, we specify the first location in the hike of the vertex that is the endpoint of this edge.
		\item For each unforced step $\unforced$, specify the first location in the hike of the vertex that is the endpoint of this edge.
		\item For each forced step $\forced$, we do nothing.
	\end{itemize}
	We now specify the decoder. That is, we show how to uniquely construct the hike $\cH$ from the above data. Indeed, suppose we have reconstructed the hike $\cH$ correctly for the first $i$ steps, ending at vertex $u$. (Note that the base case is trivial, as we are given the start vertex and so we can reconstruct $\cH$ after $0$ steps.) Then, if the $(i+1)$-th step is $\highmult$, we know which high multiplicity edge will be traversed, and so we know which neighbor of $u$ to move to.\footnote{This technically holds only for the first traversal of the high multiplicity edge, as later traversals may be in the other direction. However, after the first traversal we have constructed the edge $(u,v)$, so this can be done trivially on later steps.}
		If the step is $\fresh$, then we also know which neighbor to move to. If the step is $\bound/\unforced$, then we know that we move to the vertex that we were at in step $j < i$, so we can correctly reconstruct this step. Finally, if the step is $\forced$, then there is only one possible edge that can be traversed (and moreover, we can easily find out which edge this is, by simply keeping track of the number of times each ``seen'' edge has been traversed), so we can also reconstruct this step. Thus, one can uniquely reconstruct the hike $\cH$ from the encoding.
	
	We now count the number of hikes $\cH$ with $\tau(\cH) = \tau$. We observe that the total number of vertices specified (other than the start vertex) is $f + h$. When choosing to move to a right vertex, there are $(t-1)$ choices (except for the start vertex, where there are $m$ choices), and when choosing to move to a left vertex there are $(s-1)$ choices. So, we see that there are at most $m \cdot (s-1)^{\ell} \cdot (t-1)^{f + h - \ell}$ possibilities, using the fact that $n(\cH) \leq \ell$ always. This accounts for the choice of vertices from the $\highmult$ and $\fresh$ steps.
	
	For each $\highmult$ step, we specify which of the $h$ edges we use. As $h \leq 4 \ell$ naively, this is at most $4\ell$ choices per $\highmult$ step. Each $\bound$ step and $\unforced$ step has naively at most $4 \ell$ choices, as we traverse at most $4\ell$ distinct vertices. So, the total number of even $2\ell$-hikes $\cH$ of type $\tau$ is therefore at most
	\begin{equation*}
	\label{eq:tempbound}
		m \cdot (s-1)^{\ell} \cdot (t-1)^{f + h  - \ell} \cdot (4\ell)^{h + b + u} \enspace.
	\end{equation*}
	
	Next, we observe that for a fixed $f, h,b,u$, there are at most $(4 \ell) (4\ell)^{2(h + b + u)}$ valid types $\tau$ (namely, types $\tau$ corresponding to at least one $\cH$). Indeed, we only need to specify the string of $\{\highmult, \fresh, \bound, \forced, \unforced\}^{4 \ell}$, which can be done as follows. We choose the locations of the $\highmult$, $\bound$, and $\unforced$ steps, of which there are naively $(4\ell)^{h + b + u}$ choices. This gives us a partially filled in string, with the locations of the $\fresh$ and $\forced$ steps still undetermined. However, we observe that we cannot have a $\forced$ step after a $\fresh$ step (except at the midpoint), and so we can determine these steps by specifying the number of $\forced$ steps in each ``unfilled gap'' in the string, where we split the gap containing the midpoint (if it exists) into two gaps. There are at most $h + b + u + 2$ such gaps, as there are $h + b + u + 1$ unfilled gaps, and we split the gap with the midpoint into two. So, we have at most $(4\ell)^{h + b + u + 2}$ choices here.
	
	Thus, for a fixed $f, h, b, u$, by \cref{eq:tempbound} there are at most
	\begin{flalign*}
		m \cdot (s-1)^{\ell} \cdot (t-1)^{f + h  - \ell} \cdot (4\ell)^{h + b + u} \cdot (4 \ell) (4\ell)^{2(h + b + u)} 
	\end{flalign*}
	 even $2\ell$-hikes with $f(\cH) = f$, $h(\cH) = h$, $b(\cH) = b$ and $u(\cH) = u$. 
	
	By \cref{claim:tedious}, the total number of $\cH$ with $E(\cH) = E$ and $b(\cH) = b$ is therefore at most 
	\begin{flalign*}
		m \cdot (s-1)^{\ell} \cdot (t-1)^{E - b - \ell} \cdot (4\ell)^{15 (2 \ell - E) + 9b} \cdot (4 \ell)^3 \enspace,
	\end{flalign*}
	where we multiply by $(4\ell)^2$ to account for the choices of $h$ and $u$.
	Now, to get the total number of hikes $\cH$ with either $E(\cH) < 2 \ell$ or $b(\cH) > 0$, we sum over all possible choices of $E$ and $b$. We have $b \leq E$, $E \leq 2 \ell$, and that either $E < 2 \ell$ or $b > 0$. Using the fact that $t = \omega(\log^c m)$ for all constants $c$ and that $\ell = \log^2 m$, we conclude that the total number of such hikes is at most
	\begin{flalign*}
		m \cdot (s-1)^{\ell} \cdot (t-1)^{\ell} \cdot \frac{\ell^{O(1)}}{t} \enspace.
	\end{flalign*}
	Thus, by \cref{eq:dominantterm}, the total number of even $2\ell$ hikes is $m \cdot (s-1)^{\ell} \cdot (t-1)^{\ell} \cdot (1 + \frac{\ell^{O(1)}}{t})$, as required.
	\end{proof}

	\appendix
	
	%%%%%%%%%%%%%%%%%%%%%%%%%%%%%%%%%%%%%%%%%%%%%%%%%%%%%%%%%%%%%%%%%%%%%%%%%%%%%%%%
	%%%%%%%%%%%%%%%%%%%%%%%%%%%%%%%%%%%%%%%%%%%%%%%%%%%%%%%%%%%%%%%%%%%%%%%%%%%%%%%%
	%%%%%%%%%%%%%%%%%%%%%%%%%%%%%%%%%%%%%%%%%%%%%%%%%%%%%%%%%%%%%%%%%%%%%%%%%%%%%%%%

	\section{A weak positive bound on $\ell_2$-spread}
	\label{sec:ell2spreadpos}
	We prove \cref{mprop:ell2spreadpos}, which we recall below.
    \elltwospreadpos*
	We use the following lemma whose proof we defer to the end of this section.
	\begin{lemma}
	\label{lem:expandertoell2spread}
	Let $A \in \cM_{m,n,s,t}$ such that $G_A$ is a $(\gamma, \mu)$-unique expander, where $0 <\mu \le \frac 29$, $0 < \gamma \le 2\mu$ and $\gamma n \ge \inparen{\frac 1\mu}^c$. Then, for every $\gamma n$-sparse $x \in \R^n$, there holds \begin{equation}\label{eq:expandertoell2spreadEq}
	    \twonorm{Ax} \geq c_1\cdot \sqrt{t} \cdot \left( \frac{\sqrt{t}}{\norm{A}_2} \right)^{\frac{c_2\cdot \log (\gamma n)}{\log \frac{1}{\mu}}} \cdot \norm{x}_2\enspace.
	\end{equation}
	\end{lemma}
	\begin{remark}\label{rem:expandertoell2spreadParameters}
	    For every $A\in \cM_{m,n,s,t}$ we have $\twonorm{A}\ge \twonorm{Ae_1} = \sqrt t$. Hence, the term $\frac{\sqrt t}{\twonorm A}$ in the right-hand side of \cref{eq:expandertoell2spreadEq} is at most $1$.
	\end{remark}
	
	\begin{proof}[Proof of \cref{mprop:ell2spreadpos}]
	By \cref{prop:uniqueExpansion,mthm:singvalue}, it holds with high probability, $G_A$ is a $(\gamma, \mu)$-unique expander with $\gamma = \Omega(\alpha^2/t^4), \mu = 2/t$, and also that $\norm{A}_2 \leq 2 \sqrt{s}$. Assume that these events hold.
	
	Let $y\in \R^n$ be $(\gamma,\eps)$-compressible for some $\eps > 0$ to be determined later, and let $x\in \R^n$ be $\gamma n$-sparse with $\twonorm{x-y} \le \eps$. Then,
	\begin{equation}\label{eq:ell2spreadposAyLowerBound}
	\twonorm{Ay} \ge \twonorm{Ax} - \twonorm{A(x-y)} \ge \twonorm{Ax} - \eps\twonorm{A} \ge \twonorm{Ax} - 2\eps \sqrt s.
	\end{equation}
	
	Note that for $n$ large enough, $\gamma$ and $\mu$ satisfy the hypothesis of \cref{lem:expandertoell2spread}. In particular, $\mu \le \frac 29$ by our assumption that $t\ge 9$. Hence, \cref{eq:expandertoell2spreadEq} applies to $x$. It follows that
	$$\twonorm {Ax} \ge c_1 \sqrt t \cdot\inparen{\frac{\sqrt t}{2\sqrt s}}^{\frac{c_2\log \frac{\alpha^2 n}{t^4}}{\log t}}.$$ 
	
	Denote the right-hand side of the above by $a$, so \cref{eq:ell2spreadposAyLowerBound} yields
	$$\twonorm{Ay} \ge a - 2\eps \sqrt s.$$
	Taking $\eps = \frac{a}{4\sqrt s}$, we then have that $\twonorm{Ay} > 0$, and, in particular, $y\notin \ker(A)$. Thus, $\ker(A)$ is $(\gamma,\eps)$-$\ell_2$-spread. The proposition follows since $\eps \ge \alpha^{O\inparen{\frac{\log n}{\log t}}}$.
	\end{proof}
	
	\smallskip
	\begin{proof}[Proof of \cref{lem:expandertoell2spread}]
	Write $G= G_A = (V_L,V_R,E)$. We need the following claim.
	\begin{claim}\label{claim:epandertoell2spreadmatchingpropert}
	Fix $k,b\in \N$ such that $kb\le \gamma n$. Let $S_1,\dots, S_b\subseteq V_L$ be disjoint sets, each of size $k$. Then, there exist sets $T_1,\dots, T_b$, each of size $\ge (1-\mu b)tk$, such that for each $1\le i\le b$, every $r\in T_i$ has exactly one neighbor in $S_i$ and no neighbors in any of the sets $S_j$ ($j\ne i$).
	\end{claim}
	\begin{proof}
	    Let $S = \sqcup_{i = 1}^b S_i$, and let $T_i = U(S) \cap N(S_i)$. Note that $T_1,\dotsc, T_b$ are pairwise disjoint. As $\abs{S} \leq \gamma n$, the unique expansion of $G$ yields $\sum_{i = 1}^b \abs{T_i} = \inabs{\bigsqcup_{i=1}^b T_i} = \abs{U(S)} \geq t(1-\mu) \abs{S} = t(1 - \mu) k b$. Hence, 
	\begin{flalign*}
	\abs{T_i} \geq t(1 - \mu) k b - \sum_{j \ne i} \abs{T_j} \geq t(1 - \mu) k b - t k(b-1) = (1-\mu b)tk\enspace. \qedhere
	\end{flalign*}
	\end{proof}
	
	\medskip
Let $b = \floor{\frac{1}{2\mu}}\ge 2$. Let $M \in \N$ satisfy $b^{M-1}< \gamma n\le b^M$. We partition the interval $\inbrak{b^M}$ into consecutive intervals $R_0, R_1,\dotsc, R_{M}$ where $R_\ell = \inset{b^{\ell-1}+1,\dotsc, b^\ell}$ for $0\le \ell \le M$.

	Let $x$ be a $\gamma n$-sparse vector with $\norm{x}_2 = 1$. Without loss of generality, assume that $\abs{x_1} \geq \abs{x_2} \geq \dots \geq \abs{x_n}$. In particular,
    $\supp(x) \subseteq \inbrak{\floor{\gamma n}} \subseteq  \inbrak{b^M} =  \bigcup_{\ell=0}^{M}  R_\ell \subseteq [n]$. The last inclusion is due to $b^M =b^{M-1} b< \gamma n b \le \frac{\gamma n}{2\mu} \le n$, where the last inequality uses our assumption that $\gamma \le 2\mu$.
	
	For $0\le \ell\le M$, let $z_{\ell} = x_{R_{\ell}}$, i.e., the vector that is equal to $x$ on the set $R_{\ell}$, and is $0$ otherwise. Let $\beta = \frac{t}{32 \norm{A}_2^2}$. Note that $\beta \leq \frac{1}{2}$, due to \cref{rem:expandertoell2spreadParameters}. Consequently,
	\begin{flalign*}
    \sum_{\ell = 0}^{M} \norm{z_{\ell}}_2^2 = \norm{x}_2^2 = 1 \geq \frac{\beta}{1 - \beta} > \sum_{\ell = 0}^{M} \beta^{\ell + 1} \enspace.
	\end{flalign*}
	Hence, there must exist $0\le \ell\le M$ such that $\norm{z_{\ell}}_2 \geq \beta^{\ell+1}$. Let $\ell^*$ denote the largest $\ell$ for which this occurs. 
			
	Let $u = \sum_{\ell = 1}^{\ell^*} z_{\ell}$ and $v = \sum_{\ell = \ell^* + 1}^{M} z_{\ell}$, so that $x = u + v$.
	Because the $z_{\ell}$'s have disjoint support, we have \begin{equation}\label{eq:expandertoell2spreadvUpperBound}\norm{v}_2^2 = \sum_{\ell = \ell^* + 1}^{M} \norm{z_{\ell}}_2^2 \leq \sum_{\ell = \ell^* + 1}^{\infty} \beta^{\ell + 1} \leq 2 \beta^{\ell^* + 2}\enspace,
	\end{equation}
	 where we used that $\beta \le \frac 12$.
	
	We claim that
	\begin{equation}\label{eq:expandertoell2spreadAuLowerBound}
	    \twonorm{Au}^2 \ge \frac t2 \beta^{\ell^*+1}.
	\end{equation}
	Let $S = \bigcup_{\ell=0}^{\ell^*}R_\ell \supseteq \supp(u)$. We consider two cases. First, if $\ell^* = 0$, then $\abs{S} = 1$, so that $\norm{Au}_2^2 = t \norm{u}_2^2 \geq \frac{t}{2} \beta^{2} = \frac{t}{2} \beta^{2(\ell^* + 1)}$, implying \cref{eq:expandertoell2spreadAuLowerBound}. 
	
	Next, suppose that $\ell^* \geq 1$. Note that $\inabs S =  \inabs{\bigcup_{\ell=0}^{\ell^*} R_\ell} = b^{\ell^*}$. 
	Partition $S$ into $b$ consecutive intervals $S_1, \dots, S_b$, each of size $k := \frac{\abs{S}}{b} = b^{\ell^* - 1}$, defined by $S_\ell = \inset{(\ell-1) k+1,\dotsc, \ell k}$.
	
	Note that $S_2, \dots, S_b$ partition $R_{\ell^*}$, and $S_1 = \cup_{\ell = 0}^{\ell^* - 1} R_{\ell}$.
	By \cref{claim:epandertoell2spreadmatchingpropert}, there exist sets $T_1, \dots, T_b$, each of size $\geq (1-\mu b)tk \ge \frac{tk}2$, such that each $r \in T_i$ is in $U(S_i)$ and is not in $N(S_j)$, for all $j \ne i$. 
	Let $\eta_i = \min_{j \in S_i} \abs{u_{j}}$. Then, $\abs{(Au)_r} \geq \eta_i$ for each $r \in T_i$. Moreover, we must also have $\eta_i^2 \geq \norm{u_{S_{i+1}}}_2^2/k$, where $u_{S_{i+1}}$ denotes the restriction of $u$ to the set $S_{i+1}$, as there are $k$ entries in $S_{i+1}$, each with absolute value $\leq \eta_i$. \cref{eq:expandertoell2spreadAuLowerBound} follows since
	\begin{flalign*}
	\norm{A u}_2^2&\ge \sum_{i=1}^b \sum_{r\in T_i}\inabs{(Au)_r}^2 \geq \sum_{i = 1}^b \frac{tk}2 \eta_i^2 \geq \frac{tk}2 \sum_{i = 2}^b \frac{\twonorm{u_{S_{i}}}^2}{k} =\frac t2 \twonorm{u_{\cup_{i = 2}^b S_{i}}}^2 = \frac t 2 \twonorm{z_{\ell^*}}^2 \ge \frac t 2 \beta^{\ell^*+1}\enspace.
	\end{flalign*}
	
	\cref{eq:expandertoell2spreadvUpperBound,eq:expandertoell2spreadAuLowerBound} now yield
	\begin{flalign*}
	&\norm{Ax}_2 = \norm{A(u+v)}_2 \geq \norm{Au}_2 - \norm{A}_2 \norm{v}_2 \geq \frac{\sqrt{t}}{\sqrt{2}} \cdot \beta^{\frac{\ell^* + 1}2} - 2 \norm{A}_2 \cdot \beta^{\frac{\ell^*+ 2}2} \\
	&= \beta^{\frac{\ell^* + 1}2} \inparen{\frac{\sqrt{t}}{\sqrt{2}}  - 2 \norm{A}_2 \sqrt \beta} \geq \beta^{\frac{\ell^*+1}2} \cdot \frac{\sqrt{t}}{2 \sqrt{2}} \geq \beta^{\frac{M + 1}2}\cdot \frac{\sqrt{t}}{2 \sqrt{2}}\enspace,
	\end{flalign*}
	and the lemma follows from the definitions of $\beta$ and $M$, and from our assumptions $\mu \le \frac 29$ and $\gamma n \ge \inparen{\frac 1\mu}^c$.
	\end{proof}
	
%%%%%%%%%%%%%%%%%%%%%%%%%%%%%%%%%%%%%%%%%%%%%%%%%%%%%%%%%%%%%%%%%%%%%%%%%%%%%%%%
%%%%%%%%%%%%%%%%%%%%%%%%%%%%%%%%%%%%%%%%%%%%%%%%%%%%%%%%%%%%%%%%%%%%%%%%%%%%%%%%
%%%%%%%%%%%%%%%%%%%%%%%%%%%%%%%%%%%%%%%%%%%%%%%%%%%%%%%%%%%%%%%%%%%%%%%%%%%%%%%%
\section{General relations between spread, distortion and restricted isometry properties}
We prove \cref{prop:pspreadtoqspread,prop:riptospread,prop:companddist} in \cref{sec:pspreadtoqspread,sec:riptospread,sec:companddist}, respectively.

%%%%%%%%%%%%%%%%%%%%%%%%%%%%%%%%%%%%%%%%%%%%%%%%%%%%%%%%%%%%%%%%%%%%%%%%%%%%%%%%
%%%%%%%%%%%%%%%%%%%%%%%%%%%%%%%%%%%%%%%%%%%%%%%%%%%%%%%%%%%%%%%%%%%%%%%%%%%%%%%%
\subsection{$\ell_p$-spread implies $\ell_q$-spread for $q < p$}
\label{sec:pspreadtoqspread}
We prove \cref{prop:pspreadtoqspread}, which we restate below.
\pspreadtoqspread*

\begin{proof}
Let $x \in X$, and let $S \subseteq [n]$ with $\abs{S} = k$ denote the $k$ smallest coordinates (in absolute value) in $x$. Let $X_{\bar{S}}$ denote the subspace obtained by projecting $X$ to the coordinates \emph{not in} $S$. We first show the following claim.

\begin{claim}
\label{claim:projectionspread}
$X_{\bar{S}}$ is $(k,\eps)$-$\ell_p$-spread.
\end{claim}
\begin{proof}[Proof of \cref{claim:projectionspread}]
Let $y_{\bar{S}} \in X_{\bar{S}}$ be obtained by projecting $y \in X$ to the coordinates not in $S$, and let $S' \subseteq \bar{S}$ be a subset of size $k$. The vector $z$ with $\supp(z) = S'$ that minimizes $\norm{y_S - z}_p$ is obtained by setting $z = y_{S'}$. Letting $w = y_{S \cup S'}$, we see that $y_S - z = y - w$, as both are equal to $y_{\overline{S \cup S'}}$. As $w$ is $2k$-sparse,
\begin{equation*}
\norm{y_S - z}_p = \norm{y - w}_p \geq \eps \norm{y}_p \geq \eps \norm{y_{\bar{S}}}_p\enspace,
\end{equation*}
where the first inequality uses that $y \in X$ and $X$ is $(2k, \eps)$-$\ell_p$-spread. As $S'$ was arbitrary, this proves that $X_S$ is $(k,\eps)$-$\ell_p$-spread.
\end{proof}

We now use \cref{claim:projectionspread} to finish the proof. Let $y = x_S$. We have that
\begin{flalign*}
&\norm{x_{\bar{S}}}_p = \norm{x - y}_p \geq \eps \norm{x}_p \geq \frac{\eps}{n^{\frac{1}{q} - \frac{1}{p}}} \norm{x}_q \enspace,
\end{flalign*}
by H\"{o}lder's inequality (\cref{lem:holder}) and using that $X$ is $(2k,\eps)$-$\ell_p$-spread. We also have that
\begin{equation*}
\norm{x_{\bar{S}}}_p \leq \norm{x_{\bar{S}}}_q \cdot \frac{\Delta_{q,p}(X_{\bar{S}})}{n^{\frac{1}{q} - \frac{1}{p}}} \enspace.
\end{equation*}
As $X_{\bar{S}}$ is $(k,\eps)$-$\ell_p$-spread, we have $\Delta_{q,p}(X_{\bar{S}}) \leq \frac{1}{\eps} \cdot \left(\frac{n}{k}\right)^{\frac{1}{q}}$ by \cref{prop:companddist}. Hence,
\begin{equation*}
\norm{x_{\bar{S}}}_q \cdot \frac{1}{\eps} \cdot \left(\frac{n}{k}\right)^{\frac{1}{q}} \geq \eps \norm{x}_q \enspace.
\end{equation*}
As $\norm{x_{\bar{S}}}_q = \norm{x - y}_q$ and $S$ was arbitrary, we are done.
\end{proof}

%%%%%%%%%%%%%%%%%%%%%%%%%%%%%%%%%%%%%%%%%%%%%%%%%%%%%%%%%%%%%%%%%%%%%%%%%%%%%%%%
%%%%%%%%%%%%%%%%%%%%%%%%%%%%%%%%%%%%%%%%%%%%%%%%%%%%%%%%%%%%%%%%%%%%%%%%%%%%%%%%
\subsection{$\ell_p$-RIP implies $\ell_p$-spread}
\label{sec:riptospread}
We prove \cref{prop:riptospread}, which we restate below. Our proof is based on an argument of \cite{KT07}.
\riptospread*

\begin{proof}
Let $A$ be a $(k,\eps)$-$\ell_p$-RIP matrix. Without loss of generality, we will assume that the ``normalization factor'' $K$ is $1$. We first upper bound $\norm{A}_p$. Let $x \in \R^n$ be arbitrary. Partition $[n]$ arbitrarily into $b = \ceil{\frac{n}{k}}$ sets $S_1, \dots S_b$, each of size $\leq k$, and for $i \in [b]$, let $z_i := x_{S_i}$. That is, $z_i$ is $x$ restricted to the coordinates in $S_i$. By definition, the $z_i$'s are each $k$-sparse and satisfy $\sum_{i = 1}^b z_i = x$. We thus have
\begin{flalign*}
\norm{Ax}_p \leq \sum_{i = 1}^b \norm{A z_i}_p \leq \sum_{i = 1}^b (1 + \eps) \norm{z_i}_p \leq (1 + \eps) b^{1 - \frac{1}{p}} (\sum_{i = 1}^b \norm{z_i}^p)^{1/p} = (1 + \eps) b^{1 - \frac{1}{p}} \norm{x}_p\enspace,
\end{flalign*}
where the second inequality uses that $A$ is $(k,\eps)$-$\ell_p$-RIP and the third inequality is by \cref{lem:holder}. We thus have that $\norm{A}_p \leq (1 + \eps) b^{1 - \frac{1}{p}}$. 

Now, let $x \in \R^n$ with $\norm{x}_p = 1$ be a $(k,\eps')$-compressible vector, where $\eps' < \frac{1 - \eps}{2 + \eps(1 + b^{1 - \frac{1}{p}})}$. We show that $\norm{Ax}_p > 0$, so that $x \notin \ker(A)$, and thus $\ker(A)$ is $(k,\eps')$-spread. Let $y \in \R^n$ be a $k$-sparse vector such that $\norm{x - y}_p \leq \eps'$. Note that $\norm{y}_p \geq \norm{x}_p - \norm{x - y}_p \geq 1 - \eps'$. We have that
\begin{flalign*}
\norm{Ax}_p \geq \norm{Ay}_p - \norm{A}_p \norm{x - y}_p \geq (1 - \eps)(1 - \eps') - (1 + \eps) b^{1 - \frac{1}{p}} \eps' > 0 \enspace,
\end{flalign*}
by choice of $\eps'$.

Finally, we note that $b^{1 - \frac{1}{p}} < (\frac{2n}{k})^{1 - \frac{1}{p}}$, and so we can take $\eps' = \frac{1 - \eps}{2 + \eps(1 + (\frac{2n}{k})^{1 - \frac{1}{p}})}$. 
\end{proof}

%%%%%%%%%%%%%%%%%%%%%%%%%%%%%%%%%%%%%%%%%%%%%%%%%%%%%%%%%%%%%%%%%%%%%%%%%%%%%%%%
%%%%%%%%%%%%%%%%%%%%%%%%%%%%%%%%%%%%%%%%%%%%%%%%%%%%%%%%%%%%%%%%%%%%%%%%%%%%%%%%
\subsection{Equivalence between $\ell_p$-spread and $(\ell_q,\ell_p)$-distortion}
\label{sec:companddist}

We prove \cref{prop:companddist}, restated below. Our proof generalizes an argument in~\cite{GLR10}.
\companddist*
	\begin{proof}
		Clearly, it suffices to prove both statements under the assumption that $\norm{x}_p = 1$. Let $a = \frac{1}{q} - \frac{1}{p}$.
		
		\begin{enumerate}
			\item Let $y\in \R^n$ be $k$-sparse with $\smnorm{x-y}_p\le \eps$. Clearly, we may take $y$ so that that $y$ and $x-y$ have disjoint supports. In particular, this implies that $\smnorm{y}_p \le \norm{x}_p = 1$. Thus,
			$$\frac{1}{\Delta_{q,p}(x)} = \norm{x}_q\cdot n^{-a} \le  \inparen{\smnorm {y}_q + \smnorm{x-y}_q} \cdot n^{-a}\le \inparen{\smnorm{y}_p\cdot k^a + \smnorm{x-y}_p\cdot n^a}\cdot n^{-a}\le \inparen{\frac{k}{n}}^a  + \eps,$$
			where the penultimate inequality follows from \cref{lem:holder}.
			
			\item Let $S\subseteq [n]$ be a set consisting of the $k$ largest coordinates of $x$ in absolute value (with arbitrary tie-breaking). Namely, $\abs{S} = k$ and $\abs{x_j}\ge \abs{x_i}$ for all $j\in S$, $i\in [n]\setminus S$. Define $y \in \R^n$ by
			\begin{equation*}y_i = \begin{cases}
				x_i &\text{if }i\in S\\
				0 &\text{if }i\in [n]\setminus S \enspace.
			\end{cases}
			\end{equation*}
			
			It now suffices to show that $\smnorm{x-y}_p \le \inparen{\tfrac{n}{k}}^{1/q} \tfrac{1}{\Delta_{p,q}(v)}$. Let $\sigma = \min\{\abs{x_i} \mid i\in S\}$. Note that
			\begin{equation*}
			\frac{1}{\Delta_{p,q}(x)} = \frac{\norm{x}_q}{n^a} \ge \frac{\smnorm{y}_q}{n^a} \ge \frac{k^{\frac{1}{q}} \sigma}{n^a} \enspace, \quad 
			\text{so that} \quad 
			\sigma \le \frac{n^a}{\Delta_{p,q}(x)\cdot k^{\frac{1}{q}}} \enspace.
			\end{equation*}
			The claim follows since
			\begin{equation*} \norm{x-y}_p \le \infnorm{x-y}\cdot n^{\frac{1}{p}} \le \sigma\cdot n^{\frac{1}{p}} \le \inparen{\frac{n}{k}}^{\frac{1}{q}}\cdot \frac{1}{\Delta_{p,q}(x)} \enspace . 
			\end{equation*}
		\end{enumerate}		
		Finally, suppose that $X$ is $(k,\eps)$-$\ell_p$-spread. Let $x \in X$ be any vector, and note that $x$ is $(k,\eps)$-$\ell_p$-spread. We also have that $x$ is $\inparen{k,\frac{\inparen{\frac {n}{k}}^{\frac{1}{q}}}{\Delta_{q,p}(x)}}$-$\ell_p$-compressible, which implies that $\eps \leq \frac{\inparen{\frac {n}{k}}^{\frac{1}{q}}}{\Delta_{q,p}(x)}$. Rearranging and taking the $\sup$ over $x \in X$, we conclude that $\Delta_{q,p}(X) \leq \frac{1}{\eps} \left(\frac{n}{k}\right)^{\frac{1}{q}}$.
	\end{proof}
%%%%%%%%%%%%%%%%%%%%%%%%%%%%%%%%%%%%%%%%%%%%%%%%%%%%%%%%%%%%%%%%%%%%%%%%%%%%%%%%
%%%%%%%%%%%%%%%%%%%%%%%%%%%%%%%%%%%%%%%%%%%%%%%%%%%%%%%%%%%%%%%%%%%%%%%%%%%%%%%%
%%%%%%%%%%%%%%%%%%%%%%%%%%%%%%%%%%%%%%%%%%%%%%%%%%%%%%%%%%%%%%%%%%%%%%%%%%%%%%%%
\section*{Acknowledgements}
We thank Sidhanth Mohanty for helpful discussions about the works of \cite{BritoDH18, Bordenave19, BordenaveC19, MohantyOP20a, MohantyOP20b, ODonnellW20}, and Yuval Peled for helpful discussions about convergence theorems for graph spectra.

We thank Amir Shpilka for bringing the work of \cite{Karnin11} to our attention, and Ioana Dumitriu for bringing the works of \cite{BritoDH18, Zhu20} to our attention.

%%%%%%%%%%%%%%%%%%%%%%%%%%%%%%%%%%%%%%%%%%%%%%%%%%%%%%%%%%%%%%%%%%%%%%%%%%%%%%%%
%%%%%%%%%%%%%%%%%%%%%%%%%%%%%%%%%%%%%%%%%%%%%%%%%%%%%%%%%%%%%%%%%%%%%%%%%%%%%%%%
%%%%%%%%%%%%%%%%%%%%%%%%%%%%%%%%%%%%%%%%%%%%%%%%%%%%%%%%%%%%%%%%%%%%%%%%%%%%%%%%
	
	%%%%%%%%%%%%%%%%%%%%%%%%%%%%%%%%%%%%%%%%%%%%%%%%%%%%%%%%%%%%%%%%%%%%%%%%%%%%%%%%
	%%%%%%%%%%%%%%%%%%%%%%%%%%%%%%%%%%%%%%%%%%%%%%%%%%%%%%%%%%%%%%%%%%%%%%%%%%%%%%%%
	%%%%%%%%%%%%%%%%%%%%%%%%%%%%%%%%%%%%%%%%%%%%%%%%%%%%%%%%%%%%%%%%%%%%%%%%%%%%%%%%
	\bibliographystyle{alpha}
\bibliography{SparseSections.bbl}
	%%%%%%%%%%%%%%%%%%%%%%%%%%%%%%%%%%%%%%%%%%%%%%%%%%%%%%%%%%%%%%%%%%%%%%%%%%%%%%%%
	%%%%%%%%%%%%%%%%%%%%%%%%%%%%%%%%%%%%%%%%%%%%%%%%%%%%%%%%%%%%%%%%%%%%%%%%%%%%%%%%
	%%%%%%%%%%%%%%%%%%%%%%%%%%%%%%%%%%%%%%%%%%%%%%%%%%%%%%%%%%%%%%%%%%%%%%%%%%%%%%%%
\end{document}